\documentclass[hidelinks]{eptcs}
\usepackage{breakurl}
\usepackage{underscore}
\usepackage{latexsym}
\usepackage{wrapfig}
\usepackage{marvosym}
\usepackage{color}
\usepackage{graphicx}
\usepackage{amssymb}
\usepackage{graphicx}

\DeclareSymbolFont{frenchscript}{OMS}{ztmcm}{m}{n}
\DeclareMathSymbol{\A}{\mathord}{frenchscript}{65}    % set of ABC agent identifiers
\DeclareMathSymbol{\B}{\mathord}{frenchscript}{66}    % set of ABC broadcast names
\DeclareMathSymbol{\CH}{\mathord}{frenchscript}{67}   % set of CCS handshake communications
\newcommand{\Ch}{{\CH_h}}                             % set of CCS handshake communication names
\DeclareMathSymbol{\C}{\mathord}{frenchscript}{67}    % (static) component set
\DeclareMathSymbol{\Lab}{\mathord}{frenchscript}{76}  % set of labels in transition system
\DeclareMathSymbol{\Ge}{\mathord}{frenchscript}{71}   % liveness property
\DeclareMathSymbol{\Pow}{\mathord}{frenchscript}{80}  % powerset
\DeclareMathSymbol{\Sig}{\mathord}{frenchscript}{83}  % set of signals
\DeclareMathSymbol{\Tk}{\mathord}{frenchscript}{84}   % task set
\DeclareMathAlphabet{\mathbbm}{U}{bbm}{m}{n}          % blackboard bold
\newcommand{\IT}{\mathbbm{T}}                         % labelled transition system
\newcommand{\cT}{\mathbbm{P}}                         % closed terms, or processes
\newcommand{\IN}{\mathbbm{N}}                         % natural numbers
\newtheorem{defi}{Definition}
\newtheorem{theo}{Theorem}
\newtheorem{prop}{Proposition}
\newtheorem{lemm}{Lemma}
\newtheorem{coro}{Corollary}
\newtheorem{obs}{Observation}
\newtheorem{exam}{Example}

\newenvironment{definition}[1]{\begin{defi} \rm \label{df:#1} }{\end{defi}}
\newenvironment{definitionA}[2]{\begin{defi}[#1] \rm \label{df:#2} }{\end{defi}}
\newenvironment{theorem}[1]{\begin{theo} \rm \label{thm:#1} }{\end{theo}}
\newenvironment{proposition}[1]{\begin{prop} \rm \label{pr:#1} }{\end{prop}}
\newenvironment{lemma}[1]{\begin{lemm} \rm \label{lem:#1} }{\end{lemm}}
\newenvironment{corollary}[1]{\begin{coro} \rm \label{cor:#1} }{\end{coro}}
\newenvironment{observation}[1]{\begin{obs} \rm \label{obs:#1} }{\end{obs}}
\newenvironment{example}[1]{\begin{exam} \rm \label{ex:#1} }{\hfill\P\end{exam}}
\newenvironment{exampleNoEnd}[1]{\begin{exam} \rm \label{ex:#1} }{\end{exam}}
\newenvironment{exampleA}[2]{\begin{exam}[#1] \rm \label{ex:#2} }{\hfill\P\end{exam}}
\newenvironment{proof}{\begin{trivlist} \item[\hspace{\labelsep}\bf Proof:]}{\hfill $\Box$\end{trivlist}}
\newenvironment{proofNoBox}{\begin{trivlist} \item[\hspace{\labelsep}\bf Proof:]}{\end{trivlist}}
\newcommand{\Sec}[1]{Section~\ref{sec:#1}}

\newcommand{\df}[1]{Definition~\ref{df:#1}}
\newcommand{\thm}[1]{Theorem~\ref{thm:#1}}
\newcommand{\pr}[1]{Proposition~\ref{pr:#1}}
\newcommand{\lem}[1]{Lemma~\ref{lem:#1}}
\newcommand{\cor}[1]{Corollary~\ref{cor:#1}}
\newcommand{\Obs}[1]{Observation~\ref{obs:#1}}  
\newcommand{\ex}[1]{Example~\ref{ex:#1}}
\newcommand{\tab}[1]{Table~\ref{tab:#1}}
\newcommand{\fig}[1]{Figure~\ref{#1}}

\newenvironment{itemise}%
  {\begin{itemize}%
    \setlength{\itemsep}{0pt}%
    \setlength{\parskip}{0pt}}%
  {\end{itemize}}
\makeatletter
\def\comesfrom{\@transition\leftarrowfill}
\def\goesto{\@transition\rightarrowfill}
\def\ngoesto{\@transition\nrightarrowfill}
\def\Goesto{\@transition\Rightarrowfill}
\def\nGoesto{\@transition\nRightarrowfill}
\def\xmapsto{\@transition\mapstofill}
\def\nxmapsto{\@transition\nmapstofill}
\def\@transition#1{\@@transition{#1}}
\newbox\@transbox
\newbox\@arrowbox
\newbox\@downbox
\def\@@transition#1#2%
   {\setbox\@transbox\hbox
      {\vrule height 1.5ex depth .8ex width 0ex\hskip0.25em$\scriptstyle#2$\hskip0.25em}
   \ifdim\wd\@transbox<1.5em
      \setbox\@transbox\hbox to 1.5em{\hfil\box\@transbox\hfil}\fi
   \setbox\@arrowbox\hbox to \wd\@transbox{#1}
   \ht\@arrowbox\z@\dp\@arrowbox\z@
   \setbox\@transbox\hbox{$\mathop{\box\@arrowbox}\limits^{\box\@transbox}$}
   \dp\@transbox\z@\ht\@transbox 10pt
   \mathrel{\box\@transbox}}
\def\nrightarrowfill{$\m@th\mathord-\mkern-6mu%
  \cleaders\hbox{$\mkern-2mu\mathord-\mkern-2mu$}\hfill
  \mkern-6mu\mathord\not\mkern-2mu\mathord\rightarrow$}
\def\Rightarrowfill{$\m@th\mathord=\mkern-6mu%
  \cleaders\hbox{$\mkern-2mu\mathord=\mkern-2mu$}\hfill
  \mkern-6mu\mathord\Rightarrow$}
\def\nRightarrowfill{$\m@th\mathord=\mkern-6mu%
  \cleaders\hbox{$\mkern-2mu\mathord=\mkern-2mu$}\hfill
  \mkern-6mu\mathord\not\mathord\Rightarrow$}
\def\mapstofill{$\m@th\mathord\mapstochar\mathord-\mkern-6mu%
  \cleaders\hbox{$\mkern-2mu\mathord-\mkern-2mu$}\hfill
  \mkern-6mu\mathord\rightarrow$}
\def\nmapstofill{$\m@th\mathord\mapstochar\mathord-\mkern-6mu%
  \cleaders\hbox{$\mkern-2mu\mathord-\mkern-2mu$}\hfill
  \mkern-6mu\mathord\not\mkern-2mu\mathord\rightarrow$}
\makeatother %%%%% end of arrows definition
\newcommand{\ar}[1]{\mathrel{\goesto{#1}}}            % arrow
\newcommand{\nar}[1]{{\ngoesto{#1\;}}}                % negated arrow
\newcommand{\goto}[1]{\stackrel{#1}{\longrightarrow}} % transition
\newcommand{\Tr}{\textit{Tr}}                         % transitions
\newcommand{\source}{\textit{source}}                 % source of transition
\newcommand{\target}{\textit{target}}                 % target of transition
\newcommand{\comp}{\textit{comp\/}}                   % components involved in transition
\newcommand{\npc}{\textit{npc}}                       % components causing transition
\newcommand{\afc}{\textit{afc}}                       % components affected by transition
\newcommand{\defis}{\stackrel{{\it def}}{=}}          % recursive definition
\newcommand{\NC}{{\it NC}}                            % necessary dynamic componenets
\newcommand{\Left}{\mathrm{\scriptscriptstyle L}}     % left component indicator
\newcommand{\R}{\mathrm{\scriptscriptstyle R}}        % right component indicator
\newcommand{\Rec}{\mathit{Rec}}                       % receptive actions
\newcommand{\J}{{\rm J}}                              % J (as in J-fairness)
\newcommand{\aconc}{\mathrel{\mbox{$\smile\hspace{-.95ex}\raisebox{2.5pt}{$\scriptscriptstyle\bullet$}$}}}
\newcommand{\naconc}{\mathrel{\mbox{$\,\not\!\smile\hspace{-.95ex}\raisebox{2.5pt}{$\scriptscriptstyle\bullet$}$}}}
\newcommand{\conc}{\smile}                            % concurrency relation
\newcommand{\nconc}{\,\not\!\smile}                   % not concurrent
\newcommand{\dcup}{\mathbin{\plat{$\stackrel{\mbox{\huge .}}{\cup}$}}}   % disjoint union
\newcommand{\plat}[1]{\raisebox{0pt}[0pt][0pt]{#1}}   % no vertical space
\newcommand{\weg}[1]{}                                % omitted material

\newcommand{\signals}{\ensuremath{\mathrel{\hat{}\!}}}% syntactic opertor ``signals''
\newcommand{\sigar}[1]{^{\curvearrowright #1}}           % semantic behaviour ``signalling''
\newcommand{\sig}[1]{^{\rightarrow #1}}                 % signalling synchron
\newcommand{\RL}{L}                                   % restricted set of actions
\newcommand{\shar}[1]{\mathord{\stackrel{#1}{\rightarrow}}} % active synchron
\newcommand{\mylabel}[1]{\hypertarget{lab:#1}{\ \mbox{{\scriptsize\sc (#1)}}}}
\newcommand{\myref}[1]{\hyperlink{lab:#1}{\scriptsize\sc (#1)}}

\renewcommand{\chi}{t}                                % typical derivation of transition
\renewcommand{\zeta}{u}                               % typical derivation of transition
\newcommand{\E}{P}                                    % typical process expression
\newcommand{\F}{Q}                                    % typical process expression
\newcommand{\AI}{A}                                   % typical agent identifier
\newcommand{\We}{We\ }
\newcommand{\we}{we\ }
\newcommand{\acGH}{\aconc_{\mbox{\tiny \cite{GH15a}}}}    % concurrency relation from [GH15]
\newcommand{\signal}{indicator}
\begin{document}

\def\publicationstatus{An extended abstract of this paper appeared in the proceedings of FoSSaCS'19,
  doi:\href{http://doi.org/10.1007/978-3-030-17127-8_29}{10.1007/978-3-030-17127-8_29}.}
\def\titlerunning{Justness}
\def\authorrunning{Rob van Glabbeek}
\title{\titlerunning\\\Large A Completeness Criterion for Capturing Liveness Properties} 
\author{\authorrunning
\institute{Data61, CSIRO, Sydney, Australia}
\institute{School of Computer Science and Engineering,
University of New South Wales, Sydney, Australia}
}
\maketitle

\begin{abstract}
This paper poses that transition systems constitute a good model of distributed systems only in combination
with a criterion telling which paths model complete runs of the represented systems.
Among such criteria, progress is too weak to capture relevant liveness properties, and fairness is often
too strong; for typical applications \we advocate the intermediate criterion of justness.
Previously, \we proposed a definition of justness in terms of an asymmetric concurrency relation
between transitions. Here \we define such a concurrency relation for the transition systems
associated to the process algebra CCS as well as its extensions with broadcast communication and signals,
thereby making these process algebras suitable for capturing liveness properties requiring justness.
\end{abstract}

\section{Introduction}

Transition systems are a common model for distributed systems. They consist of sets of states, also
called \emph{processes}, and transitions---each transition going from a source state to a target
state. A given distributed system $\mathcal{D}$ corresponds to a state $P$ in a transition system
$\IT$---the initial state of $\mathcal{D}$.  The other states of $\mathcal{D}$ are the processes in $\IT$
that are reachable from $P$ by following the transitions.
A run of $\mathcal{D}$ corresponds with a \emph{path} in $\IT$: a finite or infinite alternating
sequence of states and transitions, starting with $P$, such that each transition goes from the state before to the state after it.
Whereas each finite path in $\IT$ starting from $P$ models a \emph{partial run} of $\mathcal{D}$,
i.e., an initial segment of a (complete) run, typically not each path models a run. Therefore a transition system constitutes a good model of distributed systems
only in combination with what \we here call a \emph{completeness criterion}: a selection of a
subset of all paths as \emph{complete paths}, modelling runs of the represented system.

A \emph{liveness property} says that ``something [good] must happen'' eventually \cite{Lam77}.
Such a property holds for a distributed system if the [good] thing happens in each of its possible runs.
One of the ways to formalise this in terms of transition systems is to postulate a set of good
states $\Ge$, and say that the liveness property $\Ge$ holds for the process $P$ if all complete
paths starting in $P$ pass through a state of $\Ge$ \cite{GH19}. Without a completeness criterion
the concept of a liveness property appears to be meaningless.

\begin{example}{Cataline}
The transition system on the right models Cataline eating
{\makeatletter
\let\par\@@par
\par\parshape0
\everypar{}\begin{wrapfigure}[1]{r}{0.308\textwidth}
 \vspace{-5ex}
 \input{Cataline}
  \centerline{\box\graph}
 \end{wrapfigure}
\noindent
a croissant in Paris.
It abstracts from all activity in the world except the eating of that croissant, and thus has two
states only---the states of the world before and after this event---and one transition $t$.
\We depict states by circles and transitions by arrows between them. An initial state is indicated by
a short arrow without a source state. A possible liveness property says that the croissant will be
eaten. It corresponds with the set of states $\Ge$ consisting of state $2$ only.
The states of $\Ge$ are indicated by shading.
\par}
The depicted transition system has three paths starting with state 1: $1$, $1\,t$ and $1\,t\,2$.
The path $1\,t\,2$ models the run in which Cataline finishes the croissant.
The path $1$ models a run in which Cataline never starts eating the croissant, and the path $1\,t$
models a run in which she starts eating it, but never finishes.
The liveness property $\Ge$ holds only when using a completeness criterion that disqualifies the paths
$1$ and $1\,t$, saying that they do not model actual runs of the system, thus leaving $1\,t\,2$ as the sole complete path.%
\end{example}
The transitions of transition systems can be understood to model atomic actions that can be
performed by the represented systems. Although \we allow these actions to be instantaneous or durational, in
the remainder of this paper \we adopt the assumption that ``atomic actions always terminate'' \cite{OL82}.
This is a partial completeness criterion. It rules out the path $1\,t$ in \ex{Cataline}.
\We build in this assumption in the definition of a path by henceforth requiring that finite paths
should end with a state.

\paragraph{Progress}

The most widely employed completeness criterion is \emph{progress}.\footnote{Misra
  \cite{Mis88,Mis01} calls this the `minimal progress assumption'.  In \cite{Mis01} he uses
  `progress' as a synonym for `liveness'.  In session types, `progress' and `global progress' are
  used as names of particular liveness properties \cite{CDPY13}; this use has no relation with ours.}
In the context of \emph{closed systems}, having no run-time interactions with the environment,
it is the assumption that a run will never get stuck in a state with outgoing transitions.
This rules out the path $1$ in \ex{Cataline}, as $t$ is outgoing. When adopting progress as completeness criterion,
the liveness property $\Ge$ holds for the system modelled in \ex{Cataline}.

Progress is assumed in almost all work on process algebra that deals with liveness properties, mostly implicitly. 
{Milner} makes an explicit progress assumption for the process algebra CCS in \cite{Mi80}.
% As evidence \we refer to Section 1.5 "Unobservable actions" on page 15
% and 16 of [29]: "But what does the [tau]-transition mean, in our more
% mechanistic interpretation? In means that R in state r1 ... *may* at
% any time move silently to r1', and that if a b-experiment is never
% attempted it *will* do so." (Here b is the other outgoing transition of r1.)
A progress assumption is built into the temporal logics LTL \cite{Pn77}, CTL \cite{EC82} and CTL* \cite{EH86},
namely by disallowing states without outgoing transitions and evaluating temporal formulas
by quantifying over infinite paths only.%
\footnote{Exceptionally, states without outgoing transitions are allowed, and then
    quantification is over all \emph{maximal} paths, i.e., paths that are infinite or end in a state
    without outgoing transitions \cite{DV95}.}
In \cite{KdR83} the `multiprogramming axiom' is a
progress assumption, whereas in \cite{AFK88} progress is assumed as a `fundamental liveness property'.
% In many other papers on fairness a progress assumption is made implicitly~\cite{AO83}.

\begin{definitionA}{\cite{GH19}}{stronger}
Completeness criterion F is \emph{stronger} than completeness criterion H
iff {\rm F} rules out (as incomplete) at least all paths that are ruled out by H.
\end{definitionA}

As \we argue in \cite{TR13,GH15a,GH19}, a progress assumption as above is too strong in the context of
reactive systems. There, a transition typically represents an interaction between the distributed
system being modelled and its environment. In many cases a transition can occur only if both
the modelled system \emph{and} the environment are ready to engage in it. \We therefore distinguish
\emph{blocking} and \emph{non-blocking} transitions. A transition is non-blocking if
the environment cannot or will not block it, so that its execution is entirely under the
control of the system under consideration. A blocking transition on the other hand may fail to occur
because the environment is not ready for it. The same was done earlier in the setting of Petri nets
\cite{Rei13}, where blocking and non-blocking transitions are called \emph{cold} and \emph{hot}, respectively.

In \cite{TR13,GH15a,GH19} \we work with transition systems that are equipped with a partitioning
of the transitions into blocking and non-blocking ones, and reformulate the progress assumption as follows:
\begin{quote}\it
a (transition) 
system in a state that admits a non-blocking transition will eventually progress, i.e., perform a transition.
\end{quote}
In other words, a run will never get stuck in a state with outgoing non-blocking transitions.
In \ex{Cataline}, when adopting progress as our completeness criterion, \we assume that Cataline
actually wants to eat the croissant, and does not willingly remain in State 1 forever.
When that assumption is unwarranted, one would model her behaviour by a transition system different
from that of \ex{Cataline}. However, she may still be stuck in State 1 by lack of any croissant to
eat. If \we want to model the capability of the environment to withhold a croissant, \we classify $t$
as a blocking transition, and the liveness property $\Ge$ does not hold. If \we abstract from a
possible shortage of croissants, $t$ is deemed a non-blocking transition, and, when assuming progress,
$\Ge$ holds.

As an alternative approach to a dogmatic division of transitions in a transition system, one could shift
the status of transitions to the progress property, and speak of $B$-progress when $B$ is
the set of blocking transitions. In that approach, $\Ge$ holds for State 1 of \ex{Cataline} under
the assumption of $B$-progress when $t\notin B$, but not when $t\in B$.

To properly capture reactive systems, \we work with \emph{labelled} transition systems,
where each transition is labelled with the \emph{action} that occurs (or fact that is revealed)
when taking this transition.
The labelling is typically used to describe how the transition synchronises with, and thus is
dependent on, the environment. Whether a transition is blocking is then completely determined by its
label. Hence \we work with sets $B$ of blocking actions and regard a transition as
blocking iff it is labelled by an action in~$B$.

\paragraph{Justness}

Justness is a completeness criterion proposed in \cite{TR13,GH15a,GH19}. It strengthens
progress. It can be argued that once one adopts progress it makes sense to go a step further and
adopt even justness.

\begin{exampleNoEnd}{Cataline and Alice}
The transition system on the right models Alice making an unending
{\makeatletter
\let\par\@@par
\par\parshape0
\everypar{}\begin{wrapfigure}[2]{r}{0.197\textwidth}
 \vspace{-5ex}
 \input{Alice}
  \centerline{\box\graph}
 \end{wrapfigure}
\noindent
sequence of phone calls in London.
There is no interaction of any kind between Alice and Cataline.
Yet, \we may choose to abstracts from all activity in the world except the eating of the croissant by
Cataline, and the making of calls by Alice.
\par}
{\makeatletter
\let\par\@@par
\par\parshape0
\everypar{}
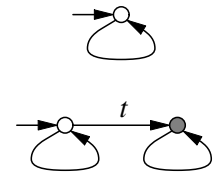
\begin{wrapfigure}[2]{r}{0.197\textwidth}
 \vspace{-6ex}
 \input{AliceCataline}
  \centerline{\box\graph}
 \end{wrapfigure}
\noindent
This yields the combined transition system on the bottom right.
Even when taking the transition $t$ to be non-blocking,
progress is not a strong enough completeness criterion to ensure that Cataline will ever eat the croissant,
for the infinite path that loops in the first state is complete.
Nevertheless, as nothing stops Cataline from making progress, in reality $t$ will occur. \cite{GH19} \P
\par}
\end{exampleNoEnd}
This example is not a contrived corner case, but a rather typical illustration of an issue that is
central to the study of distributed systems. Other illustrations of this phenomena occur in
\cite[Section 9.1]{TR13}, \cite[Section 10]{GH15b}, \cite[Section 1.4]{vG15}, \cite{vG16} and \cite[Section 4]{EPTCS255.2}.
The criterion of justness aims to ensure the liveness property occurring in these examples.
In \cite{GH19} it is formulated as follows:
\begin{quote}\it
  Once a non-blocking transition is enabled that stems from a set of parallel components,
  one (or more) of these components will eventually partake in a transition.
\end{quote}
In \ex{Cataline and Alice}, $t$ is a non-blocking transition enabled in the initial state.
It stems from the single parallel component Cataline of the distributed system under consideration.
Justness therefore requires that Cataline must partake in a transition. This can only be $t$,
as all other transitions involve component Alice only. Hence justness says that $t$ must occur.
The infinite path starting in the initial state and not containing $t$ is ruled out as
unjust, and thereby incomplete.

Unlike progress, the concept of justness as formulated above is in need of some formalisation,
i.e., to formally define a component, to make precise for concrete transition systems what it means for a
transition to stem from a set of components, and to define when a component partakes in a transition.

A formalisation of justness for the labelled transition system generated by the process algebra AWN, the
\emph{Algebra for Wireless Networks} \cite{FGHMPT12a}, was provided in \cite{TR13}. In the same vein,
\cite{GH15a} offered a formalisation for the labelled transition systems generated by CCS,
the \emph{Calculus of Communicating Systems} \cite{Mi80}, and its extension ABC, the 
\emph{Algebra of Broadcast Communication} \cite{GH15a}, a variant of CBS, the
\emph{Calculus of Broadcasting Systems} \cite{CBS91}. The same was done for CCS extended with
\emph{signals} in \cite{EPTCS255.2}. The formalisations of \cite{GH15a,EPTCS255.2} coinductively define
\emph{$B$-justness}, where $B$ ranges over sets of actions that are deemed to be blocking, as a
family of predicates on paths, and proceed by a case distinction on the operators in the language.\linebreak[3]
Although these definitions \emph{do} capture the concept of justness formulated above,
it is not easy to see why.

A more syntax-independent, and perhaps more convincing, formalisation of justness occurred in \cite{GH19}.
There it is defined directly on transition systems that are equipped with a, possibly asymmetric,
concurrency relation between transitions. However, the concurrency relation itself is defined only
for the transition system generated by a fragment of CCS, and the generalisation to full CCS, and
other process algebras, is non-trivial.

It is the purpose of this paper to make the definition of justness from \cite{GH19} available to a
large range of process algebras by defining the concurrency relation for CCS, for ABC, and for the
extension of CCS with signals used in \cite{EPTCS255.2}. \We do this in a precise (\Sec{LTSC}) as well as in an
approximate way (\Sec{components}), and show that both approaches lead to the same concept of justness (\Sec{agree}).
Moreover, in all cases \we establish a closure property on the concurrency relation ensuring that justness is a
meaningful notion. \We show that for all these process algebras justness is \emph{feasible}.
Here feasibility is a requirement on completeness criteria advocated in \cite{AFK88,La00-lb,GH19}.
Finally, \we establish agreement between the formalisation of justness from \cite{GH19} and the present
paper, and the original coinductive ones from \cite{GH15a} and \cite{EPTCS255.2}.

\paragraph{Fairness}

Fairness assumptions are special kinds of completeness criteria. They postulate that if certain
activities \emph{can} happen often enough, they \emph{will} in fact happen.

\begin{example}{fairness}
Suppose Bart stands behind a bar and wants to order a beer. But by lack of any formal queueing protocol
many other customers get their beer before Bart does. This situation can be modelled as a transition
system where in each state in which Bart is not served yet there is an outgoing transition
modelling that Bart gets served, but there are also outgoing transitions modelling that someone else
gets served instead. The essence of fairness is the assumption that Bart will get his beer eventually.
Fairness rules out as unfair, and thereby incomplete, any path in which Bart could have gotten a
beer any time, but never will.
\end{example}
Fairness comes in two flavours: \emph{weak} and \emph{strong} fairness. Weak fairness merely rules
out paths in which some task is enabled in each state, yet never occurs. Strong fairness also rules
out paths in which some task is enabled infinitely often, yet never occurs. Here a \emph{task} is an
appropriate set of transitions, in \ex{fairness} all transitions giving Bart a beer.
In \ex{fairness} the liveness property that Bart will get a beer holds under the assumption of weak
fairness, and thus certainly when assuming strong fairness. It does not hold when merely assuming
justness, let alone when merely assuming progress.

Our survey paper \cite{GH19} proposes a unifying definition of strong and weak fairness,
parametrised by the definition of a task. Many notions of fairness found in
the literature are cast as instances of this definition, differing only in how to define tasks.
The same paper also offers a taxonomy of completeness criteria, ordered by strength (cf.\ \df{stronger}).
This taxonomy contains the criteria progress and justness, as well as all these fairness criteria.
Besides strong and weak fairness we also consider a form of fairness even weaker than weak fairness,
requiring a task to be enabled in each state on a path, as well as ``during each transition''.
We are not aware of any completeness criteria occurring in the literature that is not progress,
justness or one of these forms of fairness---or the weakest possible completeness criterion,
declaring all paths complete, thereby ensuring almost no liveness properties.

In \cite{GH19} we argue that fairness assumptions are by default unwarranted.
In real-world situations akin to \ex{fairness} there is in fact no guarantee that Bart will ever get
a beer. This is in contrast to justness, which by default is warranted. One could argue
that a formalisation of justness is not necessary to arrive at a model of concurrency in which
Cataline will eat her croissant, as fairness is an alternative to justness that accomplishes
the same goal. But here we reject that argument on grounds that fairness tends to rule out as
incomplete more paths than necessary. As argued in \cite{vG16}, this can lead to false guarantees
about the satisfaction of certain liveness properties, e.g., Bart getting a beer in \ex{fairness}.

\paragraph{Reading guide}

In \Sec{justness}, following \cite{GH19}, \we present labelled transition systems with a concurrency relation
satisfying some closure property, and define justness as a predicate on paths in any such transition
system. Again following  \cite{GH19}, \we also propose an optional characterisation of the concurrency
relation in terms of more primitive notions of the \emph{necessary} and \emph{affected components} of
a transition.

In \Sec{feasibility} \we show liberal conditions under which $B$-justness meets the requirement
of feasibility.

In \Sec{fairness} \we recall the unifying definition of fairness from \cite{GH19}, and show
how progress can be cast as a particular fairness property. In spite of this we
continue to see progress as a completeness criterion different from fairness.
% For fairness guarantees the occurrence of a task only when it is enabled relentlessly,
% whereas progress or justness require it to be enabled only once---but for a special kind of task,
% that, once enabled, remains enabled until it occurs.
\We cannot cast justness as a fairness property.

\Sec{PA} recalls the syntax and semantics of CCS, and its extensions ABC and CCSS with broadcast
communication and signals, respectively. It also proposes a simplification of the operational
semantics of CCSS by encoding signal emissions as transitions, and recalls an alternative presentation
of ABC that avoids negative premises in the operational semantics.

In \Sec{LTSC} \we associate labelled transition systems with a concurrency relation to each of the five
process algebras from \Sec{PA}, and show that they satisfy the closure property of \Sec{justness}. 
The concurrency relation is defined in terms of \emph{synchrons}, novel particles out of which
transitions are seen to be composed.
Sections~\ref{sec:justness} and~\ref{sec:LTSC} together constitute a definition of justness valid
for the five process algebras of \Sec{PA}.
For CCS the concurrency relation is symmetric, but for the other four process algebras it is not.
The alternative presentations of CCSS and ABC feature \emph{{\signal} transitions}
that do not model state changes of the represented system.
Instead, an indicator transition reveals a property of its source state $P$, which is also its target state, for
instance that process $P$ is emitting a signal, or that it cannot receive a given broadcast. Indicator transitions
need to be excepted from the justness requirement.

\Sec{components} revisits the component-based characterisation of the concurrency relation
contemplated in \Sec{justness}, and proposes two alternative concepts of system components
associated to a transition, with for each a classification of components as necessary and/or
affected. The \emph{dynamic components} give rise to the exact same concurrency relation as defined
in terms of synchrons in \Sec{LTSC}, whereas the \emph{static components} yield an
underapproximation---a strictly smaller concurrency relation. However, only the static components
satisfy a closure property proposed in \cite{GH19}.

\Sec{computational} provides two computational interpretations of CCS and its extensions,
the default one corresponding to the concurrency relation of \Sec{LTSC}, and thus the dynamic
concurrency relation of \Sec{components}; the other corresponding to the static concurrency relation
of \Sec{components}. \We also provide a natural sublanguage on which the two concurrency relations coincide.

\Sec{agree} shows that the dynamic and static concurrency relations give rise to the very same
concept of justness. Hence, for the study of justness \we may use whichever of these concurrency
relations is the most convenient.
Using this, in \Sec{feasibility2} \we apply the results of \Sec{feasibility} to show that $B$-justness is
feasible for full CCS \cite{Mi90b} and its extensions with broadcast communication or signals.

\Sec{inductive} shows that the concurrency relation of \Sec{LTSC} agrees with the one defined earlier in \cite{GH15a} on pairs of
transitions for which both are defined. Yet, the concurrency relation from \cite{GH15a} was defined
only for transitions with the same source, and hence is not suitable for our formalisation of justness.

In Sections~\ref{sec:coinductive} and \ref{sec:abstract paths} \we establish that the concept of justness based on a
concurrency relation between transitions, as proposed in \cite{GH19} and applied to CCS and its
extension in the present paper, coincides with the original coinductively defined concepts of
justness from \cite{GH15a} and \cite{EPTCS255.2}.

\Sec{conclusion} summarises, and reviews related and future work.

\paragraph{Acknowledgement}
I am grateful to Weiyou Wang, Peter H\"ofner, Victor Dyseryn and Filippo de Bortoli for valuable feedback.
\newpage

\section{Labelled transition systems with concurrency}\label{sec:justness}

\We start with the formal definitions of a labelled transition system, a path, and the completeness
criterion \emph{progress}, which is parametrised by the choice of a collection $B$ of blocking actions.
Then \we define the completeness criterion \emph{justness} on labelled transition systems upgraded
with a concurrency relation.

\begin{definition}{LTS}
A \emph{labelled transition system} (LTS) is a tuple $(S, \Tr, \source,\target,\ell)$ with $S$ and $\Tr$ sets
(of \emph{states} and \emph{transitions}), $\source,\target:\Tr\rightarrow S$ and
$\ell:\Tr\rightarrow \Lab$, for some set of transition labels $\Lab$.
\end{definition}
Here \we work with LTSs labelled over a structured set of labels $(\Lab,Act,\Rec)$,
where $\Rec\subseteq Act \subseteq \Lab$.

In \cite{EPTCS255.2} and in Sections~\ref{sec:signals}--\ref{sec:ABCd}
one encounters LTSs $\IT$ enriched with \emph{{\signal}s}, revealing facts about states.
While these are naturally modelled as unary predicates on the states of $\IT$,
it is technically possible to model them as ordinary transitions $t$, satisfying
$\source(t)=\target(t)$ \cite{Bou18}. This is formalised by declaring a set of \emph{actions}
$Act \subseteq \Lab\!$.
Transitions $t$ model the occurrence of an action $\ell(t)$ if $\ell(t)\in Act$,
or the revelation of the fact $\ell(t)$ otherwise.
Indicator transitions are largely ignored in the definitions below.

$\Rec\subseteq Act$ is the set of \emph{receptive} actions.
Sets $B\subseteq Act$ of blocking actions must always contain $\Rec$.
In CCS and most other process algebras $\Rec=\emptyset$ and $Act=\Lab$.
Let $\Tr^\bullet = \{t \in \Tr \mid \ell(t)\in Act\setminus\Rec\}$ be the set of transitions that are neither
{\signal} transitions nor receptive.

\begin{definition}{path}
A \emph{path} in a labelled transition system $(S,\Tr,$ $\source,\target,\ell)$ is an alternating sequence
$s_0\,t_1\,s_1\,t_2\,s_2\cdots$ of states and non-{\signal} transitions, starting with a state and
either being infinite or ending with a state, such that $\source(t_i)=s_{i-1}$ and
$\target(t_i)=s_i$ for all relevant $i$.
\end{definition}
A \emph{completeness criterion} is a unary predicate on the paths in a labelled transition system.

\begin{definition}{progress}
Let $B\subseteq Act$ be a set of actions with $\Rec\subseteq B$---the \emph{blocking} ones.
Then $\Tr^\bullet_{\neg B} := \{t\in \Tr^\bullet \mid \ell(t)\notin B\}$ is the set of
\emph{non-blocking} transitions.
A path in $\IT$ is \emph{$B$-progressing} if either it is infinite or its last state is
the source of no non-blocking transition $t \in \Tr^\bullet_{\neg B}$.
\end{definition}
$B$-progress is a completeness criterion for any choice of $B\subseteq Act$ with $\Rec\subseteq B$.

\begin{definition}{LTSC}
  A \emph{labelled transition system with concurrency} (LTSC) is a tuple $(S, \Tr, \source,\target,\ell,\aconc)$
  consisting of an LTS $(S, \Tr, \source,\target,\ell)$ and a \emph{concurrency relation}
  ${\aconc} \subseteq \Tr^\bullet \times \Tr$, such that:
  \begin{equation}\label{irreflexivity}
  \mbox{$t \naconc t$ for all $t \in\Tr^\bullet$;}
  \end{equation}
  \begin{equation}\label{closure}\begin{minipage}{5.1in}{
  if $t\in\Tr^\bullet$ and $\pi$ is a path from $\source(t)$ to $s\in S$ such that $t \aconc v$ for
  all transitions $v$ occurring in $\pi$, then there is a $u\in\Tr^\bullet$ such that $\source(u)=s$,
  $\ell(u)=\ell(t)$ and $t \naconc u$.}
  \end{minipage}\end{equation}
\end{definition}
Informally, $t\aconc v$ means that the transition $v$ does not interfere with $t$, in the sense that
it does not affect any resources that are needed by $t$, so that in a state where $t$ and $v$ are
both possible, after doing $v$ one can still do (a future variant $u$ of) $t$.
In many transition systems $\aconc$ is a symmetric relation, denoted $\conc$.

The transition relation in a labelled transition system is often defined as a relation
$\Tr \subseteq S \times \Lab \times S$. This approach is not suitable here, as \we will encounter
multiple transitions with the same source, target and label that ought to be distinguished based on
their concurrency relations with other transitions.

\begin{definition}{justness}
  A path $\pi$ in an LTSC is \emph{$B$-just}, for $\Rec\subseteq B\subseteq Act$,
  if for each suffix $\pi'$ of $\pi$, and for each transition $t \in \Tr^\bullet_{\neg B}$
  enabled in the starting state of $\pi'$, a transition $u$ with $t \naconc u$ occurs in $\pi'$.
\end{definition}
Informally, justness requires that once a non-blocking non-{\signal} transition $t$ is enabled,
sooner or later a transition $u$ will occur that interferes with it, possibly $t$ itself.

Note that, for any $\Rec\subseteq B\subseteq Act$, $B$-justness is a completeness criterion stronger than $B$-progress.

In reasonable extensions of $\aconc$ to $\Tr\times\Tr$, {\signal} transitions $t$ would satisfy
$t\aconc t$, meaning that execution of $t$ in no way affects any resources needed to execute $t$
again. It therefore makes no sense to impose closure property (\ref{closure}), or the justness
requirement, on {\signal} transitions (see \ex{unfeasible}).

\paragraph{Components}

Instead of introducing $\aconc$ as a primitive, it is possible to obtain it as a notion derived from
two functions $\npc: \Tr^\bullet \rightarrow \Pow(\C)$ and $\afc: \Tr \rightarrow \Pow(\C)$, for a
given set of \emph{components} $\C$.
These functions could then be added as primitives to the definition of an LTS\@.
They are based on the idea that a process represents a system built from parallel components.
Each transition is obtained as a synchronisation of activities from some of these components.
Now $\npc(t)$ describes the (nonempty) set of components that are \emph{necessary participants} in
the execution of $t$, whereas $\afc(t)$ describes the components that are \emph{affected} by the
execution of $t$. The concurrency relation is then defined by
\[ t \aconc u ~~\Leftrightarrow~~ \npc(t) \cap \afc(u) = \emptyset \;,\]
saying that $u$ interferes with $t$ if and only if a necessary participant in $t$ is affected by $u$.

Most material in this section stems from \cite{GH19}. However, there $\Tr^\bullet=\Tr$, so that
$\aconc$ is irreflexive, i.e., $\npc(t) \cap \afc(t) \neq \emptyset$ for all $t\in\Tr$.
Moreover, a fixed set $B$ is postulated,
so that the notions of progress and justness are not explicitly parametrised with the choice of $B$.
Furthermore, closure property (\ref{closure}) is new here; it is the weakest closure property that
supports \thm{feasibility} and \pr{closure} below. In \cite{GH19} only the model in which $\aconc$
is derived from functions $\npc$ and $\afc$ comes with a closure property:
\begin{equation}\label{closureComp}\begin{minipage}{5.2in}{
   If $t\in\Tr^\bullet$ and $v\in\Tr$ with $\source(t)=\source(v)$ and $\npc(t)\cap\afc(v)=\emptyset$,
   then\\ there is a $u\in \Tr^\bullet$ with $\source(u)=\target(v)$, $\ell(u)=\ell(t)$ and $\npc(u)=\npc(t)$.}
\end{minipage}\end{equation}
Trivially (\ref{closureComp}) implies (\ref{closure}).

\section{Feasibility}\label{sec:feasibility}

An important requirement on completeness criteria is that any finite path can be extended into a complete
path. This requirement was proposed by {Apt, Francez \& Katz} in \cite{AFK88} and called
\emph{feasibility}. It also appears in {Lamport}~\cite{La00-lb} under the name \emph{machine closure}.
The theorem below list conditions under which $B$-justness is feasible.
Its proof is a variant of a similar theorem from \cite{GH19} showing conditions under which 
notions of strong and weak fairness are feasible.

\begin{theorem}{feasibility}
  If, in an LTSC with set $B$ of blocking actions, only countably many transitions from
  $\Tr^\bullet_{\neg B}$ are enabled in each state, then $B$-justness is feasible.
\end{theorem}
\begin{proof}
  \We present an algorithm for extending any given finite path $\pi_0$ into a $B$-just path $\pi$.
  The extension uses transitions from $\Tr^\bullet_{\neg B}$ only.
  \We build an $\IN\times\IN$-matrix with column $i$ for the---to be constructed---prefix $\pi_{i}$
  of $\pi$, for $i\geq 0$.
  Column $i$ will list the transitions from $\Tr^\bullet_{\neg B}$ enabled in the last
  state of $\pi_i$, leaving empty most slots if there are only finitely many.
  An entry in the matrix is either (still) empty, filled in with a transition, or crossed out.
  Let $f:\IN\rightarrow \IN\times\IN$ be an enumeration of the entries in this matrix.
  
  At the beginning only $\pi_0$ is known, and all columns of the matrix are empty.
  At each step $i\geq 0$ \we fill in column $i$, extend the path $\pi_i$ into $\pi_{i+1}$ if possible
  by appending one transition (and its target state), and cross out some transitions occurring in the matrix.
  As an invariant, \we maintain that a transition $t$ occurring in column $k$ is already crossed
  out when reaching step $i>k$ iff a transition $u$ occurs in the extension of $\pi_k$ into $\pi_i$
  such that $t \naconc u$.  At each step $i\geq 0$ \we proceed as follows:
  
  Since $\pi_i$ is known, \we fill in column $i$ by listing all transitions from
  $\Tr^\bullet_{\neg B}$ enabled in the last state of $\pi_i$.
  \We take $n$ to be the smallest value such that entry $f(n)\in\IN\times\IN$ is already
  filled in, say with $t\in\Tr^\bullet_{\neg B}$, but not yet crossed out. If such an $n$ does not
  exist, the algorithm terminates, with output $\pi_i$. Let $k$ be the column in which $f(n)$ appears.
  By our invariant, all transitions $v$ occurring in the extension of $\pi_k$ into $\pi_i$
  satisfy $t \aconc v$. By (\ref{closure}) there is a transition $u\in\Tr^\bullet_{\neg B}$ enabled in the
  last state of $\pi_i$ such that $t \naconc u$. \We now extend $\pi_i$ into $\pi_{i+1}$
  by appending $u$ to it, while crossing out all entries $t'$ in the matrix for which $t' \naconc u$,
  including $f(n)$, which is the entry in column $k$ representing the transition $t$. This maintains our invariant.
  
  Obviously, $\pi_{i}$ is a prefix of $\pi_{i+1}$, for $i\geq 0$.
  The desired path $\pi$ is the limit of all the $\pi_i$.
  It is $B$-just, using the invariant, because each transition $t\in\Tr^\bullet_{\neg B}$ that is
  enabled in a state of $\pi$ is either interfered with by a transition occurring in $\pi_0$, or
  will appear in the matrix, which acts like a priority queue, and be eventually crossed out.
\end{proof}
It is possible to strengthen \thm{feasibility} somewhat by calling two transitions $t$ and $t'$
\emph{equivalent} if $t\aconc u \Leftrightarrow t'\aconc u$ for all $u \in \Tr^\bullet$.
An equivalence class of transitions is \emph{enabled} iff one of its elements is.
\begin{corollary}{feasibility}
  If, in an LTSC with set $B$ of blocking actions, only countably many equivalence classes of transitions from
  $\Tr^\bullet_{\neg B}$ are enabled in each state, then $B$-justness is feasible.
\end{corollary}
\begin{proof}
  The proof is the same as the one above, except that the matrix now contains equivalence classes of
  enabled transitions.
\end{proof}

\section{Fairness}\label{sec:fairness}

Let $\Tr^\circ = \{t \in \Tr \mid \ell(t)\in Act\}$.
To formalise fairness \we use LTSs $(S,\Tr,\source,\target,\ell,\Tk)$
that are augmented with a set $\Tk\subseteq\Pow({\Tr^\circ})$ of \emph{tasks}
$T\subseteq {\Tr^\circ}$, each being a set of transitions.
The concept of \emph{\J-fairness} from \cite{GH19} is defined only for LTSCs
$(S,\Tr,\source,\target,\aconc,\ell,\Tk)$ augmented with such a $\Tk$.

\begin{definitionA}{\cite{GH19}}{fair}
For an augmented LTS $\IT=(S,\Tr,\source,\target,\ell,[\aconc,]\Tk)$ and a set $\Rec\subseteq B\subseteq Act$ of
blocking actions, a task $T\mathbin\in\Tk$ is $B$-\emph{enabled} in a state $s\in S$ if there exists a
non-blocking transition $t\in T$ with $\ell(t)\notin B$ and $\source(t)=s$. [It is $B$-enabled
\emph{during the execution} of a transition $u\in \Tr$ if there exists a $t\in T$ with
$\ell(t)\notin B$, $\source(t)=\source(u)$ and $t \aconc u$.]
The task \emph{occurs} in a path $\pi$ in $\IT$ if $\pi$ contains a transition $t\mathbin\in T\!$.
It is said to be \emph{relentlessly $B$-enabled} on $\pi$, if each suffix of $\pi$ contains a state
in which it is $B$-enabled.\footnote{This is the case if the task is $B$-enabled in infinitely many states
of $\pi$, in a state that occurs infinitely often in $\pi$, or in the last state of a finite $\pi$.}
It is \emph{perpetually $B$-enabled} on $\pi$, if it is $B$-enabled in every state of $\pi$.
[It is said to be \emph{continuously $B$-enabled} on $\pi$, if it is $B$-enabled in every state
and during every transition of $\pi$.]

A path $\pi$ in $\IT$ is \emph{strongly $B$-fair} if, for every suffix $\pi'$ of $\pi$,
each task that is relentlessly $B$-enabled on $\pi'$, occurs~in~$\pi'$.
A path $\pi$ in $\IT$ is \emph{weakly $B$-fair} if, for every suffix $\pi'$ of $\pi$,
each task that is perpetually $B$-enabled on $\pi'$, occurs in $\pi'$.
[A path $\pi$ in $\IT$ is \emph{\J-$B$-fair} if, for every suffix $\pi'$ of $\pi$,
each task that is continuously $B$-enabled on $\pi'$, occurs in $\pi'$.]
\end{definitionA}
When the set $B$ is defined once and for all or clear from context, \we may omit the parameter $B$.
This was the situation in \cite{GH19}.

In \cite{GH19} many notions of fairness occurring in the literature were cast as instances of this
definition. For each of them the set of tasks $\Tk$ was derived, in different ways, from some other
structure present in the model of distributed systems from the literature. In fact, \cite{GH19}
considers 7 ways to construct the collection $\Tk\!$, and speaks of fairness of \emph{actions},
\emph{transitions}, \emph{instructions}, \emph{synchronisations}, \emph{components},
\emph{groups of components} and \emph{events}. This yields 21 notions of fairness. To compare them,
each is defined formally on a fragment of CCS, and the 21 fairness notions (together with progress,
justness, and a few concepts of fairness found in the literature that are not instances of
\df{fair}) are ordered by strength by placing them in a lattice.

Progress can be cast as a fairness notion in the sense of \df{fair} by taking $\Tk$ to be the
collection of only one task, namely $\Tr^\circ$. Clearly weak, strong and J-fairness all coincide
for this $\Tk$. Likewise, the trivial completeness criterion, declaring all paths complete,
coincides with weak, strong and J-fairness when taking $\Tk=\emptyset$.
Nevertheless, it would be confusing to address these completeness criteria as fairness assumptions.

\We do not see how justness can be cast a fairness notion in the sense of \df{fair}.
However, \we now show that there exists a form of fairness according to \df{fair} that is at least as
strong as justness.
Namely take $\Tk:=\{T_t \mid t \in \Tr^\bullet\}$ where $T_t:= \{u \in \Tr^\circ \mid t \naconc u\}$.

\begin{proposition}{closure}
Given this $\Tk$ and $B\subseteq Act$, any path that is strongly or weakly $B$-fair is certainly $B$-just.
\end{proposition}
\begin{proof}
Any path that is strongly $B$-fair is certainly weakly $B$-fair.
This follows trivially from the definitions, for any choice of $\Tk$.
(Likewise, any path that is weakly $B$-fair is certainly J-$B$-fair.)

Suppose $\pi$ is weakly $B$-fair. \We show it is $B$-just. Suppose that $t\in \Tr^\bullet_{\neg B}$ is enabled
in the first state $s$ of a suffix $\pi'$ of $\pi$, i.e., $\source(t)=s$, but all transitions $v$
occurring in $\pi'$ satisfy $t \aconc v$.
Closure property (\ref{closure}) guarantees that for every state $s'$ of $\pi'$ there is a
$u\in\Tr^\bullet$ such that $\source(u)=s'$, $\ell(u)=\ell(t)$ and $t \naconc u$.
Hence task $T_t$ is perpetually $B$-enabled on $\pi'$.
By weak $B$-fairness $T_t$ must occur in $\pi'$, meaning that $\pi'$ contains a transition
$u\in\Tr^\circ$ with $t \naconc u$. This contradicts the assumptions.
\end{proof}

\section{CCS and its extensions with broadcast communication and signals}\label{sec:PA}

This section presents five process algebras: Milner's \emph{Calculus of Communicating Systems} (CCS) \cite{Mi80},
its extensions ABC with broadcast communication \cite{GH15a} and CCSS with signals \cite{EPTCS255.2},
an alternative presentation of CCSS where signal emissions are encoded as transitions, and an alternative
presentation of ABC that avoids negative premises in favour of \emph{discard} transitions.

\subsection{CCS}\label{sec:CCS}

\begin{table*}[t]
\caption{Structural operational semantics of CCS}
\label{tab:CCS}
\normalsize
\begin{center}
\framebox{$\begin{array}{c@{}c@{}c}
\alpha.P \goto{\alpha} P  \mylabel{Act} &
\displaystyle\frac{P\goto{\alpha} P'}{P+Q \goto{\alpha} P'}  \mylabel{Sum-l}&
\displaystyle\frac{Q\goto{\alpha} Q'}{P+Q \goto{\alpha} Q'}  \mylabel{Sum-r} \\[4ex]
\displaystyle\frac{P\goto{\eta} P'}{P|Q \goto{\eta} P'|Q} \mylabel{Par-l}&
\displaystyle\frac{P\goto{c} P' ,~ Q \goto{\bar{c}} Q'}{P|Q \goto{\tau} P'| Q'} \mylabel{Comm}&
\displaystyle\frac{Q \goto{\eta} Q'}{P|Q \goto{\eta} P|Q'} \mylabel{Par-r}\\[4ex]
\displaystyle\frac{P \goto{\ell} P'}{P\backslash \RL \goto{\ell}P'\backslash \RL}~~(\ell\not\in\RL\dcup\bar{\RL}) \mylabel{Res}&
\displaystyle\frac{P \goto{\ell} P'}{P[f] \goto{f(\ell)} P'[f]} \mylabel{Rel} &
\displaystyle\frac{P \goto{\alpha} P'}{A\goto{\alpha}P'}~~(A \stackrel{{\it def}}{=} P) \mylabel{Rec}
\end{array}$}
\vspace{-2ex}
\end{center}
\end{table*}

\noindent
CCS \cite{Mi80} is parametrised with sets ${\A}$ of \emph{agent identifiers} and $\Ch$ of \emph{(handshake communication) names};
each $A\in\A$ comes with a defining equation \plat{$A \stackrel{{\it def}}{=} P$} with $P$ being a CCS expression as defined below.
$\bar{\Ch} := \{ \bar{c} \mid c \in \Ch\}$ is the set of \emph{(handshake communication) co-names}.
Complementation is extended to $\bar\Ch$ by setting $\bar{\bar{\mbox{$c$}}}=c$.
\plat{$Act := \Ch\dcup\bar\Ch\dcup \{\tau\}$} is the set of {\em actions}, where $\tau$ is a special \emph{internal action}.
Below, $A$ ranges over $\A$, $c$ over $\Ch\dcup\bar\Ch$ and $\eta$, $\alpha$, $\ell$ over $Act$.
A \emph{relabelling} is a function $f\!:\Ch\mathbin\rightarrow \Ch$; it extends to $Act$ by
$f(\bar{c})\mathbin=\overline{f(c)}$ and $f(\tau):=\tau$.
The set $\cT_{\rm CCS}$ of CCS expressions or \emph{processes} is the smallest set including:
\begin{center}
\begin{tabular}{@{}lll@{}}
${\bf 0}$ && \emph{inaction}\\
$\alpha.P$  & for $\alpha\mathbin\in Act$ and $P\mathbin\in\cT_{\rm CCS}$ & \emph{action prefixing}\\
$P+Q$ & for $P,Q\mathbin\in\cT_{\rm CCS}$ & \emph{choice} \\
$P|Q$ & for $P,Q\mathbin\in\cT_{\rm CCS}$ & \emph{parallel composition}\\
$P\backslash \RL$ & for $\RL\subseteq\Ch$ and $P\mathbin\in\cT_{\rm CCS}$ & \emph{restriction} \\
$P[f]$ & for $f$ a relabelling and $P\mathbin\in\cT_{\rm CCS}$ & \emph{relabelling} \\
$A$ & for $A\in\A$ & \emph{agent identifier}\\
\end{tabular}
\end{center}
One often abbreviates $\alpha.{\bf 0}$ by $\alpha$, and $P\backslash\{c\}$ by $P\backslash c$.
The traditional semantics of CCS is given by the labelled transition relation
$\mathord\rightarrow \subseteq \cT_{\rm CCS}\times Act \times\cT_{\rm CCS}$, where transitions 
\plat{$P\ar{\ell}Q$} are derived from the rules of \autoref{tab:CCS}.
Here $\bar{L}:=\{\bar{c} \mid c \in L\}$.
The process $\alpha.P$ performs the action $\alpha$ first and subsequently acts as $P$.
The process $P+Q$ may act as either $P$ or $Q$, depending on which of the processes is able to act at all.
The parallel composition $P|Q$ executes an action from $P$, an action from $Q$, or in the case where
$P$ and $Q$ can perform complementary actions $c$ and $\bar{c}$, the process can perform a synchronisation, resulting in an internal action $\tau$.
The process $P \backslash \RL$
inhibits execution of the actions from $\RL$ and their complements. 
The relabelling $P[f]$ acts like process $P$ with all labels $\ell$ replaced by $f(\ell)$.
Finally, the rule for agent identifiers says that an agent $A$ has the same transitions as the body $P$ of its defining equation.

\subsection{ABC---The Algebra of Broadcast Communication}\label{sec:ABC}

The Algebra of Broadcast Communication (ABC) \cite{GH15a} is parametrised with sets ${\A}$ of \emph{agent identifiers},
$\B$ of \emph{broadcast names} and $\Ch$ of \emph{handshake communication names};
each $\AI\in\A$ comes with a defining equation \plat{$\AI \stackrel{{\it def}}{=} P$}
with $P$ being a guarded ABC expression as defined below.

The collections $\B!$ and $\B?$ of \emph{broadcast} and \emph{receive}
actions are given by $\B\sharp:=\{b\sharp \mid b\mathbin\in\B\}$ for $\sharp \in \{!,?\}$.
\plat{$Act := \B! \dcup \B? \dcup \Ch \dcup \bar\Ch \dcup \{\tau\}$} is the set of \emph{actions}.
Below, $\AI$ ranges over $\A\!$, $b$ over $\B$, $c$ over $\Ch\dcup\bar\Ch$,
$\eta$ over $\Ch\dcup\bar\Ch\dcup\{\tau\}$ and $\alpha,\ell$ over $Act$.
A \emph{relabelling} is a function $f\!:(\B\mathbin\rightarrow \B) \dcup (\Ch\mathbin\rightarrow \Ch)$.
It extends to $Act$ by $f(\bar{c})\mathbin=\overline{f(c)}$, $f(b\sharp)\mathbin{:=}f(b)\sharp$ and $f(\tau):=\tau$.
The set $\cT_{\rm ABC}$ of ABC expressions is defined exactly as $\cT_{\rm CCS}$.
An expression is guarded if each agent identifier occurs within the scope of a prefixing operator.
The structural operational semantics of ABC is the same as the one for CCS (see \tab{CCS}) but
augmented with the rules for broadcast communication in \tab{ABC}.

\begin{table*}[b]
  \vspace{-2ex}
  \caption{Structural operational semantics of ABC broadcast communication}
  \vspace{1ex}
\normalsize
\centering
\label{tab:ABC}
\framebox{
$\begin{array}{@{}ccc@{}}
\!\displaystyle\frac{\E\mathbin{\goto{b\sharp_1}} \E' ,~ \F \nar{b?}}{\E|\F \goto{b\sharp_1} \E'| \F}  \mylabel{Bro-l}&
\displaystyle\frac{\E\mathbin{\goto{b\sharp_1}} \E' ,~ \F \goto{b\sharp_2} \F'}{\E|\F \goto{b\sharp} \E'| \F'}
 \scriptstyle \sharp_1\circ\sharp_2 = \sharp \neq \_ ~~\mbox{with}~~
    \begin{array}{c@{\ }|@{\ }c@{\ \ }c}
    \scriptstyle \circ & \scriptstyle ! & \scriptstyle ? \\
    \hline
    \scriptstyle ! & \scriptstyle \_ & \scriptstyle ! \\
    \scriptstyle ?& \scriptstyle ! & \scriptstyle ? \\
    \end{array}
 \mylabel{Bro-c} &
\displaystyle\frac{\E\nar{b?} ,~ \F \mathbin{\goto{b\sharp_2}} \F'}{\E|\F \goto{b\sharp_2} \E| \F'}  \mylabel{Bro-r}\\[-3pt]
\end{array}$}
\end{table*}

ABC is CCS augmented with a
formalism for broadcast communication taken from the Calculus of Broadcasting Systems (CBS)~\cite{CBS91}.
The syntax without the broadcast and receive actions and all rules except \myref{Bro-l},
\myref{Bro-c} and \myref{Bro-r} are taken verbatim from CCS\@. However, 
the rules now cover the different name spaces; \myref{Act} for example allows labels of
broadcast and receive actions. The rule \myref{Bro-c}\linebreak[3]---without rules
like \myref{Par-l} and \myref{Par-r} for label $b!$---implements a
form of broadcast communication where any broadcast $b!$
performed by a component in a parallel composition is guaranteed to be received by any other
component that is ready to do so, i.e., in a state that admits a $b?$-transition.
In order to ensure associativity of the parallel composition, one also needs this 
rule for components receiving at the same time ($\sharp_1\mathord=\sharp_2\mathord=\mathord{?}$).
The rules \myref{Bro-l} and \myref{Bro-r} are added to make
broadcast communication \emph{non-blocking}: without them a component could be delayed in
performing a broadcast simply because one of the other components is not ready to receive it.

\subsection{CCS with signals}\label{sec:signals}

\emph{CCS with signals} (CCSS) \cite{EPTCS255.2} is CCS extended with a signalling operator
$P\signals s$. Informally, $P\signals s$ emits the signal $s$ to be read by another
process. $P\signals s$ could for instance be a traffic light emitting the signal \emph{red}.
The reading of the signal emitted by $P\signals s$ does not interfere with any transition
of $P$, such as jumping to \emph{green}.
Formally, CCS is extended with a set $\Sig$ of \emph{signals}, ranged over by $s$ and $r$.
In CCSS the set of actions is defined as \plat{$Act := \Sig \dcup \Ch\dcup\bar\Ch\dcup\{\tau\}$}.
A relabelling is a function $f:(\Sig\rightarrow\Sig)\dcup(\Ch\rightarrow\Ch)$.
As before it extends to $Act$ by $f(\bar{c})\mathbin=\overline{f(c)}$ and $f(\tau):=\tau$.
The set $\cT_{\rm CCSS}$ of CCSS expressions is defined just as $\cT_{\rm CCS}$, but now also
$P\signals s$ is a process for $P\in\cT_{\rm CCSS}$ and $s\in\Sig$, and restriction also
covers signals.

The semantics of CCSS is given by the labelled transition relation
$ \mathord\rightarrow \subseteq \cT_{\rm CCSS}\times Act \times \cT_{\rm CCSS}$
and a predicate $\sigar{}\subseteq \cT_{\rm CCSS}\times \Sig$ that
are derived from the rules of CCS (\autoref{tab:CCS}, where $\eta,\alpha,\ell$ range over $Act$ and
$L\subseteq\Ch\dcup\Sig$), and the rules of \autoref{tab:CCSssos}.
\begin{table}[t]
\caption{Structural operational semantics for signals of CCSS}
\normalsize
\begin{center}
\framebox{$\begin{array}{cccc}
(P \signals s)\sigar{s}&
\displaystyle\frac{P\ar{\alpha}P'}{P\signals r \ar{\alpha}P'}&
\displaystyle\frac{{P}\sigar{s}}{(P+Q) \sigar{s}}&
\displaystyle\frac{{Q}\sigar{s}}{(P+Q) \sigar{s}}\\[4ex]
\displaystyle\frac{P\sigar{s}}{(P|Q) \sigar{s}} &
\displaystyle\frac{P\sigar{s},~ Q \ar{s} Q'}{P|Q \ar{\tau} P| Q'} &
\displaystyle\frac{P\ar{s}P',~ Q \sigar{s}}{P|Q \ar{\tau} P'| Q}&
\displaystyle\frac{Q \sigar{s}}{(P|Q) \sigar{s}}\\[4ex]
\displaystyle\frac{P\sigar{s}}{(P\signals r)\sigar{s}}&
\displaystyle\frac{P \sigar{s}}{(P\backslash \RL) \sigar{s}}~~(s\not\in\RL) &
\displaystyle\frac{P \sigar{s}}{P[f] \sigar{f(s)}} &
\displaystyle\frac{P \sigar{s}}{A\sigar{s}}~~(A \stackrel{{\it def}}{=} P)
\end{array}$}
\end{center}
\label{tab:CCSssos}
\end{table}
The predicate $P\sigar{s}$ indicates that process $P$ emits the signal $s$,
whereas a transition $P \ar{s} P'$ indicates that $P$ reads the signal $s$ and thereby turns into $P'$.
The first rule is the base case showing that a process $P\signals s$ emits the signal $s$.
The second rule of \autoref{tab:CCSssos} models the fact that signalling cannot prevent a process from making progress. 
After having taken an action, the signalling process loses its ability to emit the signal.
The two rules in the middle of \autoref{tab:CCSssos} state that the action of reading a signal by one component in
(parallel) composition together with the emission of the same signal by another component,
results in an internal transition $\tau$; similar to the case of handshake communication.
Note that the component emitting the signal does not change through this interaction.
All the other rules of \autoref{tab:CCSssos} lift the emission of $s$ by a subprocess $P$ to the overall process.

\subsection{Encoding signal emissions as transitions}\label{sec:CCSS}

A more compact presentation of CCSS can be obtained by encoding a signal emission $P\sigar s$ as a
transition \plat{$P \goto{\bar s} P$}; this is done in \cite{Bou18}.
The price to be paid for the resulting simplification of the operational
semantics is that the new transitions \plat{$P \goto{\bar s} P$} should not be counted in the definition of justness,
since they do not model changes in the state of the represented system.

In this presentation of CCSS the set of labels is defined as $\Lab:=Act \dcup \bar\Sig$, where $Act$
is defined as in the previous section and $\bar{\Sig} := \{ \bar{s} \mid s \in \Sig\}$.
Complementation is extended to $\bar\Ch \dcup \bar\Sig$ by setting $\bar{\bar{\mbox{$c$}}}:=c$, with $c\in\Ch\dcup\Sig$.
A relabelling is a function $f:(\Sig\rightarrow\Sig)\dcup(\Ch\rightarrow\Ch)$;
it extends to $\Lab$ by $f(\bar{c})\mathbin=\overline{f(c)}$ for $c\in\Ch\dcup\Sig$, and $f(\tau):=\tau$.
The semantics is given by the labelled transition relation
$ \mathord\rightarrow \subseteq \cT_{\rm CCSS}\times \Lab \times \cT_{\rm CCSS}$
derived from the rules of CCS (\autoref{tab:CCS}), where now $\eta,\ell$ range over $\Lab$, $\alpha$
over \plat{$Act$}, $c$ over $\Ch \dcup\Sig$ and $L\subseteq\Ch\dcup\Sig$,
augmented with the rules of \autoref{tab:CCSS}.
\begin{table}[t]
\caption{Structural operational semantics of CCSS when signal emissions are encoded as transitions}
\normalsize
\begin{center}
\framebox{$\begin{array}{cccc}
P \signals s\goto{\bar s} P\signals s&
\displaystyle\frac{{P}\goto{\bar s}P'}{P+Q\goto{\bar s}P'+Q}&
\displaystyle\frac{{Q}\goto{\bar s}Q'}{P+Q\goto{\bar s}P+Q'}\\[4ex]
\displaystyle\frac{P\goto{\alpha}P'}{P\signals r \goto{\alpha}P'}&
\displaystyle\frac{P\goto{\bar s}P'}{P\signals r \goto{\bar s}P'\signals r}&
\displaystyle\frac{P\goto{\bar s} P'}{A\goto{\bar s}A}~~(A \stackrel{{\it def}}{=} P)
\end{array}$}
\end{center}
\label{tab:CCSS}\vspace{-2ex}
\end{table}

\subsection{Using {\signal} transitions to avoid negative premises in ABC}\label{sec:ABCd}

Finally, \we present an alternative operational semantics ABCd of ABC that avoids negative premises.
The price to be paid is the introduction of {\signal} transitions that indicate when a process does
not admit a receive action.%
\footnote{A process $P$ admits an action $\alpha\in Act$ if there exists a transition $P\ar{\alpha}Q$.\label{admit}}
To this end, let $\B\mathord{:} :=\{b\mathord: \mid b\in\B\}$ be the set of \emph{broadcast discards}, and
\plat{$\Lab := \B\mathord: \dcup Act$} the set of \emph{transition labels}, with $Act$ as in \Sec{ABC}.
The semantics is given by the labelled transition relation
$\mathord\rightarrow \subseteq \cT_{\rm ABC}\times \Lab \times \cT_{\rm ABC}$
derived from the rules of CCS (\autoref{tab:CCS}), where now $c$ ranges over $\Ch\dcup\bar\Ch$,
$\eta$ over $\Ch\dcup\bar\Ch\dcup\{\tau\}$, $\alpha$ over $Act$ and $\ell$ over $\Lab$,
and moreover $L\subseteq \Ch$ and $f:(\B\rightarrow\B)\dcup(\Ch\rightarrow\Ch)$,
augmented with the rules of \autoref{tab:ABCd}.
\begin{table}[b]\vspace{-3ex}%
\caption{Structural operational semantics of ABC broadcast communication with discard transitions}
\normalsize
\begin{center}
  \framebox{
$\begin{array}{ccc}
{\bf 0} \ar{b:} {\bf 0} & 
\alpha.\E \ar{b:} \alpha.\E ~~ \mbox{($\alpha \mathord{\neq} b?$)} & 
\displaystyle\frac{\E\ar{b:} \E' ,~ \F \ar{b:} \F'}{\E+\F \ar{b:} \E'+ \F'} \\[2ex] 
\multicolumn{2}{c}{\displaystyle\frac{\E\ar{b\sharp_1} \E' ,~ \F \ar{b\sharp_2} \F'}{\E|\F \ar{b\sharp} \E'| \F'}
~~\scriptstyle \sharp_1\circ\sharp_2 = \sharp \neq \_ ~~\mbox{with}~~
\begin{minipage}{0.15\textwidth}
\vspace{-1mm}
$~~\begin{array}{c@{\ }|@{\ }c@{\ \ }c@{\ \ }c}
     \circ &  ! &  ? &  : \\
    \hline
     ! &  \_ &  ! &  ! \\
     ?&  ! &  ? &  ? \\
     : &  ! &  ? &  :
    \end{array}$
\end{minipage}
}& 
\displaystyle\frac{P \ar{b:} P'}{A\ar{b:}A}~~(A \stackrel{{\it def}}{=} P)
\end{array}$}
\end{center}
\label{tab:ABCd}
\end{table}

\begin{lemma}{discards} {\rm \cite{CBS91}}
$P \ar{b:} Q$ iff $Q=P \wedge P \nar{b?}$\,, 
for $P,Q\in\cT_{\rm ABC}$ and $b\in\B$.
\end{lemma}

\begin{proof}
A straightforward induction on derivability of transitions.
\end{proof}

\begin{corollary}{modified same}
The structural operational semantics of ABC from Sections~\ref{sec:ABC} and ~\ref{sec:ABCd}
yield the same labelled transition relation $\longrightarrow$ when transitions labelled $b\mathord{:}$ are ignored.
\hfill$\Box$
\end{corollary}
This approach stems from the Calculus of Broadcasting Systems (CBS)~\cite{CBS91}.

\section{An LTSC for CCS and its extensions}\label{sec:LTSC}

The forthcoming material applies to each of the process algebras from \Sec{PA}.
Let $\cT$ be the set of processes or expressions in the appropriate language. 

\We allocate an LTS as in \df{LTS} to these languages by taking $S$ to be the set $\cT$ of processes, and $\Tr$ the set
of \emph{derivations} $\chi$ of transitions \plat{$P\goto{\ell}Q$} with $P,Q\in \cT$.
Of course $\source(\chi)\mathbin=P$, $\target(\chi)\mathbin=Q$ and $\ell(\chi)\mathbin=\ell$.
A \emph{derivation} of a formula $\varphi$ (either a transition \plat{$P\goto{\ell}Q$} or a
predicate $P\sigar s$) is a well-founded tree with the nodes labelled by formulas, such that the
root has label $\varphi$, and if $\mu$ is the label of a node and $K$ is the sequence of labels of the
children of this node then \plat{$\frac{K}{\mu}$} is an instance of a rule of Tables~\ref{tab:CCS}--\ref{tab:ABCd}.

{\makeatletter
\let\par\@@par
\par\parshape0
\everypar{}\begin{wrapfigure}[17]{r}{0.44\textwidth}
 \input{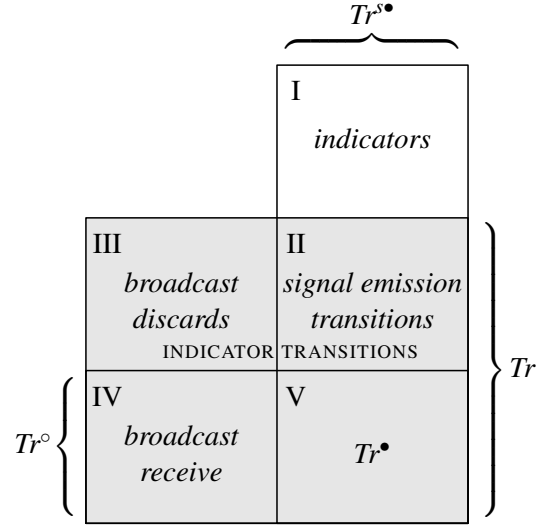}
  \centerline{\box\graph}
 \caption{\it\small Indicators and transitions}
 \label{transitions}
 \end{wrapfigure}
\noindent
In classical process algebra, and in \Sec{PA}, a transition is a formula of the form \plat{$P \goto\ell Q$}.
The CCS process $P=A|(\bar c + \tau)$ for instance, where the agent identifier $A$ has the defining
equation \plat{$A \stackrel{{\it def}}{=} c.A$}, has 3 outgoing transitions:
\plat{$P \goto{c} P$}, \plat{$P \goto{\bar c} A| {\bf 0}$} and \plat{$P \goto{\tau} A| {\bf 0}$}.
The last of these transitions can be derived in two different ways:\linebreak through a synchronisation between
$c$ and $\bar c$, or through the $\tau$ from the right component.
Here \we distinguish these two $\tau$-transitions, so that \we can say that one of them is
concurrent with the $c$-transition whereas the other is not. This is the reason that as the
transitions in the sense of \Sec{justness} \we take the derivations of transitions in the sense of \Sec{PA}.
When confusion is unlikely, \we will use the word ``transition'' for an element of $\Tr$, which is a derivation of a
formula \plat{$P \goto\alpha Q$}. Likewise, an ``{\signal}'' is a derivation of a formula $P \sigar{s}$.

In this paper \we distinguish five kinds of formula derivations, as indicated in \fig{transitions}.
Class I contains the {\signal}s, and Classes II--V the transitions. Classes II, III and IV contain the
transitions with labels in $\bar\Sig$, $\B{:}$ and $\B?$, respectively, and Class~V contains
all others. Recall that \plat{$\Lab = Act \dcup \bar\Sig \dcup \B{:}$}. So Classes
II and III contain the {\signal} transitions, whereas the transitions in Classes IV, V model action occurrences.
The latter ones were collected in the set $\Tr^\circ$ in \Sec{fairness}.
Class IV is inhabited only for ABC or ABCd, and class III only for ABCd. Class I is inhabited only
for the version of CCSS from \Sec{signals}, and Class II only for the one from \Sec{CCSS}.

By \df{path} a path contains transitions from $\Tr^\circ$ only. To properly define the concept of justness
(cf. \df{justness}) \we need a concurrency relation of type ${\aconc} \subseteq \Tr^\bullet \times \Tr^\circ$.
However, at no extra effort \we extend it to ${\aconc} \subseteq \Tr^\bullet \times \Tr$;
in fact $t \aconc u$ will always hold when $u \in \Tr{\setminus}\Tr^\circ$.
Although not required by \df{LTSC}, for some purposes it turns out to be handy to extend the type of
the concurrency relation further to ${\aconc} \subseteq \Tr^{s\bullet} \times \Tr$, where
$\Tr^{s\bullet}$ contains all derivations from Classes I, II and V---see \Sec{inductive}.

\We take $\Rec:=\B?$ in ABC and ABCd: broadcast receipts can always be blocked by the
environment, namely by not broadcasting the requested message. For CCS and CCSS \we take $\Rec:=\emptyset$,
thus allowing environments that can always participate in certain handshakes, and/or always emit
certain signals.

Following \cite{GH15a}, \we give a name to any derivation of a transition:
The unique derivation of the transition \plat{$\alpha.P \ar{\alpha} P$} using the rule \myref{Act} is
called \plat{$\shar{\alpha}P$}. The derivation obtained by application of \myref{Comm} or \myref{Bro-c}
on the derivations $\chi$ and $\zeta$ of the premises of that rule is called $\chi|\zeta$. The
derivation obtained by application of \myref{Par-l} or \myref{Bro-l} on the derivation $\chi$ of the
(positive) premise of that rule, and using process $Q$ at the right of
$|$, is $\chi|Q$. In the same way, \myref{Par-r} and \myref{Bro-r} yield $P|\zeta$,
whereas \myref{Sum-l}, \myref{Sum-r}, \myref{Res}, \myref{Rel} and \myref{Rec} yield $\chi{+}Q$, $P{+}\chi$,
$\chi\backslash \RL$, $\chi[f]$ and \plat{$A\mathord:\chi$}.
These names reflect the syntactic structure of derivations: $\chi|P \not=P |\chi$ 
and $(\chi|\zeta)|v \not= \chi|(\zeta|v)$.

For CCSS as in \Sec{signals} we also name each {\signal} $\xi\notin\Tr$.
The unique derivation of the formula \plat{$(P\signals s)\sigar s$} using the first rule of
\tab{CCSssos} is called \plat{$P\sig s$}.  The other rules of \tab{CCSssos} yield derivations
$\chi\signals r$, $\xi+Q$, $P+\xi$, $\xi|Q$, $\xi|\chi$, $\chi|\xi$, $P|\xi$, $\xi\signals r$,
$\xi\backslash\RL$, $\xi[f]$ and $A{:}\xi$, where $\xi$ is the derivation of the {\signal} premise,
and $\chi$ of the transition premise of the rule. The derivations resulting from the rules from
\Sec{CCSS} are (named) the same as the ones from \Sec{signals}, but now they are all derivations of
transitions $t \in\Tr$; in particular \plat{$P\sig s$} is now the unique derivation of the transition
\plat{$P\signals s \goto{\bar s} P\signals s$} using the first rule of \tab{CCSS}.
\par}

The derivations obtained by application of the rules of \tab{ABCd} are called
$b{:}{\bf 0}$, $b{:}\alpha.P$, $t+u$, $t|u$ and $A{:}t$, where $t$ and $u$ are the derivations of the
premises of these rules.

\paragraph{Synchrons}

Let $\mathit{Arg} := \{+_\Left, +_\R, |_\Left, |_\R, \backslash \RL, [f], A{:}, \mbox{}\signals r
\mid L\subseteq\Ch \dcup \Sig \wedge f \,\mbox{a relabelling} \wedge A\in\A \wedge r\in\Sig\}$.
A \emph{synchron} is an expression $\sigma(\shar{\alpha}P)$ or $\sigma(P\sig s)$ or $\sigma(b{:})$
with $\sigma\in\mathit{Arg}^*$, $\alpha\in Act$, $s \in\Sig$, $P\in\cT$ and $b\in\B$.
An \emph{argument} $\iota \in \mathit{Arg}$ is applied
componentwise to a set $\Sigma$ of synchrons: $\iota(\Sigma) := \{\iota\varsigma \mid \varsigma \in \Sigma\}$.
The set $\varsigma(P)$ of synchrons of a CCS, ABC or CCSS process $P$ is inductively defined by\vspace{-2ex}
\[\begin{array}{l@{~=~}l@{\hspace{1cm}}l@{~=~}l}
\varsigma({\bf 0}) & \emptyset
&
\varsigma(\alpha.P) & \{(\shar{\alpha}P)\}
\\
\varsigma(P+Q) & +_L\varsigma(P) \dcup +_R\varsigma(Q)
&
\varsigma(P|Q) & |_L\varsigma(P) \dcup |_R\varsigma(Q)
\\
\varsigma(P\backslash \RL) & \backslash \RL \, \varsigma(P)
&
\varsigma(P[f]) & [f] \varsigma(P)
\\
\varsigma(A) & A{:}\varsigma(P) \mbox{~~when~ \plat{$A \stackrel{{\it def}}{=} P$}}
&
\varsigma(P\signals s) & \{(P\sig s)\} \dcup {}\signals s\, \varsigma(P)\;.
\end{array}\]
Thus, a synchron of a process $Q$ can be seen as a path in the parse tree of $Q$ to an unguarded
subexpression $\alpha.P$ or $P\signals s$ of $Q$---except that recursion \plat{$A \stackrel{{\it def}}{=} P$}
gets unfolded in the construction of such a path. Here a subexpression of $Q$ occurs
\emph{unguarded} if it does not lay within a subexpression $\beta.R$ of $Q$.

For ABCd \we amend the clauses for inaction and prefixing:
\[\begin{array}{l@{~=~}l@{\hspace{1cm}}l@{~=~}l}
\varsigma({\bf 0}) & \{(b{:}) \mid b\in \B\}
&
\varsigma(\alpha.P) & \{(b {:}) \mid b\in \B \wedge b? \neq \alpha\} \dcup \{(\shar{\alpha}P)\}
\end{array}\]

The set of synchrons $\varsigma(\chi)$ of a derivation $\chi$ of a transition \plat{$P\goto{\ell}Q$} or formula $P\sigar s$ is
defined by
\[\begin{array}{l@{~=~}l@{\hspace{1cm}}l@{~=~}l@{\hspace{1cm}}l@{~=~}l}
\varsigma(\shar{\alpha}P) & \{(\shar{\alpha}P)\}
&
\varsigma(\chi+Q)         & +_L\varsigma(\chi)
&
\varsigma(P+\chi)         & +_R\varsigma(\chi)
\\
\varsigma(\chi|Q)         & |_L\varsigma(\chi)
&
\varsigma(\chi|\zeta)     & |_L\varsigma(\chi) \dcup |_R\varsigma(\zeta)
&
\varsigma(P|\zeta)        & |_R\varsigma(\zeta)
\\
\varsigma(\chi\backslash \RL) & \backslash \RL\, \varsigma(\chi)
&
\varsigma(\chi[f])        & [f] \varsigma(\chi)
&
\varsigma(A{:}\chi)    & A{:}\varsigma(\chi)
\\
\varsigma(P\sig s)  & \{(P\sig s)\}
&
\varsigma(\chi\signals r) & {}\signals r\, \varsigma(\chi)
\\
\varsigma(b{:}{\bf 0}) & \{(b{:})\}
&
\varsigma(b{:}\alpha.P) & \{(b{:})\}
&
\varsigma(\chi+v)         & +_L\varsigma(\chi) \dcup +_R\varsigma(v)\;.
\end{array}\]
Thus, a synchron of $\chi$ represents a path in the proof-tree $\chi$ from its root to a leaf.
Note that \we use the symbol $\varsigma$ as a variable ranging over synchrons, and as the name of
two functions---context disambiguates.

\begin{lemma}{synchrons}
If $\chi$ is a derivation of \plat{$P\goto{\ell}Q$} or \plat{$P\sigar s$} then $\varsigma(\chi)\subseteq\varsigma(P)$.
\end{lemma}
\begin{proof}
A trivial structural induction on $\chi$.
\end{proof}
Each transition derivation can be seen as the synchronisation of one or more synchrons.

\begin{example}{concurrent}
The CCS process $P=\left(\big(c.Q + (d.R|e.S)\big)|\bar c.T\right)\backslash c$ has 3 outgoing transitions:
$P \ar{\tau} (Q|T)\backslash c$, $P \ar{d} ((R|e.S)|\bar c.T)\backslash c$ and $P \ar{e} ((d.R|S)|\bar c.T)\backslash c$.
Let $\chi_\tau$, $\chi_d$ and $\chi_e\in\Tr$ be the unique derivations of these transitions.
Then $\chi_\tau$ is a synchronisation of two synchrons, whereas $\chi_d$ and $\chi_e\in\Tr$ have
only one each:
\plat{$\varsigma (\chi_\tau) = \{\backslash c\, {|_L}\, {+_L} (\shar{c}Q) , \backslash c\, {|_R} (\shar{\bar c}T)\}$},
\plat{$\varsigma (\chi_d) = \{\backslash c\, {|_L}\, {+_R}\, {|_L} (\shar{d}R)\}$} and
\plat{$\varsigma (\chi_e) = \{\backslash c\, {|_L}\, {+_R}\, {|_R} (\shar{e}S)\}$}.
The derivations $\chi_d, \chi_e\in\Tr$ can be seen as \emph{concurrent}, because their
synchrons come from opposite sides of the same parallel composition; one would expect that after
one of them occurs, a variant of the other is still possible. Indeed, there is a transition
$((d.R|S)|\bar c.T)\backslash c \ar{d} ((R|S)|\bar c.T)\backslash c$. Let $\chi'_d$ be its unique derivation.
The derivations $\chi_d$ and $\chi'_d$ are surely different, for they have a different source state.
Even their synchrons are different: \plat{$\varsigma (\chi'_d) = \{\backslash c\, {|_L}\, {|_L} (\shar{d}R)\}$}.
Nevertheless, $\chi'_d$ can be recognised as a future variant of $\chi_d$: its only synchron has
merely lost an argument $+_R$. This choice got resolved when taking the transition~$\chi_e$.
\end{example}
\We proceed to formalise the concepts ``future variant'' and ``concurrent'' that occur above, by
defining two binary relations ${\leadsto} \subseteq \Tr^{s\bullet}\times\Tr^{s\bullet}$ and
${\aconc} \subseteq \Tr^{s\bullet}\times\Tr$ such that the following
properties hold:\vspace{-4ex}

\begin{equation}\label{closureA}\begin{minipage}{5.8in}{
  The relation $\leadsto$ is reflexive and transitive.}
\end{minipage}\vspace{-5ex}\end{equation}

\begin{equation}\label{closureB}\begin{minipage}{5.8in}{
  If $t \leadsto t'$ and $t\aconc v$, then $t'\aconc v$.}
\end{minipage}\vspace{-5ex}\end{equation}

\begin{equation}\label{closureC}\begin{minipage}{5.8in}{
  If $\chi \mathbin{\aconc} v$ with $\source(\chi)\mathbin=\source(v)$
  then $\exists \chi'\mathbin\in\Tr^{s\bullet}$ with $\source(\chi')\mathbin=\target(v)$ and $\chi \leadsto \chi'\!$.}
\end{minipage}\vspace{-5ex}\end{equation}

\begin{equation}\label{closureD}\begin{minipage}{5.8in}{
  If $\chi\leadsto \chi'$ then $\ell(\chi')=\ell(\chi)$ and
  when moreover $t\in\Tr^\bullet$ (that is, $\ell(t)\in Act{\setminus}\Rec$)
  then $\chi \naconc \chi'$.}
\end{minipage}\pagebreak\end{equation}

\noindent
Here  $\source(t):=P$ and $\ell(t)=\bar s$ if $t$ is the derivation of a formula $P\sigar{s}$.
With $\chi \aconc v$ \we mean that the possible occurrence of
$\chi$ is unaffected by the occurrence of $v$.
Although for CCS the relation $\aconc$ is symmetric (and $\Tr^{s\bullet}=\Tr$), for ABC and CCSS it is not:

\begin{exampleA}{\cite{GH15a}}{concurrency broadcast}
Let $P$ be the process $b!|(b?+c)$, and let $\chi$ and
$v$ be the derivations of the $b!$- and $c$-transitions of $P$.
The broadcast $b!$ is in our view completely under the control of the left component; it will occur
regardless of whether the right component listens to it or not. It so happens that if $b!$ occurs in
state $P$, the right component will  listen to it, thereby disabling the possible occurrence of $c$.
For this reason \we have $\chi \aconc v$ but $v \naconc\chi$.
\end{exampleA}

\begin{example}{signalling}
Let $P$ be the process $a\signals s | s$, and let $\chi$ and
$v$ be the derivations of the $a$- and $\tau$-transitions of $P$.\linebreak[3]
The occurrence of $t$ disrupts the emission of the signal $s$, thereby disabling the $\tau$-transition.
However, reading the signal does not affect the possible occurrence of $t$.
For this reason, $\chi \aconc v$ but $v \naconc\chi$.
\end{example}

\begin{proposition}{closureSquig}
  Assume (\ref{closureA})--(\ref{closureC}).
  If $\chi\in\Tr^{s\bullet}$ and $\pi$ is a path from $\source(t)$ to $P\in \cT$ such that $\chi \aconc v$ for
  all transitions $v$ occurring in $\pi$, then there is a $t'\in\Tr^{s\bullet}$ such that $\source(t')=P$ and
  $\chi \leadsto \chi'$.
\end{proposition}
\begin{proof}
  By induction on the length of $\pi$.

  The induction base is trivial, taking $t':=t$, and applying the reflexivity of $\leadsto$.

  So assume $\chi\in\Tr^{s\bullet}$ and $\pi$ is a path from $\source(\chi)$ with as last transition $v'$ with
  $\source(v')=P$ and $\target(v')=Q$, such that $\chi \aconc v$ for all transitions $v$ occurring in $\pi$.
  By induction, there is a $t'\in\Tr^\bullet$ such that $\source(t')=P$ and $\chi \leadsto \chi'$.
  By (\ref{closureB}) $t'\aconc v'$. By (\ref{closureC})
  there is a $t''\in\Tr^{s\bullet}$ such that $\source(t'')=Q$ and $t' \leadsto t''$.
  Now apply the transitivity of $\leadsto$.
\end{proof}

\begin{corollary}{closureSquig}
  Assume (\ref{closureA})--(\ref{closureD}).
  If $\chi\in\Tr^\bullet$ and $\pi$ is a path from $\source(t)$ to $P\in \cT$ such that $\chi \aconc v$ for
  all transitions $v$ occurring in $\pi$, then there is a $t'\in\Tr^\bullet$ such that $\source(t')=P$,
  $\ell(\chi')=\ell(\chi)$ and $\chi \naconc \chi'$.
\end{corollary}
\begin{proof}
Immediately from \pr{closureSquig} and (\ref{closureD}).
Note that $t\in\Tr^\bullet \wedge \ell(t')=\ell(t)$ implies that $t'\in\Tr^\bullet$.
\end{proof}
It follows that the LTS $(\cT,\Tr,\source,\target,\ell)$, augmented with the
concurrency relation $\aconc$ restricted to $\Tr^\bullet \times \Tr$, is an LTSC in the sense of \df{LTSC}.
By (\ref{closureA}) and (\ref{closureD}) $\aconc$ is irreflexive on $\Tr^\bullet$, so property
(\ref{irreflexivity}) holds. That (\ref{closure}) holds is stated by \cor{closureSquig}.

\We now proceed to define the relations $\leadsto$ and $\aconc$ on synchrons, and then lift
them to derivations. Subsequently, \we establish (\ref{closureA})--(\ref{closureD}).

\begin{definition}{dynamic}
The elements $+_L$, $+_R$, $A{:}$ and $\mbox{}\signals r$ of {\it Arg} are called \emph{dynamic}
\cite{Mi80}; the others are \emph{static}.  (Static operators stay around when their arguments perform transitions.)
For $\sigma \in \mathit{Arg}^*$ let $\mathit{static}(\sigma)$ be the result of removing all dynamic elements from $\sigma$.
Moreover, for $\varsigma=\sigma\upsilon$ with \plat{$\upsilon\in\{(\shar{\alpha}P),(P\sig s),(b{:})\}$}
let $\mathit{static}(\varsigma):=\mathit{static}(\sigma)\upsilon$.
\end{definition}

\begin{definition}{possible successor synchron}
A synchron $\varsigma'$ is a \emph{possible successor} of a synchron $\varsigma$, notation
$\varsigma \leadsto \varsigma'$, if either $\varsigma'=\varsigma$, or $\varsigma$ has the form $\sigma_1|_D\varsigma_2$ for some
$\sigma_1\in {\it Arg}^*$, $D\in\{L,R\}$ and synchron $\varsigma_2$, and $\varsigma'=\mathit{static}(\sigma_1)|_D\varsigma_2$.
\end{definition}

\begin{definition}{concurrent synchrons}
Two synchrons $\varsigma$ and $\upsilon$ are \emph{directly concurrent}, notation $\varsigma \smile_d \upsilon$,
if $\varsigma$ has the form $\sigma_1 |_D \varsigma_2$ and $\upsilon= \sigma_1 |_E \upsilon_2$ with $\{D,E\}=\{L,R\}$.
Two synchrons $\varsigma'$ and $\upsilon'$ are \emph{concurrent}, notation $\varsigma'\! \smile \upsilon'\!$, if
$\exists\varsigma,\!\upsilon.\, \varsigma'\!\mathrel{\reflectbox{$\leadsto$}} \varsigma\conc_d \upsilon\leadsto\upsilon'\!$.
\end{definition}

\begin{lemma}{direct concurrency}
  If $\varsigma,\upsilon \in \varsigma(P)$ for some $P\in\cT$ and $\varsigma \smile \upsilon$,
  then $\varsigma \smile_d \upsilon$.
\end{lemma}
\begin{proof}
  By assumption, there are $\varsigma^\dagger,\upsilon^\dagger$ with
  $\varsigma \mathrel{\reflectbox{$\leadsto$}} \varsigma^\dagger \conc_d \upsilon^\dagger\leadsto\upsilon$.
  W.l.o.g.\ \we choose $\varsigma^\dagger$ and $\upsilon^\dagger$ such that either $\varsigma^\dagger=\varsigma$
  or $\upsilon^\dagger=\upsilon$.
  The synchrons $\varsigma$ and $\upsilon$ describe paths in the parse tree of $P$, so the first
  symbol where they differ must be $\textit{op}_\Left$ versus $\textit{op}_\R$ for
  some binary operator {\it op}.
  The possibility that ${\it op}={+}$ quickly leads to a contradiction, so $\varsigma \smile_d \upsilon$.
\end{proof}

\paragraph{Necessary and active synchrons}
All synchrons of the form \plat{$\sigma(\shar{\alpha}P)$} are \emph{active};
their execution causes a transition \plat{$\alpha.P \mathbin{\goto{\alpha}} P$} in the relevant component of the
represented system. Synchrons $\sigma(P\sig s)$ and $\sigma(b{:})$ are passive; they are not causing any state change.
Let $a\varsigma(t)$ denote the set of active synchrons of a derivation $t$.
It follows that $\Tr^\circ$ (see \Sec{fairness}) is the set of transitions $t\in\Tr$ with $a\varsigma(t)\neq\emptyset$.

Whether a synchron $\varsigma\in\varsigma(t)$ of a transition or {\signal} $t$
is \emph{necessary} for $t$ to occur is defined only for $t\in\Tr^{s\bullet}$.
If $\chi$ is the derivation of a broadcast transition, i.e., $\ell(\chi)=b!$ for some $b\in\B$,
then exactly one synchron $\upsilon\in\varsigma(\chi)$ is of the form \plat{$\sigma(\shar{b!}P)$},
while all the other $\varsigma\in\varsigma(\chi)$ are of the form \plat{$\sigma'(\shar{b?}Q)$} (or
possibly \plat{$\sigma'(b{:})$} in ABCd).
Only the synchron $\upsilon$ is necessary for the broadcast to occur, as a broadcast is unaffected
by whether or not someone listens to it. Hence \we define $n\varsigma(\chi):=\{\upsilon\}$.
For all other $\chi\in\Tr^{s\bullet}$
\we set $n\varsigma(\chi):=\varsigma(\chi)$, thereby declaring all synchrons of the derivation necessary.

\begin{lemma}{synchrons of same derivation}
  If $\chi\in\Tr^{s\bullet}$ and $\varsigma,\upsilon \in n\varsigma(\chi)\dcup a\varsigma(\chi)$ with $\varsigma\neq\upsilon$, then $\varsigma \smile_d \upsilon$.
\end{lemma}
\begin{proof}
A trivial structural induction on $\chi$.
\end{proof}

\begin{lemma}{future synchrons not concurrent}
If $\varsigma \leadsto \varsigma'$ and $\varsigma \leadsto \varsigma''$ then $\varsigma' \not\smile \varsigma''$.
Also, if $\varsigma \leadsto \varsigma''$ and $\varsigma' \leadsto \varsigma''$ then $\varsigma \not\smile \varsigma'$.
\end{lemma}

\begin{proof}
Note that $\varsigma \leadsto \varsigma'$ implies $\mathit{static}(\varsigma)=\mathit{static}(\varsigma')$,
and $\varsigma \conc_d \upsilon$ implies  $\mathit{static}(\varsigma) \conc_d \mathit{static}(\upsilon)$.
Hence $\varsigma \conc \upsilon$ implies  $\mathit{static}(\varsigma) \conc_d \mathit{static}(\upsilon)$.
Moreover, $\smile_d$ (and hence $\smile$) is irreflexive: $\varsigma \not\smile_d \varsigma$.

Suppose $\varsigma \leadsto \varsigma'$ and $\varsigma \leadsto \varsigma''$. Then
$\mathit{static}(\varsigma') = \mathit{static}(\varsigma) = \mathit{static}(\varsigma'')$.
So $\mathit{static}(\varsigma') \not\smile_d\mathit{static}(\varsigma'')$, and hence $\varsigma' \not\smile \varsigma''$.
The other statement follows in the same way.
\end{proof}

\begin{definition}{possible successor}
A derivation $\chi'\in\Tr^{s\bullet}$ is a \emph{possible successor} of a derivation $\chi\in\Tr^{s\bullet}$, notation
$\chi \leadsto \chi'$, if $\chi$ and $\chi'$ have equally many necessary synchrons and
each necessary synchron of $\chi'$ is a possible successor of one of $\chi$; i.e., if
$|n\varsigma(\chi)| = |n\varsigma(\chi')|$ and
$\forall \varsigma'\mathbin\in n\varsigma(\chi').\,\exists \varsigma\mathbin\in n\varsigma(\chi).\,
 \varsigma \leadsto \varsigma'$.
\end{definition}
By Lemmas~\ref{lem:synchrons of same derivation} and~\ref{lem:future synchrons not concurrent}
this implies that the relation $\leadsto$ between the necessary synchrons of $\chi$ and $\chi'$
is a bijection.

\begin{definition}{noninterfering derivations}
Derivation $\chi\mathbin\in\Tr^{s\bullet}$ is \emph{unaffected by} $\zeta\mathbin\in\Tr$, notation $\chi \aconc \zeta$,
if $\forall \varsigma\mathbin\in n\varsigma(\chi).\, \forall\upsilon\mathbin\in a\varsigma(\zeta).\, \varsigma \smile \upsilon$.
\end{definition}
So $\chi$ is unaffected by $\zeta$ if no active synchron of $\zeta$ interferes with a necessary synchron of $\chi$.
Passive synchrons do not interfere at all.

In \ex{concurrent} one has $\chi_d \smile \chi_e$, $\chi_d \leadsto \chi'_d$ and $\chi'_d \smile \chi_e$.
Here $t \conc u$ denotes $t \aconc u \wedge u\aconc t$.

\begin{proposition}{preorder}
  The relation $\leadsto$ on $\Tr^{s\bullet}$ is reflexive and transitive.
  On $\Tr^\bullet \times\Tr^{\bullet}$ it is disjoint with $\aconc$.
\end{proposition}
\begin{proof}
  The relation $\leadsto$ on synchrons is reflexive and transitive by definition.
  That it is disjoint with $\smile$ follows as in the proof of \lem{future synchrons not concurrent}.
  The lifting of these properties to derivations follows directly from the definitions, using that
  $n\varsigma(u)\cap a\varsigma(u)\neq\emptyset$ for all $u\in\Tr^\bullet$.
\end{proof}

\begin{proposition}{inherited}
If $t \leadsto t'$ and $t\aconc v$, then $t'\aconc v$.
\end{proposition}
\begin{proof}
Let $\varsigma'\in n\varsigma(t')$ and $\upsilon\in a\varsigma(v)$.
\We have to show that $\varsigma' \conc \upsilon$.
By assumption $\exists \varsigma\mathbin\in n\varsigma(\chi).\, \varsigma \leadsto \varsigma'$.
Furthermore, $\varsigma \conc \upsilon$, since $t\aconc v$. So
$\exists\varsigma^\dagger,\upsilon^\dagger.\, \varsigma\mathrel{\reflectbox{$\leadsto$}} \varsigma^\dagger\conc_d \upsilon^\dagger\leadsto\upsilon$.
By the transitivity of $\leadsto$ this entails
$\varsigma'\mathrel{\reflectbox{$\leadsto$}} \varsigma^\dagger\conc_d \upsilon^\dagger\leadsto\upsilon$,
so $\varsigma' \conc \upsilon$.
\end{proof}

\begin{proposition}{inherited label}
  If $\chi\leadsto \chi'$ then $\ell(\chi')=\ell(\chi)$.
\end{proposition}
\begin{proof}
If $\chi\leadsto \chi'$ then the derivation $\chi'$ must be obtainable from $\chi$ by reducing some
subterms of the form $\zeta+Q$, $P+\zeta$, $A{:}\zeta$ or $\zeta\signals r$ to $\zeta$, and/or
changing some receptive or discarding partners in a broadcast communication.
Given rules \myref{Sum-l}, \myref{Sum-r}, \myref{Rec}, etc., this does not alter the label of this derivation.
\end{proof}
Propositions~\ref{pr:preorder}--\ref{pr:inherited label} establish the required properties
(\ref{closureA},\ref{closureB},\ref{closureD}).
It remains to establish (\ref{closureC})---see \pr{closure PA}.

%%%%%%%%%%%%%%%%%%%%%%%%%%%%%%%%%%%%%%%%%%%%%%%%%%%%%%%%%%%%%%%%%%%%%%%%%%%%%%%%%%%%%%%%%%%%%%%%%%%%%%%%%%%%%

\begin{definition}{consistency}
A set $\Sigma \mathbin\subseteq \varsigma(P)$ of synchrons of $P$ is \emph{$P$-consistent} if there
is a derivation $\chi\in\Tr^{s\bullet}$ with $n\varsigma(\chi)=\Sigma$ and $\source(\chi)=P$.
\end{definition}

\begin{lemma}{consistency}
  Let $P,Q\in \cT$ be processes, $\Sigma\subseteq\varsigma(P)$ and
  $\Sigma'\subseteq\varsigma(Q)$ such that $\Sigma$ is $P$-consistent, $|\Sigma'| = |\Sigma|$ and
  $\forall \varsigma\mathbin\in\Sigma.\,\exists \varsigma'\mathbin\in\Sigma'.\,\varsigma \mathbin\leadsto \varsigma'\!$.
  Then $\Sigma'$ is $Q$-consistent.
\end{lemma}
\begin{proof}
Let $t\in\Tr^{s\bullet}$ be such that $n \varsigma(\chi)\mathbin=\Sigma$ and $\source(\chi)=P$. So $\Sigma\neq\emptyset$.
Each synchron $\varsigma\in \Sigma$ represents a path in the derivation $\chi$ from its root to a leaf.
If $\varsigma \leadsto \varsigma'$ then $\varsigma'$ represents a version of the same path, but in
which certain nodes of $\chi$, labelled $+Q'$, $P'+$, $A{:}$ or ${}\signals r$, are marked as being deleted by $\varsigma'$.
Since $\Sigma'\subseteq\varsigma(Q)$, this marking of deleted nodes is consistent, in the the sense
that no node of $t$ is deleted according to one element of $\Sigma'$, but kept according to another.
In fact, whether a node of $t$ is marked as deleted depends entirely on the syntactic shape of $Q$.
In view of the rules \myref{Sum-l}, \myref{Sum-r}, \myref{Rec}, etc., actually deleting the indicated
nodes from the derivation $t$ yields another derivation $t'$, with $n\varsigma(t')=\Sigma'$ and $\source(t')=Q$.
Depending on the syntactic shapes of $P$ and $Q$, some receptive or discarding partners in broadcast
communications may have been altered between $t$ and $t'$ as well.
\end{proof}

\noindent
Write $\varsigma\aconc_d \zeta$ for $\varsigma$ a synchron and $\zeta \in \Tr$ if
$\varsigma \conc_d \upsilon$ for all $\upsilon \in a\varsigma(\zeta)$.

\begin{definition}{after}
  Let $\varsigma,\upsilon$ be synchrons with $\varsigma \smile_d \upsilon$, i.e.,
  $\varsigma\mathbin=\sigma |_D \varsigma'$ and $\upsilon \mathbin= \sigma |_E \upsilon'$
  for some $\sigma \in {\it Arg}^*$, synchrons $\varsigma',\upsilon'$ and $\{D,E\}=\{L,R\}$.
  Define $\varsigma \MVAt \upsilon$, where $\MVAt$ is pronounced ``after'', to be
  $\mathit{static}(\sigma)|_D\varsigma'$.

  For $\zeta\in\Tr^\circ$ with $\varsigma \aconc_d \zeta$ let
  $\varsigma\MVAt\zeta := \varsigma\MVAt\upsilon$ for the $\upsilon \in a\varsigma(\zeta)$
  that is ``closest'' to $\varsigma$, in the sense that it has the largest prefix in common with it.
  For $\zeta \in\Tr\setminus\Tr^\circ$ let $\varsigma\MVAt\zeta := \varsigma$.
\end{definition}
In \ex{concurrent} $\backslash c\, {|_L}\, {+_R}\, {|_L} (\shar{d}R) \MVAt \chi_e = \backslash c\, {|_L}\, {|_L} (\shar{d}R)$.

\begin{observation}{after}
  Let $\zeta\in\Tr$, $P=\source(\zeta)$ and $Q= \target(\zeta)$.
  If $\varsigma \in\varsigma(P)$ with $\varsigma \aconc_d \zeta$
  then $\varsigma \MVAt \zeta \in \varsigma(Q)$.
\end{observation}

\begin{proposition}{closure PA}
  If $\chi\mathbin\in\Tr^{s\bullet}$, $v\mathbin\in\Tr$ and $\chi \aconc v$ with
  $\source(\chi)=\source(v)$ then there is a derivation $\chi'\in\Tr^{s\bullet}$
  with $\source(\chi')=\target(v)$ and $\chi \leadsto \chi'$.
\end{proposition}
\begin{proof}
  Let $P:=\source(v)$ and $Q:= \target(v)$.
  Then $\Sigma:=n\varsigma(\chi)$ is $P$-consistent.
  By \lem{direct concurrency} $\varsigma \aconc_d v$ for all $\varsigma \in \Sigma$.
  Let $\Sigma':=\{\varsigma\MVAt v \mid \varsigma \in \Sigma\}$.
  By \Obs{after} $\Sigma' \subseteq \varsigma(Q)$.
  By \df{after} $\varsigma \leadsto \varsigma\MVAt v$ for all $\varsigma \in \Sigma$.
  By \lem{direct concurrency} $\varsigma \aconc_d \upsilon$ for all $\varsigma,\upsilon \in \Sigma$ and thus
  by \lem{future synchrons not concurrent} $|\Sigma'| = |\Sigma|$.
  So by \lem{consistency} $\Sigma'$ is $Q$-consistent, that is, there exists a $\chi'\mathbin\in\Tr^{s\bullet}$
  with $n\varsigma(\chi')\mathbin=\Sigma'$ and $\source(\chi')\mathbin=Q$.
  By \df{possible successor} $\chi\leadsto \chi'\hspace{-2pt}$.%
\end{proof}

\section{Components}\label{sec:components}

This section proposes two concepts of system components associated to a transition, with for each a
classification of components as necessary and/or affected. \We then apply a definition of a
concurrency relation in terms of these components closely mirroring \df{noninterfering derivations}
in \Sec{LTSC} of the concurrency relation $\aconc$ in terms of synchrons.
The \emph{dynamic components} give rise to the exact same concurrency relation $\aconc$ from
\df{noninterfering derivations}, whereas the \emph{static components} yield a strictly smaller
concurrency relation $\aconc_s$. However, only the static components satisfy closure property
(\ref{closureComp}).
Finally, \we present three alternative versions of $\aconc_s$ that all give rise to the same concept
of justness.

\subsection{Dynamic components}\label{sec:dynamic}
\newcommand\dynamic{dynamic }

A \emph{(dynamic) component} is either the empty string $\varepsilon$ or any string $\sigma\iota$ with $\sigma\in{\it Arg}^*$
and $\iota\in{\it Arg}$ a static argument. Each synchron $\varsigma$ can be uniquely written as
$\gamma\varsigma'$ with $\gamma$ a component and $\varsigma'$ a synchron with only dynamic arguments.
The \emph{\dynamic component} $C(\varsigma)$ of such a synchron $\varsigma$ is defined to be $\gamma$.

The set of \dynamic components $\textit{COMP}(P)$ of a process $P$ is defined as $\{C(\varsigma)\mid\varsigma\in\varsigma(P)\}$.

The set of \dynamic components $\textit{COMP}(\chi)$ of a derivation $\chi$ is defined as $\{C(\varsigma)\mid\varsigma\in\varsigma(\chi)\}$.

The set of \emph{necessary} \dynamic components $\NC(\chi)$ of a derivation $\chi$ is defined as $\{C(\varsigma)\mid\varsigma\in n\varsigma(\chi)\}$.

The set of \emph{affected} \dynamic components $AC(\chi)$ of a derivation $\chi$ is defined as $\{C(\varsigma)\mid\varsigma\in a\varsigma(\chi)\}$.

A component $\gamma'$ is a \emph{possible successor} of a component $\gamma$, notation
$\gamma \leadsto \gamma'$, if either $\gamma'=\gamma$ or $\gamma$ has the form $\sigma_1|_D\gamma_2$, with
$\sigma_1\in {\it Arg}^*$, $D\in\{L,R\}$ and $\gamma_2$ a component, and $\gamma'=\mathit{static}(\sigma_1)|_D\gamma_2$.

Two components $\gamma$ and $\delta$ are \emph{directly concurrent}, notation $\gamma \smile_d \delta$,
if $\gamma = \sigma_1 |_D \gamma_2$ and $\delta= \sigma_1 |_E \delta_2$ with $\{D,E\}=\{L,R\}$.
Two components $\gamma'$ and $\delta'$ are \emph{concurrent}, notation $\gamma' \smile \delta'$, if
$\exists\gamma,\delta.\, \gamma'\mathrel{\reflectbox{$\leadsto$}} \gamma\conc_d \delta\leadsto\delta'\!$.

These definitions imply that $\varsigma\leadsto\upsilon \Rightarrow C(\varsigma)\leadsto C(\upsilon)$
and $\varsigma\smile\upsilon \Leftrightarrow C(\varsigma)\smile C(\upsilon)$.

The next lemma, whose proof is trivial, say that the concurrency relation $\aconc$ on derivations
can be equally well defined in terms of \dynamic components (rather than synchrons).

\begin{lemma}{concurrent derivations}
Derivation $\chi\in\Tr^{s\bullet}$ is unaffected by $\zeta\in\Tr$, i.e., $\chi \aconc \zeta$,
iff $\forall \gamma\mathbin\in \NC(\chi).\, \forall\delta\mathbin\in AC(\zeta).\, \gamma \smile \delta$.
\hfill$\Box$
\end{lemma}
The following shows that the functions $\NC$ and $AC$ do not satisfy closure property (\ref{closureComp}) of \Sec{justness}.
\begin{example}{noClosureComp}
In \ex{concurrent}, $\chi_d,\chi_e\in\Tr^\bullet$ with $\source(\chi_d)=\source(\chi_e)$ and $\NC(\chi_d)\cap AC(\chi_e)=\emptyset$.
   Yet, there is no $u\in \Tr^\bullet$ with $\source(u)=\target(\chi_e)$, $\ell(u)=\ell(\chi_d)=d$ and $\NC(u)=\NC(\chi_d)$.
   In fact, the unique $u\in \Tr^\bullet$ with $\source(u)=\target(\chi_e)$ and $\ell(u)=d$ is $\chi'_d$.
   However, $\NC(\chi_d)= \{\backslash c\, {|_L}\, {+_R}\, {|_L}\}$, whereas $\NC(\chi'_d)= \{\backslash c\, {|_L}\, {|_L}\}$.
\end{example}

\subsection{Static components}\label{sec:static}

A \emph{static component} is a string $\sigma\in{\it Arg}^*$ of static arguments.
Let $\C$ be the set of static components.
The \emph{static component} $c(\varsigma)$ of a synchron $\varsigma$ is defined to be the largest
prefix $\gamma$ of $\varsigma$ that is a static component.

The set of static components $\comp(P)$ of a process $P$ is defined as $\{c(\varsigma)\mid\varsigma\in\varsigma(P)\}$.

The set of static components $\comp(\chi)$ of a derivation $\chi$ is defined as $\{c(\varsigma)\mid\varsigma\in\varsigma(\chi)\}$.

The set of \emph{necessary} static components $\npc(\chi)$ of a derivation $\chi$ is defined as $\{c(\varsigma)\mid\varsigma\in n\varsigma(\chi)\}$.

The set of \emph{affected} static components $\afc(\chi)$ of a derivation $\chi$ is defined as $\{c(\varsigma)\mid\varsigma\in a\varsigma(\chi)\}$.

Since $n\varsigma(\chi) \subseteq \varsigma(\chi)$ and $a\varsigma(\chi) \subseteq \varsigma(\chi)$,
\we have $\npc(\chi) \subseteq \comp(\chi)$ and $\afc(\chi) \subseteq \comp(\chi)$. Moreover, by
\lem{synchrons}, $\comp(\chi)\subseteq \comp(\source(\chi))$.

The following lemma shows how the relations $\leadsto$ and $\smile$ from \Sec{dynamic} simplify when
applied to static components.%
\begin{lemma}{static leadsto}
If $\gamma\in\C$ and $\gamma\leadsto\gamma'$ then $\gamma'=\gamma$.
Moreover, for $\gamma,\delta\in\C$, $\gamma \smile \delta$ iff $\gamma \smile_d \delta$.
\end{lemma}
\begin{proof}
The first statement and direction ``if'' of the second are trivial.
So let $\gamma',\delta'\in\C$ with $\gamma' \smile \delta'$. Then
$\gamma'\mathrel{\reflectbox{$\leadsto$}} \gamma\conc_d \delta\leadsto\delta'$ for some components $\gamma$ and $\delta$.
Thus, using insights from the proof of \lem{future synchrons not concurrent},
$\gamma'=\mathit{static}(\gamma')=\mathit{static}(\gamma)\conc_d \mathit{static}(\delta)=\mathit{static}(\delta')=\delta'$.
\end{proof}
% Trivially, $\smile_d$, and hence $\smile$, is irreflexive on $\C$.
The next result says that any two different static components of the same process are concurrent.
\begin{lemma}{components concurrent}
Let $\gamma,\delta\in\comp(P)$ for some $P\in\cT$. Then $\gamma \conc \delta$ iff $\gamma\neq \delta$.
\end{lemma}
\begin{proof}
``Only if'' is trivial. ``If'' follows by a straightforward structural induction on $P$.
\end{proof}

\noindent
\We now define a static concurrency relation $\aconc_s$ between derivations in terms of their static components in the same way that
the (dynamic) concurrency relation $\aconc$ is characterised (by \lem{concurrent derivations}) in terms of their dynamic components:

\begin{definition}{static concurrency}
Derivation $\chi\in\Tr^{s\bullet}$ is \emph{statically unaffected by} $\zeta$, $\chi \aconc_s \zeta$,
iff $\forall \gamma\mathbin\in \npc(\chi).\, \forall\delta\mathbin\in \afc(\zeta).\, \gamma \smile \delta$.%
\end{definition}
The following shows that $\aconc_s$ is strictly contained in $\aconc$.

\begin{proposition}{static concurrency}
If $\chi \aconc_s \zeta$ then $\chi \aconc \zeta$.
\end{proposition}
\begin{proof}
Suppose $\chi \aconc_s \zeta$. Let $\varsigma\mathbin\in n\varsigma(\chi)$ and $\upsilon\mathbin\in a\varsigma(\zeta)$.
Then $c(\varsigma)\mathbin\in\npc(\chi)$, $c(\upsilon)\mathbin\in\afc(\zeta)$, so $c(\varsigma) \smile c(\upsilon)$.
Hence $c(\varsigma) \smile_d c(\upsilon)$ by \lem{static leadsto}, and thus $\varsigma \smile_d \upsilon$.
\end{proof}
In \ex{concurrent} \we have $t_d \smile t_e$ but $t_d \nconc_s t_e$, for $\npc(t_e)=\comp(t_e)=\comp(t_d)=\afc(t_d)=\{\backslash c\, {|_L}\}$.
Here $t \conc_s u$ denotes $t \aconc_s u \wedge u\aconc_s t$.
Hence the implication of \pr{static concurrency} is strict.

\begin{lemma}{static concurrency}
Let $\source(\chi)=\source(v)$. Then $\chi\aconc_s v$ iff $\npc(\chi) \cap \afc(v) = \emptyset$.
\end{lemma}
\begin{proof}
Immediately from \lem{components concurrent}.
\end{proof}

\noindent
Write $\varsigma\aconc_s v$ for $\varsigma$ a synchron and $v \in \Tr$ if
$c(\varsigma) \conc \gamma$ for all $\gamma \in \afc(v)$.

\begin{observation}{after static}
  If $\varsigma \aconc_s v$ then $\varsigma\MVAt v = \varsigma$.
\end{observation}
Henceforth \we write $t \equiv u$, for $t,u\in\Tr^{s\bullet}$,
when $n\varsigma(t)=n\varsigma(u)$. In that case $\npc(t)=\npc(u)$ and $t\leadsto u$, and thus, by (\ref{closureD}), $\ell(t)=\ell(u)$.

\begin{proposition}{closure static}
  If $\chi\mathbin\in\Tr^{s\bullet}$, $v \in \Tr$ and $\chi \aconc_s v$ with $\source(\chi)=\source(v)$
  then there is a derivation $u\in\Tr^{s\bullet}$ with $\source(u)=\target(v)$ and $\chi \equiv u$.
\end{proposition}
\begin{proof}
  By \pr{static concurrency} $\chi \aconc v$. Hence the proof of \pr{closure PA}
  finds a $u\in\Tr^{s\bullet}$ with $\source(u)=\target(v)$ and $\chi \leadsto u$,
  such that $n\varsigma(u)=\{\varsigma\MVAt v \mid \varsigma \in n\varsigma(\chi)\}$.
  Since $\chi \aconc_s v$, $\varsigma\aconc_s v$ for all $\varsigma \in n\varsigma(\chi)$,
  so by \Obs{after static} $n\varsigma(u)=n\varsigma(\chi)$, i.e., $\chi \equiv u$.
\end{proof}
In view of \lem{static concurrency}, \pr{closure static} says that the functions
$\npc$ and $\afc:\Tr\rightarrow\Pow(\C)$ satisfy closure property (\ref{closureComp}) of \Sec{justness}.

\subsection{Another compatible definition of the static concurrency relation}\label{sec:compatible}

The concurrency relation $\aconc_c$ between transitions defined in terms of static components
according to the template in \cite{GH19}, recalled in \Sec{justness}, is not identical to the
concurrency relation $\aconc_s$ of \df{static concurrency} when restricted to $\Tr^\bullet \times \Tr$.

\begin{definition}{alternative static concurrency}
Let $\chi\in\Tr^{s\bullet}$ and $\zeta\in\Tr$ be derivations. Write $\chi \aconc_{c} \zeta$ iff $\npc(t) \cap \afc(u) = \emptyset$.
\end{definition}
The following shows that $\aconc_{s}$ is strictly included in $\aconc_{c}$.

\begin{proposition}{alternative static concurrency}
If $\chi \aconc_s \zeta$ then $\chi \aconc_c \zeta$.
\end{proposition}
\begin{proof}
This follows immediately from the irreflexivity of ${\smile} \subseteq \C \times \C$ (\lem{components concurrent}).
\end{proof}

\begin{example}{alternative static concurrency}
  Let $t_1$ and $t_2$ be the unique derivations of the transitions
  $c[f]\goto{c} {\bf 0}[f]$ and $c\backslash L \goto{c} {\bf 0}\backslash L$,
  where $f(c)=c$ and $c\notin L$. Then $(n)\varsigma(t_1)=\{[f](\shar{c}{\bf 0})\}$ and
  $(a)\varsigma(t_2)=\{\backslash L (\shar{c}{\bf 0})\}$, so $\npc(t_1)=\{[f]\}$ and
  $\afc(t_2)=\{\backslash L\}$. Since $[f] \nconc \backslash L$ and $[f] \neq \backslash L$,
  one has $t_1 \naconc_s t_2$ (and $t_1 \naconc t_2$) but $t_1 \aconc_c t_2$.
\end{example}
Since in \ex{concurrent} \we have $t_d \smile t_e$ but ($t_d \nconc_s t_e$ and) $t_d \nconc_c t_e$,
since $\npc(t_d)=\afc(t_e)=\{\backslash c\, {|_L}\}$, it follows also that $\aconc_c$ is incomparable with $\aconc$.

Nevertheless, \we show that for the study of justness it makes no difference whether justness is
defined using the concurrency relation $\aconc_s$ or $\aconc_c$.

\begin{lemma}{closure static}
  If $\chi\mathbin\in\Tr^{s\bullet}$ and $\pi$ is a path from $\source(\chi)$ to a state $P'$
  such that $\chi \aconc_s v$ for all transitions $v$ on $\pi$, then there is a derivation
  $t'\in\Tr^{s\bullet}$ with $\source(t')=P'$ and $\chi \equiv \chi'$.
\end{lemma}
\begin{proof}
  This is a corollary of \pr{closure static}, obtained by a simple induction on the length of $\pi$,
  using the reflexivity and transitivity of $\equiv$, and that $\chi \aconc_s v$ and $\chi \equiv \chi'$ implies
  $\chi' \aconc_s v$.
\end{proof}

\begin{definition}{parametrised justness}
Let $\IT=(S,\Tr,\source,\target,\ell)$ be an LTS, and ${\aconc_x}\subseteq\Tr\times\Tr$ a
concurrency relation between the transitions, satisfying (\ref{irreflexivity}) and (\ref{closure}).
Call a path $\pi$ in $\IT$ ${\aconc_x}$-$B$-just, for $\Rec \subseteq B\subseteq Act$, if according to
\df{justness} it is $B$-just in the LTSC $(S,\Tr,\source,\target,\ell,\aconc_x)$.
\end{definition}

\begin{proposition}{static justness}
A path is ${\aconc_c}$-$B$-just iff it is ${\aconc_s}$-$B$-just.
\end{proposition}
\begin{proof}
  ``Only if'' is immediate from \df{justness} and \pr{alternative static concurrency}.

  ``If'': Let $\pi$ be ${\aconc_s}$-$B$-just, let $\pi'$ be a suffix of $\pi$ starting in a state $P$,
  and let $t \mathbin\in \Tr^\bullet_{\neg B}$ with $\source(t)\mathbin=P$.
  By \df{justness} a transition $u$ with $t \naconc_s u$ occurs in $\pi'$.
  W.l.o.g.\ \we take $u$ to be the first such transition in $\pi'$.
  If suffices to show that $t \naconc_c u$.
  For all transitions $v$ in $\pi$ between $P$ and $\source(u)$ \we have $t \aconc_s v$.
  Hence, by \lem{closure static}, there is a $t'\in\Tr^{s\bullet}$ with $\source(t')=\source(u)$ and $t\equiv t'$.
  Moreover, $t \naconc_s u$ implies $t' \naconc_s u$, which implies $t' \naconc_c u$ by \lem{static concurrency},
  which implies $t \naconc_c u$.
\end{proof}

\noindent
The above proof shows that for the study of justness, \we need to know whether two transitions
$t\in\Tr^\bullet$ and $u\in\Tr$ are related by the static concurrency relation or not, only when
$\exists t'\in\Tr^{s\bullet}$ with $t\equiv t'$ and $\source(t')=\source(u)$. 
And restricted to such pairs $(t,u)$ the relations $\aconc_s$ and $\aconc_c$ coincide.

\subsection{A more abstract definition of static components}\label{sec:abstract}

The arguments $\iota\in{\it Arg}$ of unary operators occurring in the definition of a static component are essentially
redundant. Consider the following alternative definitions:

An \emph{abstract static component} is a string $\sigma\in\{|_L , |_R\}^*$.
The \emph{abstract static component} $c'(\varsigma)$ of a synchron $\varsigma$ is defined to be the
result of leaving out all argument $\backslash \RL$ and $[f]$ from $c(\varsigma)$.
For a derivation $t$ let
$\comp'(\chi) := \{c'(\varsigma)\mid\varsigma\in\varsigma(\chi)\}$,
$\npc'(\chi) := \{c'(\varsigma)\mid\varsigma\in n\varsigma(\chi)\}$ and
$\afc'(\chi) := \{c'(\varsigma)\mid\varsigma\in a\varsigma(\chi)\}$.

Analogously to Definitions~\ref{df:static concurrency} and~\ref{df:alternative static concurrency},
write $\chi \aconc'_s \zeta$ iff
$\forall \gamma\mathbin\in \npc'(\chi).\, \forall\delta\mathbin\in \afc'(\zeta).\, \gamma \smile \delta$,
and write $\chi \aconc'_c \zeta$ iff $\npc'(t) \cap \afc'(u) = \emptyset$.
Note that $c(\varsigma) \smile c(\upsilon)$ implies $c'(\varsigma) \smile c'(\upsilon)$
for all synchrons $\varsigma$ and $\upsilon$.
Likewise, $c'(\varsigma) \neq c'(\upsilon)$ implies $c(\varsigma) \neq c(\upsilon)$.
Hence ${\aconc_s} \subseteq {\aconc'_s} \subseteq {\aconc'_c} \subseteq {\aconc_c}$.
Thus, \pr{static justness} implies that
a path is ${\aconc_s}$-$B$-just iff it is ${\aconc'_s}$-$B$-just
iff it is ${\aconc'_c}$-$B$-just iff it is ${\aconc_c}$-$B$-just.

\section{Computational interpretations}\label{sec:computational}

The classical computational interpretation of CCS and related languages aligns with the (dynamic)
concurrency relation $\aconc$ of Sections~\ref{sec:LTSC} and~\ref{sec:dynamic}, rather than the
static concurrency relation $\aconc_s$ of \Sec{static}. This is illustrated by the transitions
$\chi_d$ and $\chi_e$ of \ex{concurrent}, which are generally regarded as concurrent.
This computational interpretation also aligns with the semantics of CCS in terms of event structures
and Petri nets, where concurrency is made more explicit \cite{Wi87a,GV87}.

Below, \we first define a sublanguage of CCS with broadcast communication and/or signals on which the
static and dynamic concurrency relations coincide---so it does not include the process $P$ of \ex{concurrent}.
Using this, \we propose an alternative computational interpretation of CCS and its extensions that
aligns with the static concurrency relation.

The underlying intuition is that each transition occurs at some location, concurrent transitions
occur at different locations, and a choice made through the $+$-operator necessarily needs to be
made locally, so that one can not
have two truly parallel actions $d$ and $e$ for which the execution of either one constitutes the
same choice. In \ex{concurrent}, the transitions $\chi_\tau$ and $\chi_d$ stem from opposite sides
of a $+$-operator, and therefore should be co-located. The same holds for 
$\chi_\tau$ and $\chi_e$, and consequently $\chi_d$ should be co-located with $\chi_e$ and not concurrent.
This intuition is loosely inspired by \cite{GGS08d,GGS13}. Similarly, a signal can be emitted only locally, and
for convenience \we treat recursion in the same vein, so that all dynamic operators can be applied to
sequential processes only.

\begin{definition}{dynamically sequential}
The \emph{dynamically sequential} fragment of CCS with broadcast communication and/or signals is
given by the context-free grammar
\[\begin{array}{l@{~::=~}l}
               S & {\bf 0} \mid \alpha.P \mid S+S \mid S\signals s \mid A  \mid S\backslash \RL \mid S[f]\\
               P & S       \mid P|Q \mid P\backslash\RL \mid P[f]
\end{array}\]
where $S$ is the sort of (initially) \emph{sequential processes}, and $P$ the sort of \emph{parallel processes}.\vspace{1pt}
Defining equations for agent identifiers should have the form \plat{$A \defis S$}.
\end{definition}
This language is crafted in such a way that in all synchrons a dynamic argument will never precede a
parallel composition argument $|_\Left$ or $|_\R$. As a consequence, \we obtain, for synchrons
$\varsigma$ and $\upsilon$, that $\varsigma \leadsto \upsilon \Leftrightarrow \varsigma=\upsilon$ and that\vspace{-1ex}
\[\varsigma \smile \upsilon   \Leftrightarrow
  \varsigma \smile_d \upsilon   \Leftrightarrow
  C(\varsigma) \smile_d C(\upsilon)   \Leftrightarrow
  c(\varsigma) \smile c(\upsilon)\;.
\]
Hence, on this fragment, the concurrency relations $\aconc$ and $\aconc_s$ coincide.

Next, \we introduce a new unary operator $sq$, that turns a parallel process into a sequential one.
Thus ``$\mid sq(P)$'' can be added to the line for $S ::=$ in the context-free grammar above.
Its operational rules are
\[\frac{P\goto{\alpha} P'}{sq(P) \goto{\alpha}P'}  \qquad \frac{P\goto{\kappa} P'}{sq(P) \goto{\kappa}sq(P')} \]
where $\kappa$ ranges over $b{:}\in\B{:}$ and $\bar s \in\bar\Sig\!$, so that it changes its argument as
little as possible. However, the argument $sq$ is now added to synchrons, counting as dynamic, and
\df{concurrent synchrons} of $\smile_d$ is upgraded by the requirement that the argument $sq$ does not occur
in $\sigma_1$, with $\smile$ redefined to equal $\smile_d$. Consequently, the only effect of $sq$
is that any concurrency between outgoing transitions of its arguments is removed. The process
$sq(d.R|e.S)$, for instance, behaves exactly like $d(R|e.S)+e(d.R|S)$.

On this extension of the dynamically sequential fragment of CCS with broadcast communication and/or
signals \we still have that $\varsigma \smile \upsilon   \Leftrightarrow
   \varsigma \smile_d \upsilon   \Leftrightarrow
  C(\varsigma) \smile_d C(\upsilon)   \Leftrightarrow
  c(\varsigma) \smile c(\upsilon)$, and consequently $\aconc$ and $\aconc_s$ coincide.

Finally, \we propose a language that has the same syntax as CCS, possibly extended with broadcast
communication and/or signals, but is technically a sublanguage of language proposed above,
because whenever the operators $+$ or $\mbox{}\signals s$ are applied to parallel arguments,
$P+Q$ is taken to be an abbreviation of $sq(P)+sq(Q)$, and $P\signals s$ of $sq(P)\signals s$.
Likewise, \plat{$A\defis P$} can be seen as an abbreviation of \plat{$A\defis sq(P)$}.
This language can be seen as an alternative computational interpretation of CCS (plus extensions)
that aligns with the static concurrency relation $\aconc_s$.

Interestingly, the operational Petri net semantics of \cite{DDM87} follows the static computational
interpretation above, whereas its modification in \cite{Old87,Old91} follows the classical (dynamic)
interpretation of concurrency.

\section{The dynamic and static accounts of justness agree}\label{sec:agree}

\We now show that the concurrency relations $\aconc$ and $\aconc_s$ (and thus also the variants
$\aconc'_s$, $\aconc_c$ and $\aconc'_c$ of $\aconc_s$ studied in Sections~\ref{sec:compatible}
and~\ref{sec:abstract}) give rise to the same concept of justness.

Each derivation $t\in\Tr$ has only finitely many synchrons, and each synchron contains finitely many
dynamic arguments. Let $d(t)$ be the sum, over $\varsigma\in n\varsigma(t)$, of the number of
dynamic arguments in $\varsigma$.

\begin{theorem}{static justness}
A path is ${\aconc}$-$B$-just iff it is ${\aconc_s}$-$B$-just.
\end{theorem}
\begin{proof}
  ``Only if'' is immediate from \df{justness} and \pr{static concurrency}.

  ``If'': Let $\pi$ be ${\aconc_s}$-$B$-just, let $\pi'$ be a suffix of $\pi$ starting in a state $P$,
  and let $t \mathbin\in \Tr^\bullet_{\neg B}$ with $\source(t)\mathbin=P$.
  By induction on $d(t)$ \we find a transition $u$ in $\pi'$ such that $t \naconc u$.

  By \df{justness} a transition $u'$ with $t \naconc_s u'$ occurs in $\pi'$.
  W.l.o.g.\ \we take $u'\in\Tr$ to be the first such transition in $\pi'$.
  By \lem{closure static} there is a derivation $t'\in\Tr^{s\bullet}$ with $\source(t')=source(u')$ and $t \equiv t'$.
  So $t' \naconc_s u'$. Hence there are $\varsigma\in n\varsigma(t')$ and $\upsilon\in a\varsigma(u')$
  with $c(\varsigma) = c(\upsilon)$ by \lem{static concurrency}.
  In case $t' \naconc u'$ then $t \naconc u'$ and \we are done. So suppose $t'\aconc u'$.
  Then $\varsigma \conc \upsilon$, so $\varsigma \conc_d \upsilon$ by Lemmas~\ref{lem:synchrons} and~\ref{lem:direct concurrency}.
  Thus $\varsigma$ has the form $\sigma_1 |_D \varsigma_2$ and $\upsilon= \sigma_1 |_E \upsilon_2$ with $\{D,E\}=\{L,R\}$.
  Since $c(\varsigma) = c(\upsilon)$, a dynamic operator must occur in $\sigma_1$.
  So by \df{after} $\varsigma\MVAt \upsilon$ contains fewer dynamic arguments than $\varsigma$,
  and hence $\varsigma\MVAt u'$ contains fewer dynamic arguments than $\varsigma$.
  (Here we use that $u'\in \Tr^\circ$, since $a\varsigma(u')\neq\emptyset$.)
  Moreover, for $\varsigma'\neq\varsigma$,  $\varsigma'\MVAt u'$ contains at most as many dynamic arguments as $\varsigma'$.

  The proof of \pr{closure PA}
  finds a $t''\in\Tr^{s\bullet}$ with $\source(t'')=\target(u')$ and $t' \leadsto t''$,
  such that $n\varsigma(t'')=\{\varsigma\MVAt u' \mid \varsigma \in n\varsigma(t')\}$.
  It follows that $d(t'')<d(t')$. By \pr{inherited label}, $\ell(t'')=\ell(t')=\ell(t)$, and so
  $t''\in\Tr^\bullet_{\neg B}$.
  By the induction hypothesis \we find a transition $u$ occurring in $\pi'$ past the (first) occurrence of
  $\source(t'')$, such that $t'' \naconc u$. Now $t \naconc u$, using \pr{inherited}.
\end{proof}

\section{Justness is feasible even with infinitary choice}\label{sec:feasibility2}

A straightforward induction of the length of derivations shows that for each process $P\in\cT$ in any of
the languages of \Sec{PA} there are only countably many derivations $t\in \Tr^\bullet$ with $\source(t)=P$.
Consequently, \thm{feasibility} says that, for any set $B\subseteq Act$ with $\Rec\subseteq B$,
$B$-justness is feasible.
However, the standard version of CCS \cite{Mi90b} features the infinitary choice operator $\sum_{i\in I}P_i$
for any index set $I$, which was omitted in \Sec{PA} (and many of the references). Its operational
rule is
\[\displaystyle\frac{P_j \goto{\alpha} P_j'}{\sum_{i\in I}P_i \goto{\alpha} P_j'}\makebox[0pt][l]{~~($j\in I$).}\]
The work reported here can be straightforwardly extended with this infinitary choice operator.
Instead of $+_\Left$ and $+_\R$ it gives rise to dynamic arguments $\sum^j$ appearing in synchrons.
But then \we have processes $P$ with uncountably many outgoing transitions, so that \thm{feasibility}
no longer applies. Nevertheless, $B$-justness is feasible, as follows from \cor{feasibility}, in
conjunction with \thm{static justness}.
For if $t\equiv t'$ (defined in \Sec{static}), then $t\aconc_s u \Leftrightarrow t'\aconc_s u$ for all
$u \in \Tr^\bullet$. As an $\equiv$-equivalence class is completely determined by a finite set of static
components, and the set $\C$ of static components is countable, so is the collection of 
$\equiv$-equivalence classes of transitions, and thus the set of equivalence classes used in \cor{feasibility}.

\section{An inductive characterisation of the concurrency relations \texorpdfstring{$\aconc_d$ and $\aconc_s$}{}}\label{sec:inductive}

As a variant of \df{noninterfering derivations} in \Sec{LTSC}, write $\chi \aconc_d \zeta$
if $\forall \varsigma\mathbin\in n\varsigma(\chi).\, \forall\upsilon\mathbin\in a\varsigma(\zeta).\, \varsigma \smile_d \upsilon$.
By Lemmas~\ref{lem:synchrons} and~\ref{lem:direct concurrency}, when
$\source(t)=\source(u)$ then $\chi \aconc \zeta \Leftrightarrow \chi \aconc_d \zeta$.

The idea of an asymmetric concurrency relation $\aconc$ is not new here.
A similar relation, here called $\acGH$, appeared in \cite{GH15a}.
That relation was defined only between derivations $t$ and $u$ with $\source(t)\mathbin=\source(u)$.
Here \we show that $\acGH$ agrees with our $\aconc$, in the sense that
\[t \acGH u \qquad\mbox{iff}\qquad t \aconc u ~\wedge~ \source(t)=\source(u)\]
for all $t\in\Tr^{\bullet}$ and $u\in\Tr$.
\cite{GH15a} dealt with ABC only, so there $\Tr^{s\bullet} = \Tr^\bullet$.
In order to prove this, \we give an inductive
characterisation of $\aconc_d$. This effort also yields  an inductive characterisation of
$\aconc_s$, which will be used in \Sec{coinductive} to provide a coinductive characterisation of
justness, in the spirit of the definitions of justness from \cite{TR13,GH15a,EPTCS255.2,Bou18}.

\begin{proposition}{concurrency}
The relation $\aconc_d$ is the smallest relation ${\aconc_x}\subseteq \Tr^{s\bullet}\times\Tr$ such that
\begin{itemise}
\item $\chi|P\aconc_x Q|\zeta$ ~~and~~ $P|\chi\aconc_x\zeta|Q$,
\item $\chi|v \aconc_x Q|\zeta$ ~~and~~ $v|\chi \aconc_x \zeta|Q$ ~~if $\ell(\chi)\in\B!$,
\item $\chi\aconc_x\zeta$ ~~implies~~
$\chi{+}P \mathbin{\aconc_x} \zeta{+}R$,~~
$P{+}\chi \mathbin{\aconc_x} R{+}\zeta$,~~
$\chi{|}P \mathbin{\aconc_x} \zeta{|}Q$ ~~and~~
$P{|}\chi \mathbin{\aconc_x} Q{|}\zeta$,
\item $\chi\aconc_x\zeta$ ~~implies~~
$\chi{|}P \mathbin{\aconc_x} \zeta{|}w$,~~
$\chi{|}v \mathbin{\aconc_x} \zeta{|}Q$,~~
$P{|}\chi \mathbin{\aconc_x} w{|}\zeta$ ~~and~~
$v{|}\chi \mathbin{\aconc_x} Q{|}\zeta$,
\item $\chi\aconc_x\zeta$ ~~implies~~
$\chi{|}v \aconc_x \zeta{|}w$ ~~and~~
$v{|}\chi \aconc_x w{|}\zeta$ ~~if $\ell(\chi)\in\B!$,
\item $\chi\aconc_x\zeta \wedge v \aconc_x w$ ~~implies~~
$\chi{|}v \aconc_x \zeta{|}w$, ~~and~~
\item $\chi\aconc_x\zeta$ ~~implies~~
$\chi\backslash \RL \mathbin{\aconc_x} \zeta\backslash \RL$,~~
$\chi[f] \mathbin{\aconc_x} \zeta[f]$,~~
$A{:}\chi \mathbin{\aconc_x} A{:}\zeta$ ~~and~~
$\chi\signals r \mathbin{\aconc_x} \zeta\signals r$\\
  for any $L\subseteq\Ch \dcup \Sig$, relabelling $f$, $A\mathbin\in\A$ and $r \in\Sig\!$, and
\item $\chi\aconc_x\xi$~~ for any derivation $\xi$ of an {\signal} transition.
\end{itemise}
for arbitrary $\chi$, $\zeta$, $v$, $w$, $P$, $Q$ and $R$,
where $\chi$ and $v$ are derivations of {\signal}s or transitions,
$\zeta$ and $w$ are derivations of non-{\signal} transitions,
$P,R\mathbin\in\cT$ are expressions, and $Q$ is either an expression or the derivation of an {\signal}
or {\signal} transition---provided that the composed derivations exist.
\end{proposition}

\begin{proofNoBox}
It is straightforward to check that $\aconc_d$ satisfies all the properties listed in
\pr{concurrency}, so the smallest relation $\aconc_x$ is contained in $\aconc_d$.
For the other direction \we prove by structural induction on $\chi'$ that if $\chi'\aconc_d\zeta'$ then
$\chi'\aconc_x\zeta'$ can be derived by the rules of \pr{concurrency}.
\begin{itemise}
\item If $a\varsigma(\zeta')\mathbin=\emptyset$, i.e., $\zeta'$ is the derivation of an {\signal} transition, then
  $\chi' \aconc_d \zeta'\!$. Correspondingly, $\chi' \aconc_x \zeta'\!$, by the last requirement of \pr{concurrency}.
  So below assume that $a\varsigma(\zeta')\neq\emptyset$.
\item If $\chi'$ has the form $\shar{\alpha}P$ or $P\sig{s}$,
  then $\chi'\aconc_d\zeta'$ for no $\zeta'$, so there is nothing to show.
\item Let $\chi'=\chi[f]$. Then all synchrons of $\chi'$ start with $[f]$, so for $\chi'\aconc_d\zeta'$
  to hold, all active synchrons of $\zeta'$ must start with $[f]$ as well. In fact $\zeta'$ must have the form
  $\zeta[f]$ such that $\chi\aconc_d \zeta$. By induction $\chi\aconc_x \zeta$ and hence
  $\chi'\aconc_x \zeta'$ by the seventh requirement of \pr{concurrency}.
\item The cases $\chi'\mathbin=\chi\backslash\RL$, $\chi'\mathbin=\A{:}\chi$, $\chi'\mathbin=\chi\signals r$,
  $\chi'\mathbin=\chi{+}P$ and $\chi'\mathbin=P{+}\chi$
  proceed in the same way.
\item Let $\chi'=\chi|P$. \We make a further case distinction on $\zeta'$.
  \begin{itemise}
  \item Let $\zeta'=Q|\zeta$. Then always $\chi'\mathbin{\aconc_d}\zeta'\!$,
    and indeed $\chi'\mathbin{\aconc_x}\zeta'$ by the first requirement on $\aconc_x$.
  \item Let $\zeta'=\zeta|Q$. Then all synchrons of $\chi'$ and all active synchrons
    of $\zeta'$ start with $|_\Left$,
    and stripping those off shows that $\chi\aconc_d \zeta$. By induction $\chi\aconc_x \zeta$ and hence
    $\chi'\aconc_x \zeta'$ by the third requirement.
  \item Let $\zeta'=\zeta|w$. Then all synchrons of $\chi'$ and some of $\zeta'$ start with $|_\Left$,
    and stripping those off shows that $\chi\aconc_d \zeta$.  By induction $\chi\aconc_x \zeta$ and hence
    $\chi'\aconc_x \zeta'$ by the fourth requirement on $\aconc_x$.
  \item If $u$ has any other shape, then $\chi'\naconc_d\zeta'$, so there is nothing to show.
  \end{itemise}
\item The case $\chi'=P|\chi$ proceeds symmetrically.
\item Let $\chi'=\chi|v$. \We make a further case distinction on $\zeta'$.
  \begin{itemise}
  \item Let $\zeta'=Q|\zeta$. First consider the case that $\ell(\chi)\in\B!$. Then
    $\ell(v)\in\B?\dcup\B{:}$ and all necessary synchrons of $\chi'$ start with $|_\Left$.
    Since all active synchrons of $\zeta'$ start with $|_\R$, \we have $\chi'\aconc_d\zeta'$.
    Accordingly, $\chi'\aconc_x\zeta'$ by the second requirement on $\aconc_x$.

    In case $\ell(\chi)\notin\B!$, some necessary synchrons of $\chi'$ and all active synchrons of $\zeta$
    start with $|_\R$, and stripping those off shows that $v\aconc_d \zeta$.
    By induction $v\aconc_x \zeta$ and hence $\chi'\aconc_x\zeta'$ by the fourth requirement
    of \pr{concurrency} (reversing the roles of $\chi$ and $v$).
  \item The case $\zeta'=\zeta|Q$ proceeds symmetrically.
  \item Let $\zeta'=\zeta|w$. First consider the case that $\ell(\chi)\in\B!$. Then
    $\ell(v)\in\B?\dcup\B{:}$, and all necessary synchrons of $\chi'$ start with $|_\Left$.
    Stripping those off shows that $\chi\aconc_d \zeta$. By induction $\chi\aconc_x \zeta$ and hence
    $\chi'\aconc_x \zeta'$ by the fifth requirement on $\aconc_x$.

    The case that $\ell(v)\in\B!$ proceeds symmetrically.

    Otherwise, \we obtain $\chi\aconc_d \zeta$ and  $v\aconc_d w$.  By induction $\chi \aconc_x \zeta$ and
    $v\aconc_x w$ and hence $\chi'\aconc_x\zeta'$ by the sixth requirement of \pr{concurrency}.
  \item If $u$ has any other shape, then $\chi'\naconc_d\zeta'$, so there is nothing to show.
  \hfill$\Box$
  \end{itemise}
\end{itemise}
\end{proofNoBox}
The main reason for defining $\aconc$ as a relation of type $\Tr^{s\bullet}\times\Tr$ instead of
merely $\Tr^{\bullet}\times\Tr$, which is all we need in \df{justness}, is that in order to derive
$t' \aconc u'$ with $t'\in\Tr^\bullet$ from the rules of \pr{concurrency}, \we sometimes need a
judgement $t \aconc u$ with $t\in\Tr^{s\bullet}{\setminus}\Tr^\bullet$.\footnote{In the original
  version of this paper \we took ${\aconc}\subseteq\Tr^{\bullet}\times\Tr$; as a consequence,
  Propositions~\ref{pr:concurrency} and~\ref{pr:static concurrency inductive} were not correct.}

The relation $\acGH$ was defined in \cite[Definition~C.4]{GH15a} for ABC\@.
Its definition is almost the same as the one of $\aconc_x$ in \pr{concurrency}, but simplified because there
are no {\signal}s or {\signal} transitions, and adding the requirement that $\source(t)=\source(u)$.
\begin{corollary}{concurrency}
Let $t\in\Tr^\bullet$ and $u\in\Tr$. Then $t \acGH u$ iff $t \aconc u \wedge \source(t)=\source(u)$.
\end{corollary}
\begin{proof}
A trivial structural induction on $t$.
\end{proof}
In spite of this agreement between $\acGH$ and $\aconc$, the former is not suitable as an
alternative for the latter for the purposes of this paper, because our formalisation of justness
depends on judgements $t \naconc u$ for transitions $t$ and $u$ with $\source(t)\neq\source(u)$.

\begin{proposition}{static concurrency inductive}
The relation $\aconc_s$ from \Sec{static} is the smallest relation ${\aconc_x}\subseteq \Tr^{s\bullet}\times\Tr$ such that
\begin{itemise}
\item $\chi|P\aconc_x Q|\zeta$ ~~and~~ $P|\chi\aconc_x\zeta|Q$,
\item $\chi|v \aconc_x Q|\zeta$ ~~and~~ $v|\chi \aconc_x \zeta|Q$ ~~if $\ell(\chi)\in\B!$,
\item $\chi\aconc_x\zeta$ ~~implies~~
$\chi{|}P \mathbin{\aconc_x} \zeta{|}Q$ ~~and~~
$P{|}\chi \mathbin{\aconc_x} Q{|}\zeta$,
\item  $\chi\aconc_x\zeta$ ~~implies~~
$\chi{|}P \mathbin{\aconc_x} \zeta{|}w$,~~
$\chi{|}v \mathbin{\aconc_x} \zeta{|}Q$,~~
$P{|}\chi \mathbin{\aconc_x} w{|}\zeta$ ~~and~~
$v{|}\chi \mathbin{\aconc_x} Q{|}\zeta$,
\item$\chi\aconc_x\zeta$ ~~implies~~
$\chi{|}v \aconc_x \zeta{|}w$ ~~and~~
$v{|}\chi \aconc_x w{|}\zeta$ ~~if $\ell(\chi)\in\B!$,
\item $\chi\aconc_x\zeta \wedge v \aconc_x w$ ~~implies~~
$\chi{|}v \aconc_x \zeta{|}w$, ~~and~~
\item $\chi\aconc_x\zeta$ ~~implies~~
$\chi\backslash \RL \mathbin{\aconc_x} \zeta\backslash \RL$ ~~and~~
$\chi[f] \mathbin{\aconc_x} \zeta[f]$ ~~for any $L\subseteq\Ch\dcup \Sig$ and relabelling $f$, and
\item $\chi\aconc_x\xi$~~ for any derivation $\xi$ of an {\signal} transition.
\end{itemise}
for arbitrary $\chi$, $\zeta$, $v$, $w$, $P$, $Q$ and $R$,
where $\chi$ and $v$ are derivations of {\signal}s or transitions,
$\zeta$ and $w$ are derivations of non-{\signal} transitions,
$P,R\mathbin\in\cT$ are expressions, and $Q$ is either an expression or the derivation of an {\signal}
or {\signal} transition---provided that the composed derivations exist.
\end{proposition}
\begin{proof}
  A trivial adaptation of the proof of the previous proposition.
\end{proof}

\section{A coinductive characterisation of justness}\label{sec:coinductive}

In this section \we show that the $\aconc$-based concept of justness defined in this paper coincides
with a coinductively defined concept of justness, for CCS and ABC originating from \cite{GH15a}\weg{,
and for CCSS agreeing with the one defined in \cite{EPTCS255.2}}.

To obtain agreement between our $\aconc$-based and coinductive definitions for CCSS, \we first
extend the $\aconc$-based concept of $B$-justness to the case where also CCSS signal emissions from
$\bar\Sig$ may appear in $B$.  Since this extension is unsuitable as a completeness criterion, and
hence should not be confused with the proper concept of justness, \we have not treated this extension
from the start of the paper, and give it another name: $B$-sigjustness. This extension is needed
because in the coinductive definition, some cases of proper $B$-justness depend on cases of
$C$-sigjustness, where $C$ involves signal emissions.

\subsection{An extension of the notion of justness dealing with emissions}

\begin{definition}{sigjustness}
  A path $\pi$ in an LTSC is \emph{$B$-sigjust}, for $\B?\subseteq B\subseteq Act\dcup\bar\Sig$,
  if for each suffix $\pi'$ of $\pi$, and for each derivation $t \in \Tr^{s\bullet}_{\neg B}$
  enabled in the starting state of $\pi'$, a transition $u$ with $t \naconc u$ occurs in $\pi'$.
\end{definition}
$B$-sigjust corresponds with what is called \emph{$\bar B\cap\Sig$-signalling and $B\setminus\bar\Sig$-just} in \cite{EPTCS255.2}.
Here \we save double work by collapsing the similar concepts \emph{signalling} and \emph{just} from  \cite{EPTCS255.2}.
Note that a path is $B$-just in the sense of \df{justness}, where $\B?\subseteq B\subseteq Act$, iff
it is $B\dcup\bar\Sig$-sigjust according to \df{sigjustness}.

\pr{static justness} and \thm{static justness} extend from justness to sigjustness, with the same proofs.
However, sigjustness is unsuitable as a completeness criterion, because it fails the requirement of feasibility.
\begin{example}{unfeasible}
The process ${\bf 0}\signals s$ has only one path $\pi$, and $\pi$ has no transitions.
This path is not $\emptyset$-sigjust, since a transition ${\bf 0}\sig{s}$ is enabled in its only state.
So $\pi$ can not be extended into an $\emptyset$-sigjust path.

Changing the definition of a path to allow {\signal} transitions does not help;
this allows an infinite path $\pi'$ containing the transition ${\bf 0}\sig{s}$ infinitely often. But
as ${\bf 0}\sig{s} \aconc {\bf 0}\sig{s}$, also $\pi'$ fails to be $\emptyset$-sigjust.
\end{example}

\subsection{A coinductive definition of justness}

To state our coinductive definition of justness, \we need to define the notion of the decomposition of a
path starting from a process with a leading static operator.
\newcommand{\startingfrom}{of }

Any derivation $t\in\Tr$ of a transition with $\source(t)=P|Q$ has the shape
\begin{itemise}
\item $u|Q$, with $\target(t)=\target(u)|Q$,
\item $u|v$, with $\target(t)=\target(u)|\target(v)$,
\item or $P|v$, with $\target(t)=P|\target(v)$.
\end{itemise}
Let a path \emph{of} a process $P$ be a path as in \df{path} starting with $P$.
Now the \emph{decomposition} of a path $\pi$ \startingfrom $P|Q$ into paths $\pi_1$ and $\pi_2$ \startingfrom $P$ and $Q$, respectively, is obtained by
concatenating all left-projections of the states and transitions of $\pi$ into a path \startingfrom $P$ and all right-projections into a
path \startingfrom $Q$---notation $\pi \Rrightarrow \pi_1 | \pi_2$. Here it could be that $\pi$ is infinite, yet either $\pi_1$ or $\pi_2$ (but not both) are finite.

Likewise, $t\in\Tr$ with $\source(t)=P[f]$ has the shape $u[f]$ with $\target(t)=\target(u)[f]$.
The \emph{decomposition} $\pi'$ of a path $\pi$ \startingfrom $P[f]$ is the path obtained by leaving
out the outermost $[f]$ of all states and transitions in $\pi$, notation $\pi\Rrightarrow\pi'[f]$.
In the same way one defines the decomposition of a path \startingfrom $P\backslash c$.

The following co-inductive definition of the family $B$-justness of predicates on paths, with one
family member for each choice of a set $B$ of blocking actions, stems from \cite[Appendix E]{GH15a}.%
\footnote{\label{precise}To be precise, the notion of $Y$-justness from \cite{GH15a} translates to
  $Y\dcup\B?$-justness as occurs here. Furthermore, \cite{GH15a} restricts to the case that
  $Y\subseteq\Ch\dcup\bar\Ch$. This makes sense, as in the default computational interpretation
  broadcast actions $b!$ and internal actions $\tau$ can not be blocked by the environment.
  The increased generality occurring in this paper is merely because it comes with no extra costs,
  and in fact saves us here and there from listing restrictions.\label{YvsB}}
To interpret the word ``largest'', one can see justness equivalently as single
predicate on $\Pow(Act)\times\Pi$, where $\Pi$ denotes the set of all paths. To see that there
actually exists a largest such predicate, check that the class of all such predicates is closed
under arbitrary unions.

\begin{definition}{just path}\rm
\emph{$B$-justness}, for $\B?\subseteq B\subseteq Act$, is the largest family of predicates on the paths in the
LTS of ABC such that
\begin{itemise}
\item a finite $B$-just path ends in a state that admits actions from $B$ only (cf.~Footnote~\ref{admit} on Page~\pageref{admit});
\item a $B$-just path of a process $P|Q$ can be decomposed into a $C$-just path of $P$ and a $D$-just
  path of $Q$, for some $C,D\subseteq B$ such that $\tau\in B \vee C\mathord\cap \bar{D}=\emptyset$---here
  $\bar D:=\{\bar{c} \mid c\mathbin\in D\}$;
\item a $B$-just path of $P\backslash L$ can be decomposed into a
  $B\cup L \cup \bar L$-just path of $P$;
\item a $B$-just path of $P[f]$ can be decomposed into an $f^{-1}(B)$-just path of $P$;
\item and each suffix of a $B$-just path is $B$-just.
\end{itemise}
To make this definition apply to CCSS, as well as CCS, ABC and ABCd,
read ``sigjust'' for ``just'' throughout, and allow $\B?\subseteq B\subseteq Act\dcup\bar\Sig$;
a state $P$ admits an action $\bar{s}\in \bar\Sig$ iff \plat{$P\goto{\bar s}P$} or $P\sigar{s}$ holds.
\end{definition}
Intuitively, justness is a completeness criterion, telling which paths can actually occur as runs of
the represented system. A path is $B$-just if it can occur in an environment that may block the actions in $B$.\linebreak[3]
In this light, the first, third, fourth and fifth requirements above are intuitively plausible.
The second requirement first of all says that if $\pi \Rrightarrow \pi_1 | \pi_2$ and $\pi$ can
occur in the environment that may block the actions in $B$, then $\pi_1$ and $\pi_2$ must be able to
occur in such an environment as well, or in environments blocking less.
The last clause in this requirement prevents a $C$-just path of $P$ and a $D$-just path of
$Q$ to compose into a $B$-just path of $P|Q$ when $C$ contains an action $c$ and $D$
the complementary action~$\bar c$ (except when $\tau\in B$). The reason is that no
environment (except one that can block $\tau$-actions) can block both actions for
their respective components, as nothing can prevent them from synchronising with each other.

The fifth requirement helps characterising processes of the form $b+(A|b)$ and $a.(A|b)$, with  \plat{$A \stackrel{{\it def}}{=} a.A$}.
Here, the first transition `gets rid' of the choice and of the leading action $a$, respectively, 
and this requirement reduces the justness of paths of such processes to their suffixes.

\begin{example}{Cataline and Alice again}
To illustrate \df{just path} consider the unique infinite path of the process Alice$|$Cataline of
\ex{Cataline and Alice} in which the transition $t$ does not occur. Taking the empty set of blocking
actions, \we ask whether this path is $\emptyset$-just. If it were, then by the second requirement of
\df{just path} the projection of this path on the process Cataline would need to be $\emptyset$-just
as well. This is the path $1$ (without any transitions) in \ex{Cataline}. It is not $\emptyset$-just
by the first requirement of \df{just path}, because its last state 1 admits a transition.
\end{example}

\weg{
\begin{proposition}{signalling}
A path is $B$-sigjust according to \df{just path} iff it is $\bar B\cap\Sig$-signalling and
$B\setminus\bar\Sig$-just as defined in \cite{EPTCS255.2}.\footnote{To be precise, the notion of
  path in \cite{EPTCS255.2} differs from the one used here. \pr{signalling} holds when
  interpreting the definitions of signalling and justness in  \cite{EPTCS255.2} as applying to our
  notion of path. This discrepancy will be addressed in \Sec{abstract paths}.}
\end{proposition}
\begin{proof}
``If'': Define $B$-sigjust$_{\cite{EPTCS255.2}}$ as $\bar B\cap\Sig$-signalling and
  $B\setminus\bar\Sig$-just according to \cite{EPTCS255.2}. Then, using that decompositions of (our)
  paths are unique, trivially $B$-sigjustness$_{\cite{EPTCS255.2}}$ satisfies the five requirements
  of \df{just path}, and hence is included in $B$-sigjustness according to \df{just path}.

``Only if'': Define a path to be $X$-signalling$'$ iff it is $B$-sigjust for some $B$
  with $X=\bar B\cap\Sig$. Then trivially $X$-signalling$'$ satisfies the five requirements of the
  coinductive definition of signalling paths in \cite{EPTCS255.2}, and hence is included in
  $X$-signalling according to \cite{EPTCS255.2}.

  Likewise, define a path to be $X$-just$'$ iff it is $B$-sigjust for some $B$
  with $X=B\setminus\bar\Sig$. Then trivially $X$-just$'$ satisfies the five requirements of the
  coinductive definition of just paths in \cite{EPTCS255.2}, and hence is included in
  $X$-justness according to \cite{EPTCS255.2}.
\end{proof}
}

\subsection{Agreement between the concurrency-based and coinductive definitions of justness}

\We now establish that the concept of justness from \df{just path} agrees with the concept of
justness defined earlier in this paper.
The below applies to CCSS by reading $\Tr^{s\bullet}$ for $\Tr^\bullet$ and ``sigjust'' for ``just''.

\begin{theorem}{coinductive}
A path is $\aconc_s$-$B$-just iff it is $B$-just in the sense of \df{just path}.
\end{theorem}
\begin{proofNoBox}
``Only if'': It suffices to show that $\aconc_s$-$B$-justness satisfies the five requirements of
  \df{just path}.
\begin{itemise}
\item Let $\pi$ be a $\aconc_s$-$B$-just path. It follows immediately from \df{justness} that its
  last state admits no transitions $t \in \Tr^\bullet_{\neg B}$.
\item Let $\pi$ be a $\aconc_s$-$B$-just path of a process $P|Q$.
  There is a unique decomposition $\pi \Rrightarrow \pi_1 | \pi_2$ of $\pi$ into a path $\pi_1$ of $P$ and a path $\pi_2$ of $Q$.
  Let $C'$ be the set of actions $\alpha$ such that there is a $t \in \Tr^\bullet$ with
  $s\mathbin{:=}\source(t)\mathbin\in\pi_1$ and $\ell(t)\mathbin=\alpha$, but no transition $u$ with $t \naconc_s u$ occurs in $\pi_1$ past the occurrence of $s$.\linebreak[3]
  Take $C:=C'\dcup\B?$.
  Then $\pi_1$ is $\aconc_s$-$C$-just. In fact, $C$ is the smallest set $\B?\subseteq X\subseteq Act$ such that $\pi_1$ is $X$-just.
  Likewise, let $D$ be the smallest set such that $\B?\subseteq D\subseteq Act$ and $\pi_2$ is $D$-just.
  It remains to be shown that $C,D\subseteq B$ and $\tau\in B \vee C\mathord\cap \bar{D}=\emptyset$.

  Let $\alpha\in C'$. Then there is a state $P'|Q'$ in $\pi$ and a $t \in \Tr^\bullet$ with
  $P'=\source(t)\mathbin\in\pi_1$ and $\ell(t)\mathbin=\alpha$, but no transition $u$ with
  $t \naconc_s u$ occurs in $\pi_1$ past the occurrence of $P'$. \We claim that $\alpha\in B$.
  \begin{itemise}
  \item Let $\alpha\in \Sig \dcup \bar\Sig \dcup \Ch \dcup \bar \Ch \dcup \{\tau\}\!$.
  Then $t|Q' \mathbin\in \Tr^\bullet$ with $P'|Q'\mathbin=\source(t|Q')\mathbin\in\pi$ and $\ell(t|Q')\mathbin=\alpha$.\linebreak[4]
  Suppose, towards a contradiction, that $\alpha\notin B$. Then, using the $\aconc_s$-$B$-justness of $\pi$, a transition
  $t^\dagger$ must occur in $\pi$ past the occurrence of $P'|Q'$, such that $t|Q' \naconc_s t^\dagger$. 
  Since $t^\dagger$ occurs in $\pi$, $\source(t^\dagger)$ has the form $P''|Q''$.
  By \pr{static concurrency inductive}, $t^\dagger$ must have the form $u|v$ or $u|Q''$, with $t \naconc_s u$.
  Hence a transition $u$ with $t \naconc_s u$ occurs in $\pi$ past the occurrence of $P'$---a contradiction.
  So $\alpha\in B$.
  \item Let $\alpha=b!$ with $b\in \B!$. Then either $t|Q' \in \Tr^\bullet$ with $\ell(t|Q')\mathbin=\alpha$, or
    $t|v \in \Tr^\bullet$ with $\ell(v)=b?$ or $\ell(v)=b{:}$ and $\ell(t|v)\mathbin=b!=\alpha$.
    In both cases the argument proceeds as above.
  \end{itemise}
  It follows that $C\subseteq B$. By symmetry also $D\subseteq B$.

  Let $c\in C$ and $\bar c \in D$. Then there are states $P_1|Q_1$ and $P_2|Q_2$ in $\pi$ and $t_1,t_2 \in \Tr^\bullet$ with
  \begin{itemise}
  \item $P_1\mathbin=\source(t_1)\mathbin\in\pi_1$ and $\ell(t_1)\mathbin=c$, but no $u$ with
  $t_1 \naconc_s u$ occurs in $\pi_1$ past $P_1$, and
  \item $Q_2\mathbin=\source(t_2)\mathbin\in\pi_2$ and $\ell(t_2)\mathbin=\bar c$, but no $w$ with
  $t_2 \naconc_s w$ occurs in $\pi_2$ past $Q_2$.
  \end{itemise}
  Assume that either $P_1|Q_1 = P_2|Q_2$ or the state $P_2|Q_2$ occurs in $\pi$ past the state
  $P_1|Q_1$---the other case will follow by symmetry. By \lem{closure static} there is a
  $t_1'\in\Tr^\bullet$ with $\source(t_1')=P_2$ and $t_1 \equiv t'_1$. So $\ell(t'_1)=c$.
  Now $t_1'|t_2\in\Tr^\bullet$ and $\source(t_1'|t_2)=P_2|Q_2$. Moreover, $\ell(t_1'|t_2)=\tau$.
  Assume that $\tau\notin B$. Then, using the $\aconc_s$-$B$-justness of $\pi$, a transition
  $t^\dagger$ must occur in $\pi$ past the occurrence of $P_2|Q_2$, such that $t'_1|t_2 \naconc_s t^\dagger$. 
  By \pr{static concurrency inductive} $t^\dagger$ must have the form $P'|v$ or $u|v$ or $u|Q'$ with $t'_1 \naconc_s u$ or $t_2\naconc_s v$.
  Again \we obtain a contradiction.
\item Let $\pi$ be a $\aconc_s$-$B$-just path of a process $P\backslash L$.
  Let $\pi'$ be the decomposition of $\pi$. \We have to show that $\pi'$ is
  $\aconc_s$-$(B\cup\{c,\bar c \in Act \mid c \in L\})$-just.
  So assume $t \in \Tr^\bullet$ with $\ell(t)\notin B\cup\{c,\bar c \in Act \mid c \in L\}$,
  and $P':=\source(t)\in\pi'$. Then $P'\backslash L = \source(t\backslash L)\in\pi$ and
  $\ell(t\backslash L)\notin B$. By the $\aconc_s$-$B$-justness of $\pi$, a transition
  $t^\dagger$ must occur in $\pi$ past the occurrence of $P'\backslash L$, such that $t\backslash L \naconc_s t^\dagger$. 
  Now $t^\dagger$ must have the form $u\backslash L$, and by \pr{static concurrency inductive} $t \naconc_s u$.
  So a transition $u$ occurs in $\pi'$ past the occurrence of $P'$, such that $t \naconc_s u$.
\item The case that $\pi$ is a $\aconc_s$-$B$-just path of a process $P[f]$ goes likewise.
\item Finally, it follows directly from \df{justness} that each suffix of a $\aconc_s$-$B$-just path is $\aconc_s$-$B$-just.
\end{itemise}
``If'': Let $t \in \Tr^\bullet_{\neg B}$ with $s:=\source(t)\in\pi$ for a path $\pi$ that is
  $B$-just in the sense of \df{just path}. \We have to show that a
  transition $t^\dagger$ with $t \naconc_s t^\dagger$ occurs in $\pi$ past the occurrence of $s$.
  Using the last requirement of \df{just path} \we may assume, without loss of generality, that $s$
  is the first state of $\pi$.
  \We proceed by structural induction on $t$.
  \begin{itemise}
  \item Let $t$ have the form $\shar{\alpha}P$ or $P\sig r$ or $P+u$ or $u+Q$ or $A{:}u$ or
    $t\signals r$.  Then $\npc(t)=\{\epsilon\}$.
    Using the first requirement of \df{just path}, $s$ cannot be the last state of $\pi$, for it
    admits a transition $t$ with $\ell(t)\notin B$.
    Since $s$ has the form $\alpha.P$ or $P+Q$ or $A$ or $P\signals r$, the first transition $v$ of $\pi$ satisfies
    $\afc(v)=\{\epsilon\}$, and thus $t \naconc_s v$.
  \item Let $t$ have the form $u|v$. Then $s$ has the form $P|Q$, with $P:=\source(u)$ and $Q:=\source(v)$.
    By the second requirement of \df{just path}, $\pi \Rrightarrow \pi_1 | \pi_2$, with $\pi_1$ a
    $C$-just path of $P$ and $\pi_2$ a $D$-just path of $Q$, for some $C,D\subseteq B$ such that $\tau\in B \vee C\mathord\cap \bar{D}=\emptyset$.
    \begin{itemise}
    \item  Let $\ell(u)=c\in\Sig \dcup \bar\Sig \dcup \Ch\dcup\bar\Ch$. Then $\ell(v)=\bar c$.
      Since $t\in \Tr^\bullet_{\neg B}$, $\tau=\ell(t)\notin B$.
      So either $c\notin C$ or $\bar c \notin D$---by symmetry assume the former.%
    \item Let $\ell(u)=\ell(t)=b!$ with $b\in\B$. (The case $\ell(v)=b!$ follows by symmetry.) Then $b!\notin B \supseteq C$.
     \end{itemise}
    So in all relevant cases $u\in\Tr^\bullet_{\neg C}$.
      By induction, a transition $u^\dagger$ with $u \naconc_s u^\dagger$ occurs in $\pi_1$.
      Consequently, a transition $t^\dagger$ of the form $u^\dagger|Q'$ or $u^\dagger|v^\dagger$ occurs in $\pi$.
      By \pr{static concurrency inductive} $t\naconc t^\dagger$.
  \item Let $t$ have the form $u|Q$. Then $s$ has the form $P|Q$, with $P:=\source(u)$.
    By the second requirement of \df{just path}, $\pi \Rrightarrow \pi_1 | \pi_2$, with $\pi_1$ a
    $C$-just path of $P$ and $\pi_2$ a $D$-just path of $Q$, for some $C,D\subseteq B$. Since $\ell(u)=\ell(t)\notin B \supseteq C$, $u\in\Tr^\bullet_{\neg C}$.
    The argument proceeds as above.
  \item The case that $t$ has the form $P|v$ follows by symmetry.
  \item Let $t$ have the form $u\backslash L$. Then $s$ has the form $P\backslash L$, with $P:=\source(u)$.
    Moreover $\ell(u)=\ell(t)\notin B\cup\{c,\bar c \in Act \mid c \in L\}$.
    By the third requirement of \df{just path}, the decomposition $\pi'$ of $\pi$ is $(B\cup\{c,\bar c \in Act \mid c \in L\})$-just.
    So by induction, a transition $u^\dagger$ with $u \naconc_s u^\dagger$ occurs in $\pi'$.
    Consequently, a transition $t^\dagger = u^\dagger\backslash L$ occurs in $\pi$.
    By \pr{static concurrency inductive} $t\naconc t^\dagger$.
  \item Let $t$ have the form $u[f]$. Then $s$ has the form $P[f]$, with $P:=\source(u)$.
    Moreover $\ell(u)\notin f^{-1}(B)$.
    By the fourth requirement of \df{just path}, the decomposition $\pi'$ of $\pi$ is $f^{-1}(B)$-just.
    So by induction, a transition $u^\dagger$ with $u \naconc_s u^\dagger$ occurs in $\pi'$.
    Consequently, a transition $t^\dagger = u^\dagger[f]$ occurs in $\pi$.
    By \pr{static concurrency inductive} $t\naconc t^\dagger$.
    \hfill$\Box$
  \end{itemise}
\end{proofNoBox}

\section{Justness on abstract paths}\label{sec:abstract paths}
\advance\textheight 12pt

By \df{path}, a path is an alternating sequence of states and non-{\signal} transitions.
These non-{\signal} transitions are, in the LTS for CCS and its extensions constructed in \Sec{LTSC},
actually \emph{derivations} of transitions $P\ar{\alpha}Q$ according to the structural operational
semantics of these languages. Now define an \emph{abstract path} to be an alternating sequence of
states and actual transitions $P\ar{\alpha}Q$.
\begin{definition}{abstract path}
Let~~$\widehat{\cdot}$~~be the function that takes a derivation $t\in\Tr^\circ$ into the transition
$P\ar{\alpha}Q$ derived by $t$. Given a path $\pi=s_0\,t_1\,s_1\,t_2\,s_2\cdots$,
let $\widehat\pi := s_0\,\widehat t_1\,s_1\,\widehat t_2\,s_2\cdots$. An \emph{abstract path} is
such an object $\widehat\pi$.
\end{definition}
The concept of justness naturally lifts from paths to abstract paths:
\begin{definition}{abstract justness}
An abstract path $\rho$ is $B$-just iff there exists a $B$-just (concrete) path $\pi$ such that $\rho=\widehat\pi$.
\end{definition}
This definition fits with the intuition that a path is just iff it models a run that can actually occur.

The following variant of \df{justness} defines $\aconc_s$-$B$-justness directly on abstract paths.

\begin{definition}{abstract justness via concurrency}
  An abstract path $\rho$ is \emph{$\aconc_s$-$B$-just}, for $\B?\subseteq B\subseteq Act$,
  if for each derivation $t \in \Tr^\bullet_{\neg B}$ with $P:=\source(t)\in\rho$, 
  there is a $u\in\Tr$ with $t \naconc_s u$ such that $\widehat u$ occurs
  in $\rho$ past the occurrence of $P$.
\end{definition}
\We proceed to show that Definitions~\ref{df:abstract justness}
and~\ref{df:abstract justness via concurrency} agree.

\begin{proposition}{abstract justness via concurrency}
An abstract path is $\aconc_s$-$B$-just in the sense of \df{abstract justness via concurrency} iff it is
$B$-just in the sense of \df{abstract justness}, i.e., iff it has the form $\widehat\pi$ for a
concrete path $\pi$ that is $\aconc_s$-$B$-just in the sense of \df{parametrised justness}.
\end{proposition}

\begin{proof}
  ``If'': Let $\pi$ be a concrete path that is $\aconc_s$-$B$-just.
  Then by \df{abstract justness via concurrency} $\widehat\pi$ is $\aconc_s$-$B$-just.

  ``Only if'': Let $\rho$ be an abstract path that is $\aconc_s$-$B$-just in the sense of
  \df{abstract justness via concurrency}.
  \We present an algorithm for constructing a concrete path $\pi$ that is $\aconc_s$-$B$-just in the sense of
  \df{parametrised justness}, such that $\rho=\widehat\pi$.
  Loosely following the idea behind the proof of \thm{feasibility}, \we build an $\IN\times\IN$-matrix with a
  column for each of the states $P_0, P_1, P_2, \dots$ of $\rho$.
  The column $P_i$ lists the transitions from $\Tr^\bullet_{\neg B}$ enabled in state $P_i$, leaving
  empty most slots if there are only finitely many.\footnote{In case \we have an infinite choice
    operator in the language, the set of transitions $t$ with $\source(t)=P_i$ is not necessarily
    countable. Then \we work with $\equiv$-classes of transitions, just as in \Sec{feasibility2}.}
  Incrementally, \we construct prefixes $\pi_{i}$ of $\pi$.
  As an invariant, \we maintain that $\widehat\pi_i$ is the prefix of $\rho$ with $i$ transitions.
  So $\pi_i$ ends in state $P_i$.
  An entry in the matrix is either empty, filled in with a transition, or crossed out.
  Let $f:\IN\rightarrow \IN\times\IN$ be an enumeration of the entries in this matrix.
  
  At the beginning, take $\pi_0$ to be the path consisting of the first state $P_0$ of $\rho$ only.
  At each step $i\geq 0$ \we extend the path $\pi_i$ into $\pi_{j}$ for some $j>i$, if possible,
  thereby skipping over all $\pi_h$ with $i<h<j$,
  and cross out some transitions occurring in the matrix.
  As an invariant, \we maintain that a transition $t$ occurring in the $k$-th column is already crossed
  out when reaching step $i>k$ iff a transition $u$ occurs in the extension of $\pi_k$ into $\pi_i$
  such that $t \naconc_s u$. Furthermore, when reaching step $i$, no entry in a column $\ell\geq i$
  is already crossed out. At each step $i\geq 0$ \we proceed as follows:

  \We take $n\in\IN$ to be the smallest value such that entry $f(n)=(k,m)\in\IN\times\IN$---with $k$
  a column number---satisfies $k\leq i$ and is filled in, say with $t\in\Tr^\bullet_{\neg B}$,
  but not yet crossed out. If such an $n$ does not exist, just extend $\pi_i$ with an arbitrary
  transition $u$ such that $\widehat u$ is the next transition of $\rho$; if $\rho$ ends in $P_i$,
  the algorithm terminates, with output $\pi_i$.
  By our invariant, all transitions $v$ occurring in the extension of $\pi_k$ into $\pi_i$
  satisfy $t \aconc_s v$. By \lem{closure static} there is a $t'\in\Tr^\bullet_{\neg B}$ with 
  $\source(t')=P_i$ and $\chi \equiv \chi'$.
  Since $\rho$ is $\aconc_s$-$B$-just, there is a $u\in\Tr$ with
  $t' \naconc_s u$ such that $\widehat u$ occurs in $\rho$ past the occurrence of $P$.
  So also $t \naconc_s u$. Now extend $\pi_i$ into $\pi_j$, for $j>i$, such that $\pi_j$ ends with
  the transition $u$. Cross out all entries in the matrix up to row $j$ necessary to maintain the
  invariant above. This includes entry $f(n)$.
  
  The desired path $\pi$ is the limit of all the $\pi_i$.
  It is $\aconc_s$-$B$-just, using the invariant, because each transition $t\in\Tr^\bullet_{\neg B}$ that is
  enabled in a state of $\pi$ appears in the matrix, which acts like a priority queue, and is
  eventually crossed out.
\end{proof}

\begin{corollary}{intersection}
  Let $\rho$ be an abstract path. If $\rho$ is $B$-just then it is $C$-just for any $C\supseteq B$.\\
  If $\rho$ is $C$-just as well as $D$-just, then it is $C\cap D$-just.
\end{corollary}
In fact the collection of sets $B$ such that a given abstract path $\rho$ is $B$-just is closed under arbitrary intersection,
and thus there is a least set $B_\rho$ such that $\rho$ is $B$-just. Actions $\alpha\in B_\rho{\setminus}\B?$ are
called \emph{$\rho$-enabled}~\cite{GH15b}. An action $\alpha$ is $\rho$-enabled iff there is a
suffix $\rho'$ of $\rho$ and a derivation $t \in \Tr^\bullet$ with $\ell(t)=\alpha$,
enabled in the starting state of $\rho'$, such that $t \aconc_s v$ for all $v\in\Tr$ such that
$\widehat v$ occurs in $\rho'$.
As a consequence of \df{justness}, the same closure properties apply to justness on concrete paths,
but for abstract paths these results are much less trivial.

\We now show that the concepts of justness on abstract paths from Definitions~\ref{df:abstract justness}
and~\ref{df:abstract justness via concurrency} both agree with the original definition of justness from \cite{GH15a}.
This requires lifting the definition of decomposition from concrete to abstract paths.

\begin{definition}{decomposition}
An abstract path $\rho$ of a process $P|Q$ \emph{can be decomposed} into abstract paths $\rho_1$ of $P$ and
$\rho_2$ of $Q$, notation $\rho \in \rho_1|\rho_2$, if there exist paths $\pi$, $\pi_1$ and $\pi_2$
such that $\pi\Rrightarrow \pi_1|\pi_2$, $\rho=\widehat\pi$ and $\rho_i=\widehat\pi_i$.

Likewise, an abstract path $\rho$ of $P[f]$ \emph{can be decomposed} into an abstract path $\rho'$
if there are $\pi$ and $\pi'$ with $\pi\Rrightarrow \pi'[f]$, $\rho=\widehat\pi$ and $\rho'=\widehat\pi'$.
The decomposition of an abstract path $\rho$ of $P\backslash L$ is defined likewise.
\end{definition}
In \cite[Section 4.3]{GH15a} the decomposition of an abstract path was defined in a different style,
but using \cite[Observation E.3]{GH15a} the resulting notion of decomposition is the same.

\begin{figure}
\definecolor{darkorange}{cmyk}{0,0.82,1,0.01}
\definecolor{Green}{cmyk}{1,0,0.7,0.5}
\renewcommand{\df}[1]{({\color{Green}df.\,\ref{df:#1}})}
\renewcommand{\pr}[1]{({\color{blue}pr.\,\ref{pr:#1}})}
\renewcommand{\thm}[1]{({\color{blue}thm.\,\ref{thm:#1}})}
\renewcommand{\Sec}[1]{({\color{blue}\S\,\ref{sec:#1}})}
\newcommand{\equals}{{\color{magenta}$=$}}
\newcommand{\dequals}{{\color{magenta}$:=$}}
\newcommand{\vequals}{{\color{magenta}$\|$}}
\newcommand{\notion}[1]{{\color{darkorange}#1}}

 \input{overview}
  \centerline{\box\graph}
\caption{\it\small Overview of the various notions of justness defined in this paper}
\label{overview}
\end{figure}

In \cite[Definition 4.1]{GH15a} $B$-justness was defined directly on abstract paths.
The definition is the same as the one for concrete paths---see \df{just path}---but reading
``abstract path'' for ``path";\footnote{Figure~\ref{overview} positions this notion of justness in relation
to the others that appear in this paper.} see also Footnote~\ref{precise}.
The following theorem says that that definition agrees with \df{abstract justness} above.

\begin{theorem}{abstract justness}
An abstract path is $B$-just in the sense of \df{just path} iff it is $B$-just in the sense of
\df{abstract justness}, i.e., iff it has the form $\widehat\pi$ for a
concrete path $\pi$ that is $B$-just in the sense of \df{just path}.
\end{theorem}

\begin{proof}
  ``If'': It suffices to show that the family of predicates $B$-justness on abstract paths according
  to \df{abstract justness} satisfies the five requirements of \df{just path}. This is
  straightforward to check, and spelled out in \cite[Proof of Proposition E.4]{GH15a}.

  ``Only if'': Let $\rho$ be an abstract path that is $B$-just in the sense of \df{just path}.
  By \pr{abstract justness via concurrency} it suffices to show that $\rho$ is $\aconc_s$-$B$-just
  in the sense of \df{abstract justness via concurrency}. This proceeds just as in the
  ``If''-part of the proof of \thm{coinductive}.
\end{proof}

All definitions and results in this section apply equally well to ``sigjustness'' in the role of
``justness'', allowing $\B?\subseteq B\subseteq Act\dcup\bar\Sig$.
In this form $B$-sigjustness is the same as $B$-justness as defined in \cite[Definition~4]{Bou18}.%
\footnote{Following \cite{GH15a,GH15b}, \cite{Bou18} restricts to the case that
  $B\subseteq\Ch\dcup\bar\Ch \dcup \Sig \dcup \bar\Sig$---cf.\ Footnote~\ref{YvsB}.
  Also, in \cite{GH15b,EPTCS255.2,Bou18} one has $\B?=\emptyset$.}
Finally, \we compare our definitions of justness to the one in \cite{EPTCS255.2}.

\begin{proposition}{CCSS signalling}
An abstract path is $Y$-signalling as defined in \cite{EPTCS255.2} iff it is $\bar Y \dcup Act$-sigjust
by \df{just path}.
\end{proposition}
\begin{proof}
``If'': $\bar Y \dcup Act$-sigjustness trivially satisfies the five conditions for
$Y$-signalling of \cite[Definition 2]{EPTCS255.2}.\linebreak[3]
The second condition (for paths starting in $P|Q$) uses the first statement of \cor{intersection}.

``Only if'': Call an abstract path $\rho$ $B$-sigjust$^\dagger$ iff $B=\bar Y \dcup Act$ for an $Y\subseteq\bar\Sig$
such that $\rho$ is $Y$-signalling as defined in \cite{EPTCS255.2}. Then trivially $B$-sigjustness$^\dagger$
satisfies the five conditions of \df{just path}.
\end{proof}

\begin{proposition}{CCSS justness}
  An abstract path is $Y$-just as defined in \cite{EPTCS255.2} iff it is
  $B$-sigjust according to \df{just path} for some $B$ with $Y=B\cap Act$, which is the case iff it is
  $Y \dcup \bar\Sig$-sigjust according to \df{just path}.
\end{proposition}

\begin{proof}
  ``If'': Call an abstract path $Y$-just$^\dagger$ iff it is
  $B$-sigjust by \df{just path} for some $B$ with $Y=B\cap Act$.
  Then trivially $B$-justness$^\dagger$ satisfies the five conditions of \cite[Definition 3]{EPTCS255.2}.
  The second condition (for paths starting in $P|Q$) uses \pr{CCSS signalling} in conjunction with the first statement of \cor{intersection}.

  ``Only if'': It suffices to show that each abstract path $\rho$ that is $Y$-just as defined in
  \cite{EPTCS255.2} is also $\naconc_s$-$Y \dcup \bar\Sig$-sigjust according to \df{abstract justness via concurrency}.
  So for each $t \in \Tr^\bullet$ with $\ell(t)\in Act\setminus Y$ and $s:=\source(t)\in\rho$ for such a path $\rho$, \we have to find a
  transition $t^\dagger$ with $t \naconc_s t^\dagger$ such that \plat{$\widehat {t^\dagger}$} occurs in $\rho$ past the occurrence of $s$.
  The proof of this statement is similar to direction ``If'' of the proof of \thm{coinductive}.
  Using the last requirement of \cite[Definition 3]{EPTCS255.2} \we may assume, without loss of generality, that $s$
  is the first state of $\rho$.  \We proceed by structural induction on $t$.
  The only case that deviates from the proof of \thm{coinductive} is where $t$ has the form $u|v$.
    In this case $s$ has the form $P|Q$, with $P:=\source(u)$ and $Q:=\source(v)$.
    By the second requirement of \cite[Definition 3]{EPTCS255.2}, $\rho \Rrightarrow \rho_1 | \rho_2$,
    with $\rho_1$ an $X$-just and $X'$-signalling abstract path of $P$ and $\rho_2$ a $Z$-just and
    $Z'$-signalling abstract path of $Q$, for some $X,Z\subseteq Y$ and $X',Z'\subseteq\Sig$, such
    that (when $\tau\notin Y$) $X\mathord\cap \bar{Z}=\emptyset$, $X\cap Z'=\emptyset$ and $X'\cap Z=\emptyset$.
    \begin{itemise}
    \item  Let $\ell(u)=c\in\Ch\dcup\bar\Ch$. Then $\ell(v)=\bar c$ and $\tau=\ell(t)\notin Y$.
      So either $c\notin X$ or $\bar c \notin Z$---by symmetry assume the former.
    \item  Let $\ell(u)=s\in\Sig$. Then $\ell(v)=\bar s\in\bar\Sig$ and $\tau=\ell(t)\notin Y$.
      So either $s\notin X$ or $s \notin Z'$.
    \item  The case $\ell(v)=c\in\Sig$ will follow by symmetry.
    \item Let $\ell(u)=\ell(t)=b!$ with $b\in\B$. (The case $\ell(v)=b!$ follows by symmetry.) Then $b!\notin Y \supseteq X$.
     \end{itemise}
    So in all but one of the relevant cases $u\in\Tr^\bullet$ with $\ell(u)\in Act\setminus X$.
      By induction, there is a transition $u^\dagger$ with $u \naconc_s u^\dagger$ such that \plat{$\widehat{u^\dagger}$} occurs in $\rho_1$.\vspace{2pt}
      Consequently, there is a transition $t^\dagger$ of the form $u^\dagger|Q'$ or
      $u^\dagger|v^\dagger$ such that \plat{$\widehat {t^\dagger}$} occurs in $\rho$.
      By \pr{static concurrency inductive} $t\naconc t^\dagger$.

      In the remaining case $\ell(v)=\bar s\in\bar\Sig$ and $\rho_2$ is $Z'$-signalling with $s\notin Z'$.
      Using \pr{CCSS signalling}, Theorems~\ref{thm:abstract justness} and~\ref{thm:coinductive} and \pr{abstract justness via concurrency},
      $\rho_2$ is $\naconc_s$-$\bar Z' \dcup Act$-sigjust in the sense of \df{abstract justness via concurrency}, so
  there is a $v^\dagger\in\Tr$ with $v \naconc_s v$ such that \plat{$\widehat {v^\dagger}$} occurs in $\rho_2$.
      Consequently, there is a transition $t^\dagger$ of the form $P'|v^\dagger$ or
      $u^\dagger|v^\dagger$ such that \plat{$\widehat {t^\dagger}$} occurs in $\rho$.
      By \pr{static concurrency inductive} $t\naconc t^\dagger$.
\end{proof}

\begin{corollary}{CCSS just}
An abstract path is $B$-just as defined in \cite{EPTCS255.2} iff it is $B$-just according to \df{justness}.
\end{corollary}

\begin{corollary}{CCSS just and signalling}
An abstract path is $B$-sigjust according to \df{just path} iff it is $\bar B\cap\Sig$-signalling as
well as $B\cap Act$-just as defined in \cite{EPTCS255.2}.
\end{corollary}
\begin{proof}
``Only if'': Let $\rho$ be $B$-sigjust according to \df{just path}. By \cor{intersection} it is also
  $B\cup Act$-sigjust, so by \pr{CCSS signalling} it is $\bar B\cap\Sig$-signalling.
  By \pr{CCSS justness} $\rho$ is $B\cap Act$-just.

``If'': Let $\rho$ be $\bar B\cap\Sig$-signalling as well as $B\cap Act$-just as defined in \cite{EPTCS255.2}.
  By \pr{CCSS signalling} it is $B \cup Act$-sigjust. By \pr{CCSS justness} it is $B \cup \bar\Sig$-sigjust.
  So by \cor{intersection} it is $B$-sigjust.
\end{proof}

In \cite{GH15a,EPTCS255.2,Bou18} a(n abstract) path is called \emph{just} (without a predicate $B$)
iff it is $B$-just for some\linebreak[4] \plat{$\B? \subseteq B \subseteq \B? \dcup \Ch \dcup \bar\Ch \dcup \Sig$},
which is the case iff it is \plat{$\B? \dcup \Ch \dcup \bar\Ch \dcup \Sig$}-just. This amounts to
making a default choice for the set $B$ of blocking actions, in which CCS handshake synchronisations $c$ and
$\bar c$ as well as broadcast receive and signal read actions can always be blocked by the
environment (namely by withholding a synchronisation partner, or failing to broadcast or to emit a signal).
Using this definition it follows that an abstract path is just as defined in \cite{GH19} and the current paper
(using \df{justness} with any of the five concurrency relations $\aconc$, $\aconc_s$, $\aconc'_s$,
$\aconc_c$ or $\aconc'_c$ and taking $B:=\B? \dcup \Ch \dcup \bar\Ch \dcup \Sig$) iff it is just as
defined in \cite{GH15a}, \cite{EPTCS255.2} or \cite{Bou18}.

\section{Conclusion}\label{sec:conclusion}
\advance\textheight -10pt

\We advocate justness as a reasonable completeness criterion for formalising liveness
properties when modelling distributed systems by means of transition systems.
In \cite{GH19} \we proposed a definition of justness in terms of a, possibly asymmetric, concurrency
relation between transitions. The current paper defines such a concurrency relation for the transition
systems associated to CCS, as well as its extensions with broadcast communication or signals, thereby
making the definition of justness from \cite{GH19} available to these languages. In fact, we
provided five versions of the concurrency relation, and showed that they all give rise to the same
concept of justness. \We expect that this style of definition will carry over to many other process algebras.
\We have shown that justness satisfies the criterion of feasibility, and proved that our formalisation
agrees with previous coinductive formalisations of justness for these languages.

Concurrency relations between transitions in transition systems have been studied in \cite{Sta89}.
Our concurrency relation $\aconc$ follows the same computational intuition.
However, in \cite{Sta89} transitions are classified as concurrent or not only
when they have the same source, whereas as a basis for the definition of justness here \we compare
transitions with different sources. Apart from that, our concurrency relation is more general in
that it satisfies fewer closure properties, and moreover is allowed to be asymmetric.

Concurrency is represented explicitly in models like Petri nets \cite{Rei13}, event structures
\cite{Wi87a}, or asynchronous transition systems \cite{Sh85-lb,Bz87,WN95}.
\We believe that the semantics of CCS in terms of such models agrees with its semantics in terms of labelled
transition systems with a concurrency relation as given here. However, formalising such a claim
requires a choice of an adequate justness-preserving semantic equivalence defined on the compared models.
Development of such semantic equivalences is a topic for future research \cite{GHW21}.

\newcommand{\zs}{{CCS}}
\bibliographystyle{eptcsini}
\bibliography{justness}

\begin{thebibliography}{10}
\providecommand{\bibitemdeclare}[2]{}
\providecommand{\surnamestart}{}
\providecommand{\surnameend}{}
\providecommand{\urlprefix}{Available at }
\providecommand{\url}[1]{\texttt{#1}}
\providecommand{\href}[2]{\texttt{#2}}
\providecommand{\urlalt}[2]{\href{#1}{#2}}
\providecommand{\doi}[1]{doi:\urlalt{http://dx.doi.org/#1}{#1}}
\providecommand{\bibinfo}[2]{#2}

\bibitemdeclare{article}{AFK88}
\bibitem{AFK88}
\bibinfo{author}{K.R.~\surnamestart Apt\surnameend},
  \bibinfo{author}{N.~\surnamestart Francez\surnameend} \&
  \bibinfo{author}{S.~\surnamestart Katz\surnameend} (\bibinfo{year}{1988}):
  \emph{\bibinfo{title}{Appraising Fairness in Languages for Distributed
  Programming}}.
\newblock {\sl \bibinfo{journal}{Distributed Computing}}
  \bibinfo{volume}{2}(\bibinfo{number}{4}), pp. \bibinfo{pages}{226--241},
  \doi{10.1007/BF01872848}.

\bibitemdeclare{phdthesis}{Bz87}
\bibitem{Bz87}
\bibinfo{author}{M.~\surnamestart Bednarczyk\surnameend}
  (\bibinfo{year}{1987}): \emph{\bibinfo{title}{Categories of asynchronous
  systems}}.
\newblock Ph.D. thesis, \bibinfo{school}{Computer Science, University of
  Sussex}, \bibinfo{address}{Brighton}.

\bibitemdeclare{techreport}{Bou18}
\bibitem{Bou18}
\bibinfo{author}{M.S.~\surnamestart Bouwman\surnameend} (\bibinfo{year}{2018}):
  \emph{\bibinfo{title}{Liveness analysis in process algebra: simpler
  techniques to model mutex algorithms}}.
\newblock \bibinfo{type}{Technical Report}, \bibinfo{institution}{Eindhoven
  University of Technology}.
\newblock
  \urlprefix\url{http://www.win.tue.nl/~timw/downloads/bouwman_seminar.pdf}.

\bibitemdeclare{inproceedings}{CDPY13}
\bibitem{CDPY13}
\bibinfo{author}{M.~\surnamestart Coppo\surnameend},
  \bibinfo{author}{M.~\surnamestart Dezani{-}Ciancaglini\surnameend},
  \bibinfo{author}{L.~\surnamestart Padovani\surnameend} \&
  \bibinfo{author}{N.~\surnamestart Yoshida\surnameend} (\bibinfo{year}{2013}):
  \emph{\bibinfo{title}{Inference of Global Progress Properties for Dynamically
  Interleaved Multiparty Sessions}}.
\newblock In: {\sl \bibinfo{booktitle}{{\rm Proc. Coordination'13}}}, {\sl
  \bibinfo{series}{\rm LNCS}} \bibinfo{volume}{7890},
  \bibinfo{publisher}{Springer}, pp. \bibinfo{pages}{45--59},
  \doi{10.1007/978-3-642-38493-6_4}.

\bibitemdeclare{article}{DV95}
\bibitem{DV95}
\bibinfo{author}{R.~\surnamestart De~Nicola\surnameend} \&
  \bibinfo{author}{F.W.~\surnamestart Vaandrager\surnameend}
  (\bibinfo{year}{1995}): \emph{\bibinfo{title}{Three Logics for Branching
  Bisimulation}}.
\newblock {\sl \bibinfo{journal}{Journal of the ACM}}
  \bibinfo{volume}{42}(\bibinfo{number}{2}), pp. \bibinfo{pages}{458--487},
  \doi{10.1145/201019.201032}.

\bibitemdeclare{inproceedings}{DDM87}
\bibitem{DDM87}
\bibinfo{author}{P.~\surnamestart Degano\surnameend},
  \bibinfo{author}{R.~\surnamestart {De Nicola}\surnameend} \&
  \bibinfo{author}{U.~\surnamestart Montanari\surnameend}
  (\bibinfo{year}{1987}): \emph{\bibinfo{title}{{CCS} is an (Augmented) Contact
  Free {C/E} System}}.
\newblock In \bibinfo{editor}{M.V.~\surnamestart Zilli\surnameend}, editor:
  {\sl \bibinfo{booktitle}{Mathematical Models for the Semantics of
  Parallelism}}, {\sl \bibinfo{series}{\rm LNCS}} \bibinfo{volume}{280},
  \bibinfo{publisher}{Springer}, pp. \bibinfo{pages}{144--165},
  \doi{10.1007/3-540-18419-8\_13}.

\bibitemdeclare{inproceedings}{EPTCS255.2}
\bibitem{EPTCS255.2}
\bibinfo{author}{V.~\surnamestart Dyseryn\surnameend},
  \bibinfo{author}{R.J.~\surnamestart van Glabbeek\surnameend} \&
  \bibinfo{author}{P.~\surnamestart H\"ofner\surnameend}
  (\bibinfo{year}{2017}): \emph{\bibinfo{title}{Analysing Mutual Exclusion
  using Process Algebra with Signals}}.
\newblock In \bibinfo{editor}{K.~\surnamestart Peters\surnameend} \&
  \bibinfo{editor}{S.~\surnamestart Tini\surnameend}, editors: {\sl
  \bibinfo{booktitle}{{\rm Proc. Combined 24th International Workshop on}
  Expressiveness in Concurrency {\rm and 14th Workshop on} Structural
  Operational Semantics}}, {\sl \bibinfo{series}{Electronic Proceedings in
  Theoretical Computer Science}} \bibinfo{volume}{255},
  \bibinfo{publisher}{Open Publishing Association}, pp.
  \bibinfo{pages}{18--34}, \doi{10.4204/EPTCS.255.2}.

\bibitemdeclare{article}{EC82}
\bibitem{EC82}
\bibinfo{author}{E.A.~\surnamestart Emerson\surnameend} \&
  \bibinfo{author}{E.M.~\surnamestart Clarke\surnameend}
  (\bibinfo{year}{1982}): \emph{\bibinfo{title}{Using Branching Time Temporal
  Logic to Synthesize Synchronization Skeletons.}}
\newblock {\sl \bibinfo{journal}{Science of Computer Programming}}
  \bibinfo{volume}{2}(\bibinfo{number}{3}), pp. \bibinfo{pages}{241--266},
  \doi{10.1016/0167-6423(83)90017-5}.

\bibitemdeclare{article}{EH86}
\bibitem{EH86}
\bibinfo{author}{E.A.~\surnamestart Emerson\surnameend} \&
  \bibinfo{author}{J.Y.~\surnamestart Halpern\surnameend}
  (\bibinfo{year}{1986}): \emph{\bibinfo{title}{{`Sometimes'} and {`Not Never'}
  revisited: on branching time versus linear time temporal logic}}.
\newblock {\sl \bibinfo{journal}{Journal of the ACM}}
  \bibinfo{volume}{33}(\bibinfo{number}{1}), pp. \bibinfo{pages}{151--178},
  \doi{10.1145/4904.4999}.

\bibitemdeclare{inproceedings}{FGHMPT12a}
\bibitem{FGHMPT12a}
\bibinfo{author}{A.~\surnamestart Fehnker\surnameend},
  \bibinfo{author}{R.J.~\surnamestart van Glabbeek\surnameend},
  \bibinfo{author}{P.~\surnamestart H{\"{o}}fner\surnameend},
  \bibinfo{author}{A.K.~\surnamestart McIver\surnameend},
  \bibinfo{author}{M.~\surnamestart Portmann\surnameend} \&
  \bibinfo{author}{W.L.~\surnamestart Tan\surnameend} (\bibinfo{year}{2012}):
  \emph{\bibinfo{title}{A Process Algebra for Wireless Mesh Networks}}.
\newblock In \bibinfo{editor}{H.~\surnamestart Seidl\surnameend}, editor: {\sl
  \bibinfo{booktitle}{{\rm Proc.\ ESOP'12}}}, {\sl \bibinfo{series}{\rm LNCS}}
  \bibinfo{volume}{7211}, \bibinfo{publisher}{Springer}, pp.
  \bibinfo{pages}{295--315}, \doi{10.1007/978-3-642-28869-2_15}.

\bibitemdeclare{techreport}{TR13}
\bibitem{TR13}
\bibinfo{author}{A.~\surnamestart Fehnker\surnameend},
  \bibinfo{author}{R.J.~\surnamestart van Glabbeek\surnameend},
  \bibinfo{author}{P.~\surnamestart H{\"{o}}fner\surnameend},
  \bibinfo{author}{A.K.~\surnamestart McIver\surnameend},
  \bibinfo{author}{M.~\surnamestart Portmann\surnameend} \&
  \bibinfo{author}{W.L.~\surnamestart Tan\surnameend} (\bibinfo{year}{2013}):
  \emph{\bibinfo{title}{A Process Algebra for Wireless Mesh Networks used for
  Modelling, Verifying and Analysing {AODV}}}.
\newblock \bibinfo{type}{Technical Report} \bibinfo{number}{5513},
  \bibinfo{institution}{NICTA}.
\newblock \urlprefix\url{http://arxiv.org/abs/1312.7645}.

\bibitemdeclare{inproceedings}{vG15}
\bibitem{vG15}
\bibinfo{author}{R.J.~\surnamestart van Glabbeek\surnameend}
  (\bibinfo{year}{2015}): \emph{\bibinfo{title}{Structure Preserving
  Bisimilarity, Supporting an Operational Petri Net Semantics of {CCSP}}}.
\newblock In \bibinfo{editor}{R.~\surnamestart Meyer\surnameend},
  \bibinfo{editor}{A.~\surnamestart Platzer\surnameend} \&
  \bibinfo{editor}{H.~\surnamestart Wehrheim\surnameend}, editors: {\sl
  \bibinfo{booktitle}{{\rm Proceedings} Correct System Design - Symposium in
  Honor of Ernst-R{\"{u}}diger Olderog on the Occasion of His 60th Birthday}},
  {\sl \bibinfo{series}{\rm LNCS}} \bibinfo{volume}{9360},
  \bibinfo{publisher}{Springer}, pp. \bibinfo{pages}{99--130},
  \doi{10.1007/978-3-319-23506-6_9}.
\newblock \urlprefix\url{http://arxiv.org/abs/1509.05842}.

\bibitemdeclare{misc}{vG16}
\bibitem{vG16}
\bibinfo{author}{R.J.~\surnamestart van Glabbeek\surnameend}
  (\bibinfo{year}{2016}): \emph{\bibinfo{title}{Ensuring Liveness Properties of
  Distributed Systems (A Research Agenda)}}.
\newblock \bibinfo{howpublished}{Position paper}.
\newblock \urlprefix\url{https://arxiv.org/abs/1711.04240}.

\bibitemdeclare{inproceedings}{GGS08d}
\bibitem{GGS08d}
\bibinfo{author}{R.J.~\surnamestart van Glabbeek\surnameend},
  \bibinfo{author}{U.~\surnamestart Goltz\surnameend} \&
  \bibinfo{author}{J.-W.~\surnamestart Schicke\surnameend}
  (\bibinfo{year}{2008}): \emph{\bibinfo{title}{On Synchronous and Asynchronous
  Interaction in Distributed Systems}}.
\newblock In \bibinfo{editor}{E.~\surnamestart Ochma\'nski\surnameend} \&
  \bibinfo{editor}{J.~\surnamestart Tyszkiewicz\surnameend}, editors: {\sl
  \bibinfo{booktitle}{{\rm Proceedings 33rd International Symposium on}
  Mathematical Foundations of Computer Science {\rm (MFCS 2008), Toru\'n,
  Poland, August 2008}}}, {\sl \bibinfo{series}{\rm LNCS}}
  \bibinfo{volume}{5162}, \bibinfo{publisher}{Springer}, pp.
  \bibinfo{pages}{16--35}, \doi{10.1007/978-3-540-85238-4\_2}.

\bibitemdeclare{article}{GGS13}
\bibitem{GGS13}
\bibinfo{author}{R.J.~\surnamestart van Glabbeek\surnameend},
  \bibinfo{author}{U.~\surnamestart Goltz\surnameend} \&
  \bibinfo{author}{J.-W.~\surnamestart Schicke-Uffmann\surnameend}
  (\bibinfo{year}{2013}): \emph{\bibinfo{title}{On Characterising
  Distributability}}.
\newblock {\sl \bibinfo{journal}{Logical Methods in Computer Science}}
  \bibinfo{volume}{9}(\bibinfo{number}{3}):\bibinfo{eid}{17},
  \doi{10.2168/LMCS-9(3:17)2013}.

\bibitemdeclare{techreport}{GH15a}
\bibitem{GH15a}
\bibinfo{author}{R.J.~\surnamestart van Glabbeek\surnameend} \&
  \bibinfo{author}{P.~\surnamestart H{\"o}fner\surnameend}
  (\bibinfo{year}{2015}): \emph{\bibinfo{title}{Progress, Fairness and Justness
  in Process Algebra}}.
\newblock \bibinfo{type}{Technical Report} \bibinfo{number}{8501},
  \bibinfo{institution}{NICTA}.
\newblock \urlprefix\url{http://arxiv.org/abs/1501.03268}.

\bibitemdeclare{article}{GH15b}
\bibitem{GH15b}
\bibinfo{author}{R.J.~\surnamestart van Glabbeek\surnameend} \&
  \bibinfo{author}{P.~\surnamestart H{\"o}fner\surnameend}
  (\bibinfo{year}{2015}): \emph{\bibinfo{title}{\zs: It's not fair!}}
\newblock {\sl \bibinfo{journal}{Acta Informatica}}
  \bibinfo{volume}{52}(\bibinfo{number}{2-3}), pp. \bibinfo{pages}{175--205},
  \doi{10.1007/s00236-015-0221-6}.

\bibitemdeclare{article}{GH19}
\bibitem{GH19}
\bibinfo{author}{R.J.~\surnamestart van Glabbeek\surnameend} \&
  \bibinfo{author}{P.~\surnamestart H{\"o}fner\surnameend}
  (\bibinfo{year}{2019}): \emph{\bibinfo{title}{Progress, Justness, and
  Fairness}}.
\newblock {\sl \bibinfo{journal}{{ACM} Computing Surveys}}
  \bibinfo{volume}{52}(\bibinfo{number}{4}), pp. \bibinfo{pages}{69:1--69:38},
  \doi{10.1145/3329125}.

\bibitemdeclare{inproceedings}{GHW21}
\bibitem{GHW21}
\bibinfo{author}{R.J.~\surnamestart van Glabbeek\surnameend},
  \bibinfo{author}{P.~\surnamestart H{\"o}fner\surnameend} \&
  \bibinfo{author}{W.~\surnamestart Wang\surnameend} (\bibinfo{year}{2021}):
  \emph{\bibinfo{title}{Enabling Preserving Bisimulation Equivalence}}.
\newblock In \bibinfo{editor}{S.~\surnamestart Haddad\surnameend} \&
  \bibinfo{editor}{D.~\surnamestart Varacca\surnameend}, editors: {\sl
  \bibinfo{booktitle}{{\rm Proc. {CONCUR}'21}}}, {\sl \bibinfo{series}{LIPIcs}}
  \bibinfo{volume}{203}, \bibinfo{publisher}{Schloss Dagstuhl - Leibniz-Zentrum
  f{\"{u}}r Informatik}, pp. \bibinfo{pages}{33:1--33:20},
  \doi{10.4230/LIPIcs.CONCUR.2021.33}.

\bibitemdeclare{inproceedings}{GV87}
\bibitem{GV87}
\bibinfo{author}{R.J.~\surnamestart van Glabbeek\surnameend} \&
  \bibinfo{author}{F.W.~\surnamestart Vaandrager\surnameend}
  (\bibinfo{year}{1987}): \emph{\bibinfo{title}{Petri net models for algebraic
  theories of concurrency (extended abstract)}}.
\newblock In \bibinfo{editor}{J.W.d.~\surnamestart Bakker\surnameend},
  \bibinfo{editor}{A.J.~\surnamestart Nijman\surnameend} \&
  \bibinfo{editor}{P.C.~\surnamestart Treleaven\surnameend}, editors: {\sl
  \bibinfo{booktitle}{{\rm Proceedings} PARLE, Parallel Architectures and
  Languages Europe, {\rm Eindhoven, The Netherlands, June 1987, Vol. II:
  Parallel Languages}}}, {\sl \bibinfo{series}{\rm LNCS}}
  \bibinfo{volume}{259}, \bibinfo{publisher}{Springer}, pp.
  \bibinfo{pages}{224--242}, \doi{10.1007/3-540-17945-3\_13}.

\bibitemdeclare{inproceedings}{KdR83}
\bibitem{KdR83}
\bibinfo{author}{R.~\surnamestart Kuiper\surnameend} \&
  \bibinfo{author}{W.-P.~\surnamestart de~Roever\surnameend}
  (\bibinfo{year}{1983}): \emph{\bibinfo{title}{Fairness Assumptions for {CSP}
  in a Temporal Logic Framework}}.
\newblock In \bibinfo{editor}{D.~\surnamestart Bj{\o}rner\surnameend}, editor:
  {\sl \bibinfo{booktitle}{Formal Description of Programming Concepts II}},
  \bibinfo{publisher}{North-Holland}, pp. \bibinfo{pages}{159--170}.

\bibitemdeclare{article}{Lam77}
\bibitem{Lam77}
\bibinfo{author}{L.~\surnamestart Lamport\surnameend} (\bibinfo{year}{1977}):
  \emph{\bibinfo{title}{Proving the correctness of multiprocess programs}}.
\newblock {\sl \bibinfo{journal}{IEEE Transactions on Software Engineering}}
  \bibinfo{volume}{3}(\bibinfo{number}{2}), pp. \bibinfo{pages}{125--143},
  \doi{10.1109/TSE.1977.229904}.

\bibitemdeclare{article}{La00-lb}
\bibitem{La00-lb}
\bibinfo{author}{L.~\surnamestart Lamport\surnameend} (\bibinfo{year}{2000}):
  \emph{\bibinfo{title}{Fairness and hyperfairness}}.
\newblock {\sl \bibinfo{journal}{Distr. Comp.}}
  \bibinfo{volume}{13}(\bibinfo{number}{4}), pp. \bibinfo{pages}{239--245},
  \doi{10.1007/PL00008921}.

\bibitemdeclare{incollection}{Mi90b}
\bibitem{Mi90b}
\bibinfo{author}{R.~\surnamestart Milner\surnameend} (\bibinfo{year}{1990}):
  \emph{\bibinfo{title}{Operational and algebraic semantics of concurrent
  processes}}.
\newblock In \bibinfo{editor}{J.~\surnamestart van Leeuwen\surnameend}, editor:
  {\sl \bibinfo{booktitle}{Handbook of Theoretical Computer Science}},
  chapter~\bibinfo{chapter}{19}, \bibinfo{publisher}{Elsevier Science
  Publishers B.V. (North-Holland)}, pp. \bibinfo{pages}{1201--1242}.

\bibitemdeclare{book}{Mi80}
\bibitem{Mi80}
\bibinfo{author}{R.~\surnamestart Milner\surnameend} (\bibinfo{year}{1980}):
  \emph{\bibinfo{title}{A Calculus of Communicating Systems}}.
\newblock {\sl \bibinfo{series}{\rm LNCS}}~\bibinfo{volume}{92},
  \bibinfo{publisher}{Springer}, \doi{10.1007/3-540-10235-3}.

\bibitemdeclare{misc}{Mis88}
\bibitem{Mis88}
\bibinfo{author}{J.~\surnamestart Misra\surnameend} (\bibinfo{year}{1988}):
  \emph{\bibinfo{title}{A Rebuttal of Dijkstra's Position on Fairness}}.
\newblock
  \urlprefix\url{http://www.cs.utexas.edu/users/misra/Notes.dir/fairness.pdf}.

\bibitemdeclare{book}{Mis01}
\bibitem{Mis01}
\bibinfo{author}{J.~\surnamestart Misra\surnameend} (\bibinfo{year}{2001}):
  \emph{\bibinfo{title}{A Discipline of Multiprogramming --- Programming Theory
  for Distributed Applications}}.
\newblock \bibinfo{publisher}{Springer}, \doi{10.1007/978-1-4419-8528-6}.

\bibitemdeclare{inproceedings}{Old87}
\bibitem{Old87}
\bibinfo{author}{E.-R.~\surnamestart Olderog\surnameend}
  (\bibinfo{year}{1987}): \emph{\bibinfo{title}{Operational Petri net semantics
  for {CCSP}}}.
\newblock In \bibinfo{editor}{G.~\surnamestart Rozenberg\surnameend}, editor:
  {\sl \bibinfo{booktitle}{Advances in Petri Nets 1987, {\rm covers the 7th
  European Workshop on} Applications and Theory of Petri Nets, {\rm Oxford, UK,
  June 1986}}}, {\sl \bibinfo{series}{\rm LNCS}} \bibinfo{volume}{266},
  \bibinfo{publisher}{Springer}, pp. \bibinfo{pages}{196--223},
  \doi{10.1007/3-540-18086-9_27}.

\bibitemdeclare{book}{Old91}
\bibitem{Old91}
\bibinfo{author}{E.-R.~\surnamestart Olderog\surnameend}
  (\bibinfo{year}{1991}): \emph{\bibinfo{title}{Nets, Terms and Formulas: Three
  Views of Concurrent Processes and their Relationship}}.
\newblock {\sl \bibinfo{series}{Cambridge Tracts in Theor. Comp.
  Sc.}}~\bibinfo{volume}{23}, \bibinfo{publisher}{Cambridge University Press}.

\bibitemdeclare{article}{OL82}
\bibitem{OL82}
\bibinfo{author}{S.S.~\surnamestart Owicki\surnameend} \&
  \bibinfo{author}{L.~\surnamestart Lamport\surnameend} (\bibinfo{year}{1982}):
  \emph{\bibinfo{title}{Proving Liveness Properties of Concurrent Programs}}.
\newblock {\sl \bibinfo{journal}{{ACM} TOPLAS}}
  \bibinfo{volume}{4}(\bibinfo{number}{3}), pp. \bibinfo{pages}{455--495},
  \doi{10.1145/357172.357178}.

\bibitemdeclare{inproceedings}{Pn77}
\bibitem{Pn77}
\bibinfo{author}{A.~\surnamestart Pnueli\surnameend} (\bibinfo{year}{1977}):
  \emph{\bibinfo{title}{The Temporal Logic of Programs}}.
\newblock In: {\sl \bibinfo{booktitle}{{\rm Proc.\ 18th Annual Symposium on}
  Foundations of Computer Science {\rm (FOCS'77)}}}, \bibinfo{publisher}{IEEE},
  pp. \bibinfo{pages}{46--57}, \doi{10.1109/SFCS.1977.32}.

\bibitemdeclare{inproceedings}{CBS91}
\bibitem{CBS91}
\bibinfo{author}{K.V.S.~\surnamestart Prasad\surnameend}
  (\bibinfo{year}{1991}): \emph{\bibinfo{title}{A Calculus of Broadcasting
  Systems}}.
\newblock In \bibinfo{editor}{S.~\surnamestart Abramsky\surnameend} \&
  \bibinfo{editor}{T.S.E.~\surnamestart Maibaum\surnameend}, editors: {\sl
  \bibinfo{booktitle}{{\rm TAPSOFT'91: Proceedings of the International Joint
  Conference on} Theory and Practice of Software Development, {\rm Volume 1:}
  Colloquium on Trees in Algebra and Programming {\rm (CAAP'91)}}}, {\sl
  \bibinfo{series}{\rm LNCS}} \bibinfo{volume}{493},
  \bibinfo{publisher}{Springer}, pp. \bibinfo{pages}{338--358},
  \doi{10.1007/3-540-53982-4\_19}.

\bibitemdeclare{book}{Rei13}
\bibitem{Rei13}
\bibinfo{author}{W.~\surnamestart Reisig\surnameend} (\bibinfo{year}{2013}):
  \emph{\bibinfo{title}{Understanding Petri Nets --- Modeling Techniques,
  Analysis Methods, Case Studies}}.
\newblock \bibinfo{publisher}{Springer}, \doi{10.1007/978-3-642-33278-4}.

\bibitemdeclare{article}{Sh85-lb}
\bibitem{Sh85-lb}
\bibinfo{author}{M.W.~\surnamestart Shields\surnameend} (\bibinfo{year}{1985}):
  \emph{\bibinfo{title}{Concurrent machines}}.
\newblock {\sl \bibinfo{journal}{The Computer Journal}}
  \bibinfo{volume}{28}(\bibinfo{number}{5}), pp.
  \bibinfo{pages}{449--465,\newline\weg}, \doi{10.1093/comjnl/28.5.449}.

\bibitemdeclare{article}{Sta89}
\bibitem{Sta89}
\bibinfo{author}{E.W.~\surnamestart Stark\surnameend} (\bibinfo{year}{1989}):
  \emph{\bibinfo{title}{Concurrent transition systems}}.
\newblock {\sl \bibinfo{journal}{Theoretical Computer Science}}
  \bibinfo{volume}{64}(\bibinfo{number}{3}), pp. \bibinfo{pages}{221--269},
  \doi{10.1016/0304-3975(89)90050-9}.

\bibitemdeclare{inproceedings}{Wi87a}
\bibitem{Wi87a}
\bibinfo{author}{G.~\surnamestart Winskel\surnameend} (\bibinfo{year}{1987}):
  \emph{\bibinfo{title}{Event structures}}.
\newblock In \bibinfo{editor}{W.~\surnamestart Brauer\surnameend},
  \bibinfo{editor}{W.~\surnamestart Reisig\surnameend} \&
  \bibinfo{editor}{G.~\surnamestart Rozenberg\surnameend}, editors: {\sl
  \bibinfo{booktitle}{Petri Nets: Applications and Relationships to Other
  Models of Concurrency, Advances in Petri Nets 1986, Part II; Proceedings of
  an Advanced Course, {\rm Bad Honnef, September 1986}}}, {\sl
  \bibinfo{series}{\rm LNCS}} \bibinfo{volume}{255},
  \bibinfo{publisher}{Springer}, pp. \bibinfo{pages}{325--392},
  \doi{10.1007/3-540-17906-2\_31}.

\bibitemdeclare{incollection}{WN95}
\bibitem{WN95}
\bibinfo{author}{G.~\surnamestart Winskel\surnameend} \&
  \bibinfo{author}{M.~\surnamestart Nielsen\surnameend} (\bibinfo{year}{1995}):
  \emph{\bibinfo{title}{Models for Concurrency}}.
\newblock In \bibinfo{editor}{S.~\surnamestart Abramsky\surnameend},
  \bibinfo{editor}{D.~\surnamestart Gabbay\surnameend} \&
  \bibinfo{editor}{T.~\surnamestart Maibaum\surnameend}, editors: {\sl
  \bibinfo{booktitle}{Handbook of Logic in Computer Science}},
  chapter~\bibinfo{chapter}{1}, \bibinfo{volume}{4: Semantic Modelling},
  \bibinfo{publisher}{Oxford University Press}, pp. \bibinfo{pages}{1--148}.

\end{thebibliography}
\end{document}